\documentclass[11pt]{article}

\usepackage{anysize}
\marginsize{2.3cm}{2.3cm}{2.1cm}{2.1cm}

\usepackage{lineno,hyperref}
\usepackage{multido}
\usepackage[T1]{fontenc}
\usepackage[utf8]{inputenc}
\DeclareUnicodeCharacter{02BC}{'}

\usepackage{comment}
\usepackage{amssymb}
\usepackage{amsmath}
\usepackage{amsthm}
\usepackage{xcolor}
\usepackage{graphicx} 
\usepackage{rotating}
\usepackage{algorithm}
\usepackage{multirow}
\usepackage{multicol}
\usepackage{algorithmic}
\usepackage{graphicx}
\usepackage{enumerate}
\usepackage{url}
\usepackage{hhline}
\usepackage{subcaption}
\usepackage{mathrsfs}
\usepackage{lipsum}       
\usepackage{changepage}   
\usepackage{lscape}
\usepackage{empheq}
\usepackage{cleveref}
\usepackage{caption}
\usepackage{subcaption}
\usepackage[round]{natbib}
\usepackage[english]{babel}
\usepackage{longtable}

\newcommand{\rset}{\mathbb{R}}

\newcommand{\nset}{\mathbb{N}}

\newcommand{\T}{^\top}

\newcommand\SN[1]{{\color{orange} #1}}

\newcommand\Rew[1]{{\color{black} #1}}
\newcommand\RRew[1]{{\color{black} #1}}

\modulolinenumbers[5]

\theoremstyle{plain}

\newtheorem{proposition}{Proposition}
\newtheorem{corollary}{Corollary}
\newtheorem{lemma}{Lemma}

\newtheorem{assumption}{Assumption}
\newtheorem{example}{Example}

\makeatletter

\allowdisplaybreaks

\begin{document}

\title{Learning Paths to Multi-Sector Equilibrium:\\
Belief Dynamics Under Uncertain Returns to Scale}

\author{Stefano Nasini\thanks{IESEG School of Management, Univ. Lille, CNRS, UMR 9221 - LEM- Lille Economie Management, F-59000 Lille, France, \texttt{s.nasini@ieseg.fr}} 
\and 
Rabia Nessah \thanks{IESEG School of Management, Univ. Lille, CNRS, UMR 9221 - LEM- Lille Economie Management, F-59000 Lille, France, \texttt{r.nessah@ieseg.fr}}
\and
Bertrand Wigniolle\thanks{Paris School of Economics and Universit\'e Paris 1 Pantheon-Sorbonne, 48 Boulevard Jourdan, 75014 Paris, France, \texttt{bertrand.wigniolle@univ-paris1.fr}}}

\maketitle

\begin{abstract}
This paper explores the dynamics of learning in a multi-sector general equilibrium model where firms operate under incomplete information about their production returns to scale. Firms iteratively update their beliefs using maximum a-posteriori estimation, derived from observed production outcomes, to refine their knowledge of their returns to scale. The implications of these learning dynamics for market equilibrium and the conditions under which firms can effectively learn their true returns to scale are the key objects of this study. Our results shed light on how idiosyncratic shocks influence the learning process and demonstrate that input decisions encode all pertinent information for belief updates. Additionally, we show that a long-memory  (path-dependent) learning which keeps track of all past estimations ends up having a worse performance than a short-memory (path-independent) approach. 
\end{abstract}

\medskip


\bigskip


\newpage


\section{Introduction}\label{sec:introduction}


An important challenge in economics is to understand how agents interpret their environment and learn the true state of nature, with learning viewed as the process by which they infer the suitability of their actions to that environment \citep{dasgupta1988learning, kaelbling1996reinforcement, penczynski2017nature, mossel2020social}. As an example, consumers and firms possess incomplete knowledge of some of their payoff-relevant parameters, yet they are sometimes able to observe the outcomes of their actions and learn from them. \cite{10.2307/2295952} conjectured about the importance of such experience-based learning for improving economic outcomes and emphasized the need to incorporate learning mechanisms into the main corpus of economic theory. 

Since then, different modeling approaches have integrated knowledge acquisition into economic theory and studied the implications of learning in multiple economic contexts. A first domain is concerned with expectations formed by agents about some relevant variables when the rational expectation assumption is relaxed. \cite{evans2009learning} provide a general survey of this approach by studying adaptive learning schemes based on econometric techniques, and consider the convergence to a rational expectation equilibrium. A second branch of literature considers settings where imperfect information is due to private information held by each agent. Observation of other agents' behaviors and market equilibrium prices may reveal information, but also lead to herding behaviors \citep{banerjee1992simple, bikhchandani1992theory, duan2009informational}.

In this paper, we take a different approach from these two strands by studying a dynamic multi-sector general equilibrium economy in which firms learn a production-technology parameter governing their own returns to scale.
They form rational expectations in the sense that they use all available information to anticipate relevant variables, such as the market price law or the equilibrium wage. But, as they have a limited knowledge of their technology, they can only make imperfect estimates of these variables. At each period, they use the information conveyed by their additional observation to re-estimate their scale parameter. In this regard, our model is related to recent studies of learning processes in which agents try to estimate some profit-relevant parameters. \cite{collin2016parameter}, \cite{bottazzi2023market}, and \cite{wu2024reinforcement} are concerned with financial markets, and consider the learning process of macroeconomic aggregates (such as GDP growth rate and its variance), having a direct bearing on portfolio decisions. \cite{babiak2024parameter} are closer to our model, as they consider a production economy where agents are unaware of the technology shock parameters. The authors propose a learning rule and explore the impact of learning on macroeconomic and financial aggregates. 

Our focus is on the learning of the returns to scale parameters, with a specific interest in the convergence of the learning dynamics. In our model, firms obtain information exclusively from a private channel, which consists of observing their own production. This is a distinction with respect to \cite{amador2012learning}, who study agents' learning from the actions of others through two channels: a public channel, such as equilibrium market prices, and a private channel, such as local interactions. Our dynamic multi-sector general equilibrium framework is inspired by the work of \cite{horvath2000sectoral}, \cite{gabaix2011granular}, \cite{acemoglu2012network}, and \cite{acemoglu2017microeconomic} and considers a collection of representative firms operating in a noisy environment. The importance of adopting a multi-sector perspective in the analysis of learning paths is motivated by the fact that idiosyncratic variations can delay the law of large numbers in aggregate production and translate into noisy estimates of the returns to scale parameter \citep{horvath2000sectoral, gabaix2011granular, coulson1995sources}.\footnote{For decades, macroeconomic models aggregated idiosyncratic variations, hindering the development of multi-sector dynamic general equilibrium models with endogenous learning \citep{leontief1970environmental}. In economies with $n$ sectors, the diversification argument by \cite{lucas1977understanding} suggested that microeconomic shocks average out, preventing detailed exploration of inter-sectoral networks.Aggregate fluctuations in an economy with independent shocks at its $n$ sectors would have a magnitude proportional to $1/\sqrt{n}$, a negligible effect at high levels of disaggregation.} 

The choice of focusing on uncertainty about returns to scale pivots on a long line of contributions, which has addressed its statistical estimation by parametric and non-parametric approaches \citep{golany1997estimating, basu1997returns, banker2004returns, Ackerberg2015}. The topic is significant not only for its links to externalities and public goods \citep{starrett1977measuring}, but also for its impact on input misallocation \citep{gong2016role}, as the incorrect knowledge of returns to scale can induce a firm to miss the optimal selection of production inputs, with major consequences at the macro level \citep{baqaee2020productivity}. 

Unlike \cite{collin2016parameter}, \cite{bottazzi2023market}, \cite{wu2024reinforcement} and \cite{babiak2024parameter}, our study is statistically grounded in the Bayesian framework for optimal information collection and processing of \cite{el1993bayesian}, who describe the optimization problem encountered by a single infinitely-living Bayesian agent, which alternates between a decision-making stage and a learning stage. Analogously, in the decision-making stage, firms make input decisions based on their current beliefs, which are encoded in their maximum a-posteriori (MAP) estimation of the returns to scale. In the learning stage, firms observe their realized production, and compute the Bayesian posterior of their returns to scale to be used in the next period.\footnote{This approach mirrors a reinforcement learning procedure where the mismatch between a firm's observed production and the one expected for different returns to scale induces a performance metric (likelihood function) quantifying the correctness of a firm's belief at each period. This is in line with the recent work of \cite{beggs2022reference}, in which agents evaluate outcomes relative to reference points (point-wise estimation), as well as with contributions assessing the statistical properties of the maximizer of the Bayesian posterior distribution \citep{bassett2019maximum}. This MAP estimation constitutes the main distinction between a fully Bayesian procedure and our proposed reinforcement learning procedure.} We investigate two alternative Bayesian updating methods: a path-dependent (non-Markovian) approach, where firms keep track of the full history of past productions when updating their belief about the returns to scale; a path-independent (Markovian) approach, where firms only record one previous period.

Our analysis produces four key theoretical insights. Firstly, all relevant information in the learning dynamics (i.e., the information required to produce an update in firms' belief) is encoded in the input decisions. In other words, conditioned on input decisions, the learning path is invariant with respect to prices and household consumption. Secondly, we derive  in closed-form conditions that ensure firms can accurately learn their true returns to scale. These conditions are primarily linked to the magnitude of firms' idiosyncratic shocks and the uncertainty (informativeness) of their prior knowledge.  Thirdly, while the dynamics of the path-independent MAP estimator to infer the returns to scale parameter is in line with the case of path-dependent learning, its convergence is faster and more accurate in the context of path-independent learning. Therefore, keeping track of all past estimations ends up having a worse performance than a short-memory approach in which all past information is collected into a single prior distribution. We provide a mathematically grounded explanation of this phenomenon, based on the theory of Bayesian filtering. Finally, with a view to generalize our results beyond the case of uncertainty about returns to scale, we provide a dedicated extension of our general equilibrium framework to the high-dimensional learning, highlighting the consistency between two learning dynamics.

In summary, our results assess the reliability of market mechanisms in cases where market participants are unaware of profit-relevant parameters, in agreement with \cite{Vives1996589}. The convergence of the MAP estimator suggests that firms' ability to discover their own returns to scale from the dynamic observation of realized noisy productions allows input misallocation to be asymptotically circumvented by a learning mechanism when partial information is dynamically revealed. 

The rest of this paper is organized as follows. Section \ref{sec:model} introduces the proposed modeling approach for a dynamic multi-sector input-output economy. Section \ref{sec:equilibrium} characterizes its equilibrium conditions. Section \ref{sec:Learning_exogenous_l} establishes our main results on incomplete information and learning dynamics. Sections \ref{sec:learning_dynamics_long} and \ref{sec:learning_dynamics_short} study, respectively, the path-dependent (non-Markovian) and the path-independent (Markovian) learning methods. Section \ref{subsec:BayesianFiltering} interprets the results from the Bayesian filtering perspective. Section \ref{sec:Hight_dimension} extends our approach to the case of multi-input production with unknown input-output elasticities. Section \ref{sec:conclusions} concludes the paper. All proofs are collected in Appendix~\ref{Section:appendix1}, and the main notation is summarized in Appendix~\ref{app:notation}.


\section{The model}\label{sec:model}


We introduce a stylized multi-sector economy, where production depends on labor and on a productivity shock.\footnote{This may be viewed as a simplified version of \cite{horvath2000sectoral}, \cite{gabaix2011granular}, \cite{acemoglu2012network}, and \cite{acemoglu2017microeconomic}, with the specific distinction that inputs of production are limited to labor. \RRew{Despite this simplification, our model remains grounded in the general equilibrium tradition: prices and allocations are jointly determined across interlinked sectors through market-clearing conditions and an optimizing representative agent.}} The economy consists of a collection $\mathcal{N}$ of distinct sectors (with $|\mathcal{N}| = n$), each producing a different good, and an additional sector $0$, whose labor demand $l_0(t)$ is treated as an exogenous control variable in our analysis.

In addition to the sectors, the economy comprises a representative consumer, taking decisions about a preferred consumption plan over the $n + 1$ goods, with the aim of maximizing its lifetime utility. We let $x_i(t)$, $c_i(t)$, $l_i(t)$ and $p_i(t)$ be the corresponding production, consumption, labor demand and price of the $i$-th sector at the $t$-th production period. The profit of sector $i$ in this economy is denoted as $\pi_i(t)$. The unitary salary is denoted as $w(t)$. 

\paragraph{The consumer.}

The representative consumer is endowed with a quasi-linear utility function, and maximizes its lifetime utility subject to a period-by-period budget constraint:
\begin{equation}\label{eq:consumer_model}
    \max \overset{\infty}{\underset{t=1}{\sum}}\rho^{t}\left(  c_{0}%
(t)+\overset{n}{\underset{i=1}{\sum}}\frac{c_{i}^{\alpha}(t)}{\alpha}\right)\quad \mbox{subj. to } \quad p_{0}(t) c_{0}(t)+\overset{n}{\underset{i=1}{\sum}}p_{i}(t)c_{i}(t)\leq E(t),
\end{equation}
where $\alpha$ is an exogenous parameter, with $0<\alpha<1$, and $c_{0}(t)$ is the numeraire good, so that its price is $p_{0}(t)=1$. The representative consumer's income (and spending) at time $t$ is $E(t) = w(t)\delta(t)  + \pi_0(t) + \pi(t)$, which is composed of the labor income $w(t)\delta(t)$ (with  $\delta(t)$ being the active population at time $t$, or labor supply, and $w(t)$ being the unitary salary), the sum of profits $\pi(t) = \sum_{i \in \mathcal{N}} \pi_i(t)$ from the $n$ sectors, and the profit of sector $0$, denoted $\pi_0(t)$.

Notably, because the model includes neither capital nor a storage technology, the consumer spends the entire income in every period. Hence, the intertemporal utility maximization \eqref{eq:consumer_model} boils down to a sequence of static allocation problems, whose solution is $p_{i}(t)/p_{0}(t) = c_{i}^{\alpha-1}(t)$.

As discussed in the next section, while this simplifying assumption allows isolating the learning paths of firms as the only intertemporal mechanism relevant to our analysis, we provide in Appendix \ref{Section:appendix_C} an extension of the baseline consumer model \eqref{eq:consumer_model}, where we demonstrate that the learning paths studied in this paper remain structurally identical. 

\paragraph{The producers.} In this stylized economy, labor is the only input of production. Sector $0$ has a linear technology $x_{0}(t)=l_{0}(t)$, where $l_{0}(t)$ is labor employed in sector $0$ and $x_{0}(t)$ is its output. This can be interpreted as the production of a public good. Throughout the paper, $l_0(t)$ is treated as an exogenous process (e.g., capturing factors outside the model).

Each sector in $\mathcal{N}$ is controlled by a price-taker representative firm that operates so as to maximize its expected profit.\footnote{\RRew{The simplifying notion of a representative firm is often used by macroeconomic models \citep{hartley2002representative}. By this is meant a hypothetical firm whose production is equal to the aggregate production of the sector as a whole, and whose inputs are equal to the aggregate inputs of the sector as a whole. At the empirical level, industry taxonomies (such as the Global Industry Classification Standard, or the Standard Industrial Classification) can be used to group companies, based on similar production processes and products.}}
The production technology of sector $i \in \mathcal{N}$ is:
\begin{equation}\label{eq:ExpectedProduction0}
    x_{i}(t)=\eta_{i}(t)l_{i}^{\zeta_{i}^{(\ast)}}(t) = \eta_{i}(t) \mu_{i}(t),
\end{equation}
where the returns to scale parameter $\zeta_{i}^{(\ast)}$ is not known by the firm  (with $\zeta_{i}^{(\ast)} \in [\underline{\zeta}, \overline{\zeta}]$, where $\underline{\zeta} > 0$ and $\overline{\zeta} < 1$ are given constants), and the random factors $\eta_{i}(t) \sim \log N(m_i,\sigma_i)$ are independent idiosyncratic productivity shocks, with $\log N(m_i,\sigma_i )$ denoting a log-Gaussian distribution with parameters $m_i$ and $\sigma_i$.\footnote{
The log-normal specification for $\eta_i(t)$ ensures strictly positive productivity levels and is widely used in the literature due to its empirical relevance and tractability (see, e.g., \cite{gabaix2011granular}). Importantly, our equilibrium results do not rely critically on the exact log-normal form. What matters is that $\eta_i(t)$ admits a finite moment of order $\alpha$, so that expected profits are well-defined and the labor demand function remains tractable. Alternative continuous, positive distributions with similar moment properties (e.g., Gamma or Weibull) could be accommodated within our framework. However, since the learning paths studied in Section \ref{sec:Learning_exogenous_l} are defined parametrically, each alternative specification of $\eta_i(t)$ requires a dedicated assessment of the belief dynamics and its convergence.} The scalar quantity $m_i$ relates to the $i$-th firm productivity constant, whose expectation is $q_{i} = \mathbb{E}[\eta_{i}(t)] = \exp(m_i + \sigma_i^2/2)$. 

\paragraph{Market clearing.} In accordance with \cite{acemoglu2012network,acemoglu2017microeconomic}, the market clears each period such that $x_i(t) = c_i(t)$. Further, the total supply of labor is fixed and the representative household is endowed with total labor, which is supplied inelastically:
\begin{equation}\label{eq:labor-supply}
\sum_{i = 1}^{n} l_i(t) = \delta(t) - l_0(t).
\end{equation}

\RRew{Three features of this stylized dynamic multi-sector model merit emphasis. First, firms' profit-maximizing behavior depends not only on their expectations about output prices but also on the prevailing wage, which governs the cost of labor inputs. Although firms take the wage 
$w(t)$ as given when making decisions, its equilibrium value is determined endogenously through the interaction of firms' labor demand with the household labor supply. The full derivation of equilibrium wages, prices, and labor allocations is provided in Section \ref{sec:equilibrium},\footnote{See Proposition~\ref{prop:equilibrium} and Equation~\eqref{eq:equilibrium_w} for closed-form expressions of equilibrium wages and prices.} where these market interactions are explicitly formalized. Second, because firm productivity is subject to idiosyncratic shocks and no risk-sharing assets are available, the model exhibits incomplete markets both in the short and the long run. In the short run, firms cannot insure against current productivity shocks. In the long run, there is no mechanism for intertemporal risk smoothing or wealth accumulation. While extensions with hedging instruments or exchange-based contracts are in principle feasible, our framework deliberately excludes them in order to focus sharply on firms' learning dynamics under uncertainty. Third, following \cite{acemoglu2012network,acemoglu2017microeconomic}, we adopt the exogenous labor supply assumption \eqref{eq:labor-supply}. Appendix \ref{Section:appendix_D} shows, however, that endogenizing the labor supply does not materially affect the characterization of learning paths. Overall, this stylized multi-sector model retains the essential structure of a general equilibrium economy, in which all key market variables are jointly determined by equilibrium conditions. At the same time, it allows isolating the belief dynamics as the only intertemporal mechanism relevant for our purpose.}

\paragraph{Decision-estimation dynamics.} Our model distinguishes between the true returns to scale parameter $\zeta_{i}^{(\ast)}$ and the belief that firm $i$ has about it. We denote with $\zeta_{i}(t)$ firm $i$'s guess (estimation) of $\zeta_{i}^{(\ast)}$ at time $t$, so that the expected production conditioned on the $i$-th firm belief becomes
\begin{equation}\label{eq:ExpectedProduction}
    \mu_{i}(t) = q_i l_{i}^{\zeta_{i}(t)}.
\end{equation}

\RRew{For each firm $i$, we denote by $\mathcal{I}_i(t)$ the information available at date $t$ and we assume that firms form rational expectations conditionally to the information they have.\footnote{All profit expectations, predictive prices, and any expectation entering the firm's decision problem should be read as conditional on $\mathcal{I}_i(t)$. If a primitive (e.g.\ $\eta_i(t)$) is independent of $\mathcal{I}_i(t)$, the conditional expectation coincides with the unconditional one.}} More precisely, they know the law of the technology shock, and they have a guess $\zeta_{i}(t)$ on the parameter $\zeta_{i}^{(\ast)}$ of their production technology. To determine their optimal labor demand, they must anticipate the law followed by equilibrium prices. Hence, consistently with the rational expectations hypothesis, they must make their expectations by solving the economic equilibrium, based on their estimation $\zeta_{i}(t)$ of $\zeta_{i}^{(\ast)}$. Precisely, we introduce the distinction between true and \emph{believed} production (denoted as $x_i(t)$ and $\Bar{x}_i(t)$, respectively), as well as the distinction between the true and \emph{believed} price (denoted as $p_i(t)$ and $\Bar{p}_i(t)$, respectively). This distinction is driven by the fact that firm $i$'s expected profit depends on $\mathbb{E}[\Bar{p}_i(t) \eta_{i}(t) \;\big|\;\mathcal{I}_i(t)]$, where $\Bar{p}_i(t)$ satisfies the \emph{believed} market clearing conditions $\Bar{x}_i(t) = c_i(t)$, and $\mathbb{E}\big[ \cdot \big| \mathcal{I}_i(t)\big]$ is the conditional expectation given $\mathcal{I}_i(t)$. In fact, it would be incoherent to assume that firms are able to anticipate $\mathbb{E}[p_i(t) \eta_{i}(t)]$, as this would imply knowing $\zeta_{i}^{(\ast)}$. 

According to the framework of information collection and processing introduced by \cite{el1993bayesian} (and further explored by \cite{cogley2008anticipated}), each firm $i$ takes input decisions and updates its belief about $\zeta_{i}^{(\ast)}$, alternating between a decision-making stage and a learning stage:

\begin{itemize}
    \item[(i)] The labor input $l_{i}(t)$ is selected ex-ante by profit maximization, without knowing the value of $\eta_{i}(t)$ and with a (potentially incorrect) belief about $\zeta_{i}^{(\ast)}$:\RRew{    
    \begin{equation}\label{eq:firm_problem}
    l_i(t) \;=\; \underset{l_i(t)\ge 0}{\arg\max}\;
    \mathbb{E}\big[\,\bar p_i(t)\,\eta_i(t)\,l_i(t)^{\hat{\zeta}_i(t)} - w(t)\,l_i(t)\;\big|\;\mathcal{I}_i(t)\big],
    \end{equation}
    where $\hat{\zeta}_i(t) = \max\left\{ \underline{\zeta}, \min\left\{ \overline{\zeta}, \zeta_i(t) \right\} \right\}$. We emphasize that this is a static optimization problem, conditional on the current information set $\mathcal{I}_i(t)$, and no intertemporal nesting of expectations is involved in the firm's decision problem \citep{blackwell1962merging}. }
    \item[(ii)] The true production $x_{i}(t)$ is observed ex-post, and an estimate $\zeta_{i}(t+1)$ is generated and used at production period $t+1$.  
\end{itemize}

It is worth noting that, based on \eqref{eq:firm_problem}, if it happens that the estimated value $\zeta_{i}(t+1)$ in step (ii) is greater or equal than the constant $\overline{\zeta} < 1$, then the firm determines its optimal labor demand with $\overline{\zeta}$ in step (i). In other words, the firm dismisses the empirical estimation of $\zeta_{i}(t)$ and adopts a rule of thumb. Analogously, if it happens that the estimated value $\zeta_{i}(t)$ in step (ii) is less or equal than the constant $\underline{\zeta} > 0$, the firm determines its optimal labor demand with $\underline{\zeta}$ in step (i). Therefore, $\hat{\zeta}_i(t) = \max\left\{ \underline{\zeta}, \min\left\{ \overline{\zeta}, \zeta_i(t) \right\} \right\}$. 

In this alternating dynamics, the firm treats $\zeta_i(t)$ as a fixed parameter when choosing $l_i(t)$ and only updates it after observing $x_i(t)$ in each period. This yields a series of conditional one-period problems (in the form of \eqref{eq:firm_problem}) which follow the temporary equilibrium tradition of \cite{grandmont1977temporary}, where firms solve within-period optimization problems taking as given a forecast of payoff-relevant market objects (in particular, the price vector and the wage) conditional on their information set $\mathcal{I}_i(t)$, and markets clear once shocks are realized. Because these forecasts are formed under the firm's current technology belief $\zeta_i(t)$, they need not coincide with the realized equilibrium outcomes generated by the true parameter $\zeta_i^{(*)}$. The next period starts from realized outcomes (here, the privately observed production $x_i(t)$), which trigger a revision of beliefs and hence of the forecasted equilibrium objects. In this sense, the decision--estimation dynamics can be read as a temporary equilibrium sequence equipped with an explicit expectation-updating (learning) mechanism. This frames our approach as an adaptive filter with endogenous feedback, in which the state evolution is shaped directly by the agents' decisions (see \cite{marcet1989convergence} for a related adaptive filter learning mechanisms in economic theory). In this vein, a state-space formulation of our model consists of a hidden state $\zeta_i(t)$ (representing the firm's belief about its returns to scale) which is recursively estimated using a noisy observation $x_i(t)$ that is itself influenced by the firm's actions (i.e., the labor input $l_i(t)$). As studied later in the paper (see Section \ref{sec:Learning_exogenous_l}), we consider two deterministic update estimation functions for the learning stage (ii) of our alternating dynamics: $\zeta_{i}(t+1) = \phi(x_i(1), \ldots, x_i(t))$, which we call path-dependent, and $\zeta_{i}(t+1) = \phi(\zeta_i(t), x_i(t))$, which we call path-independent. Both require the characterization of the dependency between $x_i(t)$ and $\zeta_i(t)$ through market equilibrium conditions.


\section{The equilibrium}\label{sec:equilibrium}


In accordance with step (i) of the decision--estimation dynamics, for each $i \in \mathcal{N}$, the optimal demand for labor is obtained by solving \eqref{eq:firm_problem}. The profit of sector zero is $l_{0}(t)(1-w(t))$, so that if $w(t) > 1$, sector zero incurs a loss, and if $w(t) < 1$, sector zero earns a gain. Once $l_{0}(t)$ is set, prices are determined by the market clearing conditions on the observed production ex-post.

\RRew{\begin{proposition}\label{prop:equilibrium} We have the following equilibrium labor and prices, for each sector $i \in \mathcal{N}$:
\begin{align}\label{eq:labor_eq}
        l_i(t) &= \displaystyle \left( \dfrac{\mathbb{E}\left[ \eta_{i}(t)^{\alpha}  \right] \hat{\zeta}_{i}(t) }{ w(t)  } \right)^{\dfrac{1}{1-\hat{\zeta}_i(t)\alpha}} ,   \\[0.1cm]\label{eq:price_eq}
        p_i(t) & = \left( \eta_{i}(t) l_i(t)^{\zeta_{i}^{(\ast)}} \right)^{\alpha-1}.
\end{align}
\end{proposition}
It is worth highlighting the joint dependence of both equilibrium prices $p_i(t)$ and production quantities $x_i(t)$ on the same random variable $\eta_i(t)$. This arises naturally in the general equilibrium framework: while $x_i(t)$ is directly subject to the productivity shock through the production function, prices adjust endogenously to clear the market given this realized stochastic output. In particular, using the equilibrium condition $x_i(t) = c_i(t)$ and the demand equation $p_i(t) = c_i(t)^{\alpha - 1}$, it follows that prices inherit the randomness of output via $\eta_i(t)$, as equilibrium enforces a coupling between prices and production. This said, \eqref{eq:labor_eq} and \eqref{eq:price_eq} are deduced coherently with the principle that firms' expected profits are sensitive to $\hat{\zeta}_i(t)$ not only through production but also through prices, as discussed in the previous section. In fact, when the returns to scale parameter $\zeta_i(t)$ is unknown (and statistically estimated) by firms, a fundamental distinction emerges between the true equilibrium prices and the believed equilibrium prices. This distinction is adopted in the proof of Proposition \ref{prop:equilibrium} to characterize firm's decision in a way which is consistent with our belief updating design.  Following this line of reasoning, the unitary salary is determined by the exogenous labor supply equation \eqref{eq:labor-supply} and the firms' demand for labor \eqref{eq:labor_eq}:
\begin{equation}\label{eq:equilibrium_w}
    \delta(t) - l_0(t) = \displaystyle \sum_{i=1}^n \left( \dfrac{\mathbb{E}\left[ \eta_{i}(t)^{\alpha}  \right] \hat{\zeta}_{i}(t) }{ w(t)  } \right)^{\dfrac{1}{1-\hat{\zeta}_i(t)\alpha}} := L_t(w(t)).
\end{equation}}

For existence and uniqueness of the equilibrium solution, we establish the following lemma.

\begin{lemma}\label{lemma:L}
$L_t(w)$ is continuous and decreasing in $\rset_+$, with $\underset{w \rightarrow + \infty}{\lim} L_t(w) = 0$ and $\underset{w \rightarrow 0}{\lim} L_t(w) = +\infty$. 
\end{lemma}

As a consequence of Lemma \ref{lemma:L}, for any exogenous labor demand $l_0(t) \in (0, \delta(t)]$ in sector $0$, there exists a unique wage $w^*(t)$ such that $\delta(t) - l_0(t) = L_t(w^*(t))$. \RRew{Hence, Equation \eqref{eq:equilibrium_w} admits a unique wage level at which aggregate labor demand equals aggregate labor supply (which entails the uniqueness of commodity prices through \eqref{eq:price_eq}). Notably, once $l_0(t)$ and the equilibrium wage $w(t)$ are pinned down, sector $0$'s profit is uniquely determined as $\pi_0(t)=l_0(t)\,(1-w(t))$.

It is worth emphasizing that, while the existence of unique equilibrium prices ensures internal consistency of the framework, the actual dynamics of equilibrium prices are strongly shaped by firms' beliefs about their returns to scale. Appendix \ref{Section:appendix_E} provides a numerical illustration of this point, showing that uncertainty about returns to scale can significantly affect equilibrium price dynamics, thereby underscoring the importance of analyzing how firms learn about these parameters.}


\section{The learning paths}\label{sec:Learning_exogenous_l}


We now examine the learning paths of the sequence $\{\zeta_i(t)\}_t$ generated in step (ii) of the decision-estimation dynamics. To lighten the mathematical presentation, the following notation is used throughout this section:
\begin{equation}\label{eq:reparam}
    \varepsilon_{i}(t) = \log(\eta_{i}(t)), \qquad z_{i}(t)  = \log \left( l_i(t)\right) \quad \mbox{ and } \quad s_i(\ell) = \varepsilon_{i}(\ell) +  \zeta_i^{(*)}z_i(\ell). 
\end{equation}

Therein, firms take decisions based on $\zeta_i(t)$, while assessing the appropriateness of their beliefs to the realized production when constructing a probability distribution for the next period belief $\zeta_i(t+1)$. This is done by invoking the Bayesian rule and a collection of learning assumptions.

\begin{assumption}\label{ass:1} At each production period, firms use the maximum a-posteriori as a point-wise estimator of the unknown $v^{(*)}_i$.
\end{assumption}

The adoption of a point-wise estimation method when alternatingly combining a decision stage and a learning stage, relies on the anticipated utility approach introduced by \cite{kreps1998anticipated} and  \cite{cogley2008anticipated}. In this approach, firms treat the profit-relevant parameter $\zeta_i(t)$ as a random quantity when they learn but as a constant when they take decisions. \RRew{More precisely, at each period the firm bases its input choice on the current point forecast of $\zeta_i(t)$, which is then updated recursively after observing $x_i(t)$.} Conversely, a full Bayesian procedure would regard it as a random quantity both for learning and for decision-making.\footnote{As noted by \cite{cogley2008anticipated}, a full Bayesian procedure is mathematically intractable for most economic problems, so that they studied the goodness of the anticipated utility for different economic models.} 

An important point to raise is that firms do not observe ex-post the random shock $\eta$. When period $t$ becomes the present, firm $i$ knows the amount of labor $l_{i}(t)$ that it had employed and observes the amount of production $x_{i}(t)$. As the firm does not know the true value of the scale parameter $\zeta_{i}^{(\ast)}$ but only has an estimation $\zeta_{i}(t)$ of this parameter, it cannot infer from $x_{i}(t)$ and $l_{i}(t)$ the value of $\eta_{i}(t)$, but can use these information to formulate a new estimation of $\zeta_{i}^{(\ast)}$ that will be used in period $t+1$ to determine the labor demand and the production level.

As studied in Subsection \ref{subsec:BayesianFiltering}, our approach can be analyzed from the Bayesian filter perspective, where the latent state variable $\zeta_i(t)$ is recursively updated through a deterministic mapping $\phi$, using noisy observations of output $x_i(1), \ldots, x_i(t)$ that are themselves shaped by the firm's own actions. In the following sections, we focus on two specifications of $\phi$: a path-dependent specification, where firms keep track of the full history of past productions when building their probability distribution of $\zeta_i(t)$ (as studied in Section \ref{sec:learning_dynamics_long}) and a path-independent specification, where firms only record one previous period when building $\zeta_i(t)$ (as studied in Section \ref{sec:learning_dynamics_short}).

\section{Path-dependent (non-Markovian) learning}\label{sec:learning_dynamics_long}

A path-dependent (non-Markovian) learning method is a point-wise estimation approach constructed by enforcing the following assumptions on the information set and the initial knowledge.

\begin{assumption}\label{ass:2}
At period $t$, the firm's $i$ information set is given by the series of past realized productions up to the current period: $\mathcal{I}_i(t-1) = \{x_i(t-1), x_i(t-2), \ldots, x_i(1)\}$, for $t \geq 2$.
\end{assumption}

\begin{assumption}\label{ass:3}
Firms have an initial knowledge quantified as a zero-truncated Gaussian with parameters $\zeta^{(0)}_{i}$ and $\tau_i$ and density function $\pi_{0,i}(\zeta_i)$.\footnote{It is worth noting that $\zeta_i(t)$ represents a MAP estimator, while $\pi_{t,i}(\zeta_i)$ denotes its probability density function.}
\end{assumption}

\normalsize

In line with the production technology \eqref{eq:ExpectedProduction0}, let $\log x_i(h) = m_i + \log \mu_i(h) + \varepsilon_i(h)$, where $\varepsilon_i(h)\stackrel{\text{i.i.d.}}{\sim}\mathcal{N}(0,\sigma_i^2)$, for each $h=1,\ldots,t-1$, so that, conditionally on the input choices (hence on $\{\mu_i(h)\}_{h\le t-1}$), the log-output is Gaussian with mean $m_i+\log\mu_i(h)$ and variance $\sigma_i^2$.\footnote{Modeling log-output as Gaussian is standard in empirical production and implies a log-normal distribution for output levels. Writing the likelihood for $\{x_i(h)\}$ (rather than for $\{\log x_i(h)\}$) adds the Jacobian factor $\prod_{h}1/x_i(h)$, which depends only on realized data and not on $\zeta_i$, so it does not affect maximization over $\zeta_i$.}
With respect to Assumption \ref{ass:2}, the realized production encodes all required information (in the sense of statistical sufficiency) for the likelihood characterization. In fact, the likelihood function of the returns to scale parameter $\zeta_i(t)$ is the joint density of $\{\log x_i(h)\}_{h=1}^{t-1}$, which is induced by the idiosyncratic log-shocks $\{\varepsilon_i(h)\}_{h=1}^{t-1}$:
\begin{equation}\label{eq:likelihood}
\mathscr{L}(\zeta_i(t);\,\mathcal{I}_i(t-1))
~=~ \prod_{h=1}^{t-1}\frac{1}{\sigma_i\sqrt{2\pi}}
\exp\!\left(
-\frac{1}{2\sigma_i^{2}}
\Big(\log x_i(h) - m_i - \log \mu_i(h)\Big)^{2}
\right),\qquad t\ge2,
\end{equation}
\noindent where $\mu_i(h)$ has been defined in \eqref{eq:ExpectedProduction0} and represents the deterministic part of the production function (as a function of inputs and $\zeta_i(t)$). Hence, $\mathscr{L}(\zeta_i(t);\,\mathcal{I}_i(t-1))$ quantifies the mismatch between the log-production of firm $i$ observed up to period $t$ and the log-production implied by a candidate value of $\zeta_i(t)$. This can be seen as a way of assessing the goodness of firm $i$'s knowledge of its production function, or as an environment reward to the appropriateness of its input allocation.\footnote{In the context of the empirical production literature, likelihood-based estimation of returns to scale from realized (log-)production dates back to \cite{marschak1944random}. Contextually, the described learning approach compounds a reinforcement learning mechanism, defined upon (i) a state space of the system, (ii) a set of actions to be taken by firms, (iii) a policy, and (iv) a performance metric.
\begin{itemize}
    \item[(i)] The state space of the system is the set of feasible returns to scale parameters.
    \item[(ii)] The action taken by firm $i$ at period $t$ is $l_i(t)$.
    \item[(iii)] The policy defining the decision criteria at each state visited by the system is the profit maximization.
    \item[(iv)] To judge the performance of the policy, $\mathscr{L}(\zeta_i(t);\,\mathcal{I}_i(t-1))$ constitutes a metric for the difference between the target log-production and the realized one.
\end{itemize}}

The Bayesian posterior distribution of $\zeta_i(t)$ at the $t$-th period is
\begin{equation}\label{eq:T-posterior}
\pi_{t,i}(\zeta_i) ~=~
\begin{cases}
\displaystyle \frac{\mathscr{L}(\zeta_i;\,\mathcal{I}_i(t-1))\,\pi_{0,i}(\zeta_i)}
{\int \mathscr{L}(\zeta_i;\,\mathcal{I}_i(t-1))\,\pi_{0,i}(\zeta_i)\,d\zeta_i}
& \qquad \text{if $t>1$,}\\[0.3cm]
\pi_{0,i}(\zeta_i) & \qquad \text{if $t=1$.}
\end{cases}
\end{equation}

\begin{proposition}[MAP dynamics (path-dependent)]\label{prop:v_learning_long_mem} 
For each $t \geq 1$, the MAP estimator is
\begin{equation}\label{eq:v_learning_n1_long_mem}
\zeta_i(t+1) = \phi(x_i(1), \ldots, x_i(t)) =\left\{
\begin{array}{ll}
      \left( \dfrac{\zeta^{(0)}_i  + \gamma_i \sum_{\ell=1}^{t} z_i(\ell) s_i(\ell)}{\left(1 + \gamma_i \sum_{\ell=1}^{t} z_i(\ell)^2\right)} \right)^{+}, &  \Rew{\mbox{ if } l_i(\ell) > 0 \mbox{ for } ~ \ell \leq t,} \\[0.5cm]
      \Rew{\zeta_i(t)} & \mbox{ otherwise, }
\end{array} \right.
\end{equation}
where $\gamma_i = (\tau_i/\sigma_i)^2$ and the notation $(z)^+$ refers to $\max(0, z)$ for any $z \in \rset$.
\end{proposition}

The first notable insight from \eqref{eq:v_learning_n1_long_mem} is that $\zeta_i(t+1)$ is updated solely based on the $i$-th labor input decisions (which appear in $z_i(\ell)$). Consequently, $\zeta_i(t+1)$ is conditionally independent, once input decisions are made, and the $i$-th learning path is invariant with respect to any other market information. The second key insight is that, if there exists a previous period $\ell \leq t$ for which $l_i(\ell) = 0$, the maximization of \eqref{eq:T-posterior} admits multiple solutions (all $\zeta_i(t+1) \in \rset_+$ are optimal for \eqref{eq:T-posterior}). Therefore, a firm cannot update its knowledge. This case is, however, prevented by the equilibrium labor demand  \eqref{eq:labor_eq}, as firms are assumed to dismiss the empirical estimation of $\zeta_{i}(t)$ when $\zeta_i(t+1) < \underline{\zeta}$ (and $l_i(\ell) > 0$ as long as $\hat{\zeta}_i(t) = \underline{\zeta} > 0$).\footnote{A technical note about the probability of $\zeta_i(t+1) < \underline{\zeta}$ (namely the probability that a firm dismisses the empirical estimation of $\zeta_{i}(t)$ and adopt a rule of thumb) is provided in Appendix \ref{Section:appendix2}, demonstrating that this eventuality becomes unlikely for small values of $\tau_i$ and $\sigma_i$.} This said, a peculiar form of learning emerges in the limit case (i.e., when there exists $\ell \in \{1, \ldots, t\}$ such that $l_i(\ell)$ converges to zero). In fact, this induces the learning path to converge almost surely to $\zeta_i^{(*)}$.

\begin{proposition}\label{prop:conv_pointwise} For each $\ell \in \{1, \ldots, t\}$, the MAP estimator \eqref{eq:v_learning_n1_long_mem} converges (pointwise) to  $v^{(*)}_i$: 
$$
\mathbb{P}\left(  \underset{l_i(\ell)\rightarrow 0}\lim \zeta_i(t+1)  = v^{(*)}_i \right) = 1.
$$  
\end{proposition}

This form of learning indicates that the disappearance of a sector (driven by a drop in its labor demand) is a \emph{price of learning} the true returns to scale. Under the assumption that firms dismiss the MAP estimation of $\zeta_{i}(t)$ when $\zeta_{i}(t) < \underline{\zeta}$ (and determine their optimal labor demand using the rule of thumb $\hat{\zeta}_i(t) = \max\left\{ \underline{\zeta}, \min\left\{ \overline{\zeta}, \zeta_i(t) \right\} \right\}$), this peculiar form of learning is associated with situations in which equilibrium labor input becomes arbitrarily small for some sectors---for instance when the unitary wage $w(t)$ becomes arbitrarily large, so that labor demand vanishes.

At this stage, it is worth reminding that the MAP estimator \eqref{eq:v_learning_n1_long_mem} is a random variable induced by the noisy production $x_i(t)$ through the idiosyncratic shocks $\eta_i(t)$. A fundamental step in the analysis of the learning dynamics is therefore the evaluation of the expectation (Proposition~\ref{prop:expect_v_long_memory} below) and the mode (Proposition~\ref{prop:mode_v_long_memory} below) of $\zeta_i(t)$, as well as their limiting behavior as $t$ grows large.

\begin{proposition}[Expectation and variance of the MAP estimator (path-dependent)]\label{prop:expect_v_long_memory} For each $t \geq 1$, let us define
\begin{equation}\label{eq:v_bar}
    \overline{v}_i(L_i(t))  = \dfrac{ \zeta^{(0)}_i  + \gamma_i \zeta_i^{(*)} \tilde{z}^{(2)}_{i}(t)}{1 + \gamma_i \tilde{z}^{(2)}_{i}(t)} \quad \mbox{ and } \quad
    \overline{\varphi}_i(L_i(t))  = \dfrac{\gamma_i  \sigma_i \tilde{z}^{(1)}_{i}(t)}{1 + \gamma_i \tilde{z}^{(2)}_{i}(t)},
\end{equation}
where $\tilde{z}^{(1)}_{i}(t) = \sum_{\ell=1}^{t-1} z_i(\ell)$,  $\tilde{z}^{(2)}_{i}(t) = \sum_{\ell=1}^{t-1} z_i(\ell)^2$, and $L_i(t) = \{l_i(\ell)\}_{\ell=1}^{t}$. Based on Proposition \ref{prop:v_learning_long_mem}, conditioned on the $i$-th input decisions up to period $t$, we have
\begin{eqnarray}
\mathbb{E}[\zeta_i(t+1)  ~|~ L_i(t)] & = & \overline{v}_i(L_i(t)) \tilde{F}_{i}(t) +  |\overline{\varphi}_i(L_i(t))| \tilde{f}_{i}(t), \\[0.3cm] 
\mathbb{E}[\zeta_i(t+1)^2 ~|~ L_i(t)] & = & \left(\overline{v}_i(L_i(t))^2 +  |\overline{\varphi}_i(L_i(t))|^2\right) \tilde{F}_{i}(t) + \overline{v}_i(L_i(t)) |\overline{\varphi}_i(L_i(t))|\tilde{f}_{i}(t),
\end{eqnarray}
where 
$$
\begin{array}{lll}
    \tilde{F}_{i}(t) & \RRew{:=} ~ G_i\Big(\tilde{z}^{(1)}_{i}(t),\tilde{z}^{(2)}_{i}(t)\Big) &\RRew{:=}  ~ 1-F\left(-\left|\dfrac{\overline{v}_i(L_i(t))}{\overline{\varphi}_i(L_i(t))}\right|\right),\\[0.4cm]
    \tilde{f}_{i}(t) & \RRew{:=}  ~ g_i(\tilde{z}^{(1)}_{i}(t),\tilde{z}^{(2)}_{i}(t)) &\RRew{:=}  ~ f\left(-\left|\dfrac{\overline{v}_i(L_i(t))}{\overline{\varphi}_i(L_i(t))}\right|\right). 
\end{array}
$$
\end{proposition}

\begin{proposition}[Mode of the MAP estimator (path-dependent)]\label{prop:mode_v_long_memory} \Rew{Let $\mathbb{M}[\zeta_i(t+1) ~|~ L_i(t) ]$ be the mode of the MAP estimator $\zeta_i(t+1)$, conditioned on the $i$-th input decisions up to period $t$. Based on Proposition \ref{prop:v_learning_long_mem}, we have}
$$
\mathbb{M}[\zeta_i(t+1) ~|~ L_i(t) ] = 
\begin{cases}
  \overline{v}_i(L_i(t)) & \quad \mbox{if } \displaystyle \dfrac{1}{|\overline{\varphi}_i(L_i(t))|\sqrt{ 2 \pi}} \geq  F\left(-\left|\dfrac{\overline{v}_i(L_i(t))}{\overline{\varphi}_i(L_i(t))}\right|\right),\\[0.3cm]
  0 & \quad \mbox{otherwise.}
\end{cases}
$$
\end{proposition}

The findings of Propositions \ref{prop:expect_v_long_memory} and \ref{prop:mode_v_long_memory} are notable for their broad applicability, as the derived expectation, variance, and mode of the MAP estimator remain invariant even when market participants make sub-optimal input decisions guided by heuristic algorithms or rules of thumb \citep{blonski1999rational}. However, to analyze the limiting properties and convergence of the expectation (Proposition \ref{prop:expect_v_long_memory}) and mode (Proposition \ref{prop:mode_v_long_memory}), the next subsection assumes that the economy operates in equilibrium, adhering to conditions \eqref{eq:labor_eq}--\eqref{eq:equilibrium_w}.

\subsection{Limit properties and convergence under the path-dependent learning}\label{subsec:limit_long_mem}

\RRew{Propositions~\ref{prop:expect_v_long_memory} and~\ref{prop:mode_v_long_memory}, studied in the previous subsection, allow for a closed-form analysis of the convergence of the MAP estimator under path-dependent learning. A central contribution of this subsection is to establish its asymptotic properties. In particular, we show that unbiased convergence is the exception rather than the rule, with sustained labor demand emerging as the most economically plausible scenario. 

Throughout this subsection we work with the conditional distribution of $\zeta_i(t)\mid L_i(t-1)$, so that expectations are always taken with respect to the firm's information set at date $t-1$. This formulation connects our analysis to the general theory of belief convergence under expanding information sets \citep{blackwell1962merging}. To obtain tractable limit properties, we introduce a mean-field approximation, where all occurrences of $\eta_i(1), \ldots, \eta_i(t-1)$ in $l_i(t)$ are replaced by their expectation $q_i$, so that the past trajectory of labor inputs until period $t-1$ (and thus $\tilde{z}^{(1)}_{i}(t-1)$ and $\tilde{z}^{(2)}_{i}(t-1)$) can be treated as deterministic quantities. This approximation substitutes the actual (potentially fluctuating) influence of $\tilde{z}^{(1)}_{i}(t-1)$ and $\tilde{z}^{(2)}_{i}(t-1)$ on $\zeta_i(t)$ with an average influence yielding an accurate description for sectors in which the magnitude of idiosyncratic shocks is sufficiently small.}

\begin{proposition}[Convergence of the expectation (path-dependent)]\label{prop:long_memory_learning} 
The dynamics of the path-dependent learning method has five cases:

\begin{itemize}
    \item[(1)] If for each $i\in \cal{N}$, the dynamics of input decisions satisfies $\underset{t \rightarrow +\infty }\lim |\tilde{z}_{i}^{(1)}(t)| = \underset{t \rightarrow +\infty }\lim \tilde{z}_{i}^{(2)}(t)=+\infty$ and $\underset{t \rightarrow +\infty }\lim\dfrac{\zeta^{(0)}_i  + \gamma_i \zeta_i^{(*)} \tilde{z}^{(2)}_{i}(t)}{\gamma_i  \sigma_i \tilde{z}^{(1)}_{i}(t)} = L < +\infty$, then we have
    $$
    \underset{t \rightarrow +\infty }\lim \mathbb{E}[\zeta_i(t) ~|~ L_i(t-1)] =  \zeta_i^{(*)}(1-F(-L)) +
    \sigma_i f(-L) \left(\underset{t \rightarrow +\infty }\lim\dfrac{|\tilde{z}_{i}^{(1)}(t)|}{\tilde{z}_{i}^{(2)}(t)}\right).
    $$
    \item[(2)] If for each $i\in \cal{N}$, the dynamics of input decisions satisfies $\underset{t \rightarrow +\infty }\lim \tilde{z}_{i}^{(1)}(t) =L_1$ and $\underset{t \rightarrow +\infty }\lim \tilde{z}_{i}^{(2)}(t) = L_2$, then we have
    $$\underset{t \rightarrow +\infty }\lim \mathbb{E}[\zeta_i(t)~|~ L_i(t-1)] =
    \left\{
    \begin{array}{ll}
      \dfrac{\Big( \zeta^{(0)}_i  + \gamma_i \zeta_i^{(*)}L_2\Big)G_i(L_1,L_2)+\sigma_i \gamma_i L_1 g_i(L_1,L_2)}{1+\gamma_i L_2} & \text{if } L_1\neq 0, \\
      \dfrac{\zeta^{(0)}_i  + \gamma_i \zeta_i^{(*)}L_2}{1+\gamma_i L_2} & \text{if } L_1 = 0.
    \end{array}
    \right.
    $$
    \item[(3)] If for each $i\in \cal{N}$, the dynamics of input decisions satisfies $\underset{t \rightarrow +\infty }\lim \tilde{z}_{i}^{(1)}(t) =L_1< +\infty$ and $\underset{t \rightarrow +\infty }\lim \tilde{z}_{i}^{(2)}(t) =\infty$, then we have
    $$
    \underset{t \rightarrow +\infty }\lim \mathbb{E}[\zeta_i(t)~|~ L_i(t-1)] = \zeta_i^{(*)}. 
    $$
    \item[(4)] If for each $i\in \cal{N}$, the dynamics of input decisions satisfies $\underset{t \rightarrow +\infty }\lim \tilde{z}_{i}^{(1)}(t)=\infty$ and $\underset{t \rightarrow +\infty }\lim \tilde{z}_{i}^{(2)}(t) =L_2< +\infty$, then we have
    $$
    \underset{t \rightarrow +\infty }\lim \mathbb{E}[\zeta_i(t)~|~ L_i(t-1)] = + \infty.
    $$
    \item[(5)] If for each $i\in \cal{N}$, the dynamics of input decisions satisfies 
    $$
    \underset{t \rightarrow +\infty }\lim |\tilde{z}_{i}^{(1)}(t)| = \underset{t \rightarrow +\infty }\lim \tilde{z}_{i}^{(2)}(t)=\underset{t \rightarrow +\infty }\lim\dfrac{\zeta^{(0)}_i  + \gamma_i \zeta_i^{(*)} \tilde{z}^{(2)}_{i}(t)}{\gamma_i \sigma_i \tilde{z}^{(1)}_{i}(t)} =  +\infty,
    $$ 
    then we have
    $$
    \underset{t \rightarrow +\infty }\lim \mathbb{E}[\zeta_i(t)~|~ L_i(t-1)] = \zeta_i^{(*)}. 
    $$
\end{itemize}

\end{proposition}

This proposition reveals that under the path-dependent assumption, there are two possible dynamics of input decisions in which the expected belief of the returns to scale converges to $\zeta_i^{(*)}$, which are case (3) and case (5).  While case (5) corresponds to sectors with non-decaying labor demands (with $l_i(t) > 1$ when $t$ grows large), case (3)  corresponds to sectors whose demand for labor input follows very specific decay patterns. In fact, case (3) requires the sequence of labor inputs $\{l_i(t)\}_{t}$ to decrease sufficiently quickly (as to make $\{\tilde{z}_{i}^{(1)}(t)\}_{t}$ converge) but not too quickly (as to make $\{\tilde{z}_{i}^{(2)}(t)\}_{t}$ diverge). This means that the unbiasedness of $\zeta_i(t)$ in case (3) can only happen for a sector whose demand for labor input follows a very specific decay pattern, which renders case (3) of limited economic significance.\footnote{As an example, a possible converging sequence of labor input is of the form $l_i(t) = e^{1/t^p}$, for $p \in (1/2, 1]$. Since the equilibrium labor input \eqref{eq:labor_eq} implies that $l_i(t)$ is a continuous random variable (induced by $\hat{\zeta}_i(t)/\eta_i(t)^{1-\alpha}$) with support $\rset_+$, the probability of $\{l_i(t)\}_{t}$ to follow the decay pattern of a converging sequence can be numerically computed.} 

Another very specific scenario under which the expected belief of the returns to scale converges to $\zeta_i^{(*)}$ is provided in case (2), but only under the unique condition $\lim_{t \to +\infty} \tilde{z}_{i}^{(1)}(t) = 0$, $\lim_{t \to +\infty} \tilde{z}_{i}^{(2)}(t) > 0$ and $\sigma \to 0$. 
Hence, this very specific scenario has limited economic significance. 

Still focusing on case (2), it is worth noticing that if $\lim_{t \to +\infty} \tilde{z}_{i}^{(1)}(t) \neq 0$, then
\begin{equation}\label{eq:v_expected_long_mem}
\underset{t \rightarrow +\infty }\lim \mathbb{E}[\zeta_i(t)~|~ L_i(t-1)] =
      \dfrac{\Big( \zeta^{(0)}_i  + \gamma_i \zeta_i^{(*)}L_2\Big)G_i(L_1,L_2)+\sigma_i \gamma_i L_1 g_i(L_1,L_2)}{1+\gamma_i L_2} . 
\end{equation}
Since $G_i(L_1,L_2) \in [0, 1]$ and  $g_i(L_1,L_2) \in [0, 1/\sqrt{2\pi}]$, when the magnitude of idiosyncratic productivity shocks approaches zero, we have:
$$
\underset{\sigma_i \rightarrow 0 }{\lim}~ \underset{t \rightarrow +\infty }{\lim} G_i(L_1,L_2) = 1,
$$
which implies
\begin{equation}\label{eq:v_expected_sigma_0}
    \underset{\sigma_i \rightarrow 0 }{\lim}  ~ \underset{t \rightarrow +\infty }{\lim} \mathbb{E}[\zeta_i(t)~|~ L_i(t-1)] = \zeta^{(0)}_i + \zeta_i^{(*)}.
\end{equation}

Therefore, for a sector whose demand for labor quickly decays, its limit belief has a positive bias when the magnitude of productivity shocks approaches zero. Also this case has a limited economic significance, as it concerns disappearing sectors with disappearing productivity shocks.

An economically plausible case is that of sectors with non-quickly-decaying labor demands, corresponding to case (1) in Proposition \ref{prop:long_memory_learning}. Since when $|\tilde{z}_{i}^{(1)}(t)|$ and $\tilde{z}_{i}^{(2)}(t)$ diverge, one has the limit of $|\tilde{z}_{i}^{(1)}(t)|/\tilde{z}_{i}^{(2)}(t) \leq 1$, then the limit MAP estimator of these sectors systematically underestimates the true returns to scale when the magnitude of productivity shocks approaches zero:
$$
\underset{\sigma_i \rightarrow 0 }{\lim}  ~ \underset{t \rightarrow +\infty }{\lim} \mathbb{E}[\zeta_i(t)~|~ L_i(t-1)] = \zeta_i^{(*)}(1- F(-L)).
$$

We recall that in our model, each firm $i$ knows that her scale parameter $\zeta_{i}^{(\ast)}$ belongs to the interval $\lbrack\underline{\zeta},\overline{\zeta}]$ with $\underline{\zeta}<1$ and $\overline{\zeta}>0$. Therefore, by construction firms dismiss the empirical estimation of $\zeta_{i}(t)$ and adopt a rule of thumb when the MAP estimator is outside of this adopted range. This implies that, even in the case of successive bad shocks that lead to $\zeta_{i}(t)=0$ during one or several periods, there is no question of the survival of the firm as it appears in \cite{blume2006if} for traders. In case of an estimation such that $\zeta_{i}(t)=0$, the firm behaves according to $\underline{\zeta}$.

Overall, Proposition \ref{prop:long_memory_learning} highlights a fundamental insight: in a path-dependent framework, unbiased convergence is exceptional rather than the rule, and the economically plausible scenario is that of sectors with sustained labor demand (case (5)). This is true even when the magnitude of the idiosyncratic shocks approaches zero.
The next proposition affirms that this case is possible under the equilibrium conditions of our multi-sector model.

\begin{proposition}[Demographic expansion (path-dependent)]\label{prop:demography_long_mem} By the equilibrium conditions \eqref{eq:labor_eq}--\eqref{eq:equilibrium_w}, any monotonically increasing demographic path $\{\delta(t)\}_t$ (with $\delta(t) \geq \delta(t-1) + 1$) satisfies the condition of case (5) of Proposition \ref{prop:long_memory_learning}. 
\end{proposition}

While the idea of unlimited demographic expansion constitutes a peculiar (and only theoretically plausible) case of unbiasedness, a less restrictive conceptualization of learning is provided in the following proposition, showing that cases (3) and (4) of Proposition \ref{prop:long_memory_learning}  (for which the expected belief converges to $\zeta_i^{(*)}$ or diverges to $+\infty$, respectively) are both associated to a correct mode of the MAP estimator. 
 
\begin{proposition}[Convergence of the mode (path-dependent)]\label{prop:convergence_mode_long_memory} Let $\zeta_i(0) \geq 0$ and $\zeta_i^{(*)} > 0$, for all $i \in \mathcal{N}$. Let us define the ordered set 
\begin{equation}\label{eq:Psi_mode}
    \overline{\Psi}_i = \mathbb{N}/\left\{t \in \nset ~:~ |\overline{\varphi}_i(L_i(t))|\sqrt{ 2 \pi} F\left(-\left|\dfrac{\overline{v}_i(L_i(t))}{\overline{\varphi}_i(L_i(t))}\right|\right)> 1 \right\},
\end{equation}
\Rew{and index its elements by $\varrho(t)$, for $t \in \mathbb{N}$. There exists a sub-sequence $\{\overline{v}_i(\varrho(t))\}_{t}$ such that} 
$$
    \underset{t \rightarrow +\infty }\lim \mathbb{M}[\zeta_i(t+1) ~|~ L_i(t) ]=\left\{
    \begin{array}{ll}
      0 & \text{if } \overline{\Psi}_i\text{ is a finite set,}\\
      \zeta_i^{(*)} & \text{if } \overline{\Psi}_i\text{ is an infinite set and } \underset{t \rightarrow +\infty }\lim z_{i}(t) =+\infty, \\
      \dfrac{\zeta_i^{0}+\gamma_i \zeta_i^{(*)}L}{1+\gamma_i L}  & \text{if } \overline{\Psi}_i\text{ is an infinite set and }\underset{t \rightarrow +\infty }\lim z_{i}(t)=L <+\infty.
    \end{array}\right.
$$
\end{proposition}

The ordered set $\overline{\Psi}_i$ is an auxiliary representation of the collection of time periods for which $\mathbb{M}[\zeta_i(t) ~|~ L_i(t) ] > 0$. This allows selecting a sub-sequence to study the convergence of the original sequence, as detailed in Appendix A. Under the limit conditions of $\tau_i \rightarrow 0$, $\overline{\Psi}_i$ is an infinite set. Note that, since $\sum_{i=1}^n l_i(t) = \delta(t) - l_0(t)$, we have that if $\underset{t \rightarrow +\infty }\lim \delta(t) < +\infty$ then $\underset{t \rightarrow +\infty }\lim z_{i}(t) <+\infty$. Therefore, the only economically feasible case is
$$
\underset{t \rightarrow +\infty }\lim \mathbb{M}[\zeta_i(t+1) ~|~ L_i(t)] =\dfrac{\zeta_i^{0}+\gamma_i \zeta_i^{(*)}L}{1+\gamma_i L}  .
$$
Under the limit conditions of $\sigma_i \rightarrow 0$ or $\tau_i \rightarrow + \infty$, we have
\begin{equation}\label{eq:mode_long_mem_sigma_0}
\underset{\sigma_i \rightarrow 0 }{\lim}~ \underset{t \rightarrow +\infty }\lim \mathbb{M}[\zeta_i(t+1) ~|~ L_i(t) ] = \zeta_i^{(*)} \quad \mbox{ and } \quad \underset{\tau_i \rightarrow + \infty }{\lim}~ \underset{t \rightarrow +\infty }\lim \mathbb{M}[\zeta_i(t+1) ~|~ L_i(t) ] = \zeta_i^{(*)}.
\end{equation}

This result suggests that, within a path-dependent framework, smaller idiosyncratic production shocks enable a more precise estimation of the true returns to scale. Combined with Proposition \ref{prop:demography_long_mem}, this finding underscores the robustness of market mechanisms, extending their applicability beyond well-established deterministic cases \citep{Vives1996589}. This is further illustrated by the following stylized numerical example.

\begin{example}\label{ex:1} We consider the illustrative case $n=1$, $\zeta_{i}^{(\ast)}= 0.5$, $\alpha = 0.5$, and $m_i = 0$ \RRew{(also used in Appendix \ref{Section:appendix_D})}. We set $\sigma_i = 0.1$. Figures \ref{fig:Example2_1} and \ref{fig:Example2_2} show the dynamics of $v_1(t)$ along a time horizon $t \in \{1, \ldots, 500\}$, for different values of $\tau_i$ and $\zeta_{i}^{(0)}$. Visibly, while the convergence of $\zeta_{i}(t)$ to $\zeta_{i}^{(\ast)}$ results under all scenarios in Figures \ref{fig:Example2_1} and \ref{fig:Example2_2}, the speed of convergence strongly depends on $\tau_i$.  Smaller values of $\tau_i$ (Figure \ref{fig:Example2_2}) entail that firms do not have strong prior beliefs about the returns to scale (high prior uncertainty).

\begin{figure}[H]
        \centering
        \begin{subfigure}[b]{0.4\textwidth}     
        \includegraphics[scale=0.68]{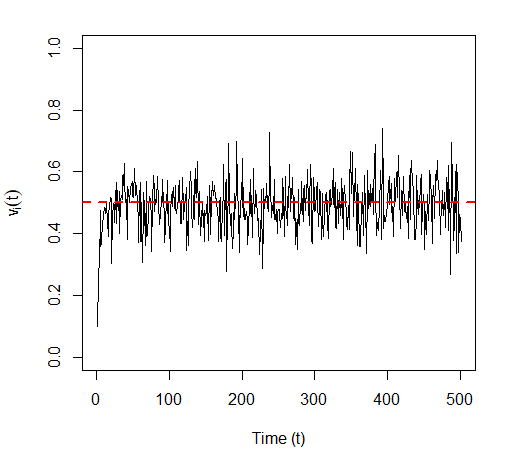}
        \caption{\footnotesize Setting $\zeta_{i}^{(0)} = 0.1$.}
        \end{subfigure}
        \begin{subfigure}[b]{0.4\textwidth}
        \includegraphics[scale=0.68]{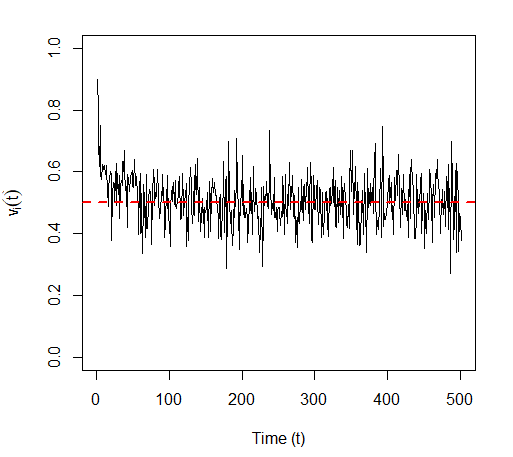}
        \caption{Setting $\zeta_{i}^{(0)} = 0.9$.}
        \end{subfigure} 
        \caption{\footnotesize Learning path $\{\zeta_i(t)\}_t$ with $\tau=0.1$.  \label{fig:Example2_1} }
\end{figure}

\begin{figure}[H]
        \centering
        \begin{subfigure}[b]{0.4\textwidth}     
        \includegraphics[scale=0.68]{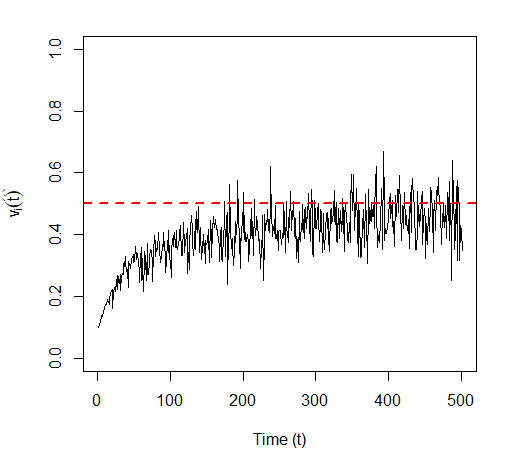}
        \caption{\footnotesize Setting $\zeta_{i}^{(0)} = 0.1$.}
        \end{subfigure}
        \begin{subfigure}[b]{0.4\textwidth}
        \includegraphics[scale=0.68]{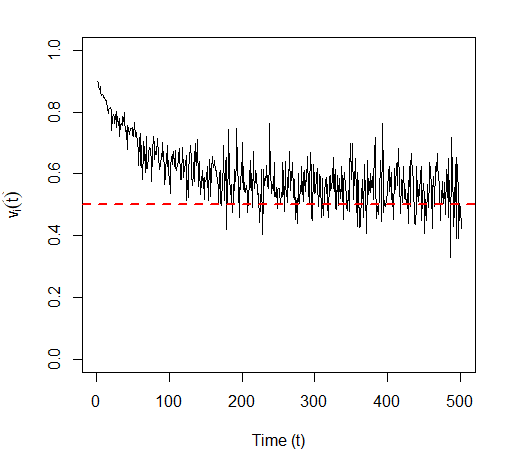}
        \caption{Setting $\zeta_{i}^{(0)} = 0.1$.}
        \end{subfigure} 
        \caption{\footnotesize Learning path $\{\zeta_i(t)\}_t$ with $\tau=0.01$.  \label{fig:Example2_2} }
\end{figure}    

Table \ref{tab:Ex2} reports the difference between $\mathbb{E}[\zeta_i(500)]$ and $\zeta_{i}^{(\ast)}$, for different values of $\sigma_i$ and $\tau_i$.  In line with Figures \ref{fig:Example2_1} and \ref{fig:Example2_2}, this supports our theory about the convergence of the learning path (as established in Propositions \ref{prop:long_memory_learning} and \eqref{prop:convergence_mode_long_memory}), establishing that  within the path-dependent framework smaller idiosyncratic production shocks and high prior uncertainty enable to accurately approach the true returns to scale, which is fully consistent with the limit modes in \eqref{eq:mode_long_mem_sigma_0}. 

\begin{table}[H]
\centering
\scalebox{0.98}{   
    \begin{tabular}{|cc|cc|cc|}
        \hline
        $\tau_1$ & $\sigma_1$ & \multicolumn{2}{c|}{$\zeta_{i}^{(\ast)}=0.4$} & \multicolumn{2}{c|}{$\zeta_{i}^{(\ast)}=0.6$} \\ \cline{3-6}
        & & $\zeta_{i}^{(0)} = 0.1$ & $\zeta_{i}^{(0)} = 0.9$ & $\zeta_{i}^{(0)} = 0.1$ & $\zeta_{i}^{(0)} = 0.9$ \\ 
        \hline
        $0.01$ & $0.01$ & 0.0036042 & 0.0053890 & 0.0023717 & 0.0013001 \\
        $0.01$ & $0.05$ & 0.0904860 & 0.0995278 & 0.0749700 & 0.0260160 \\
        $0.01$ & $0.10$ & 0.2788238 & 0.2530222 & 0.2803820 & 0.0728513 \\\hline
        $0.05$ & $0.01$ & 0.0001426 & 0.0002194 & 0.0000921 & 0.0000526 \\
        $0.05$ & $0.05$ & 0.0036030 & 0.0053845 & 0.0023696 & 0.0012986 \\
        $0.05$ & $0.10$ & 0.0146437 & 0.0204167 & 0.0100418 & 0.0049896 \\\hline
        $0.10$ & $0.01$ & 0.0000356 & 0.0000548 & 0.0000230 & 0.0000131 \\
        $0.10$ & $0.05$ & 0.0008935 & 0.0013649 & 0.0005800 & 0.0003278 \\
        $0.10$ & $0.10$ & 0.0035991 & 0.0053703 & 0.0023629 & 0.0012937
 \\\hline
    \end{tabular}}
    \caption{\footnotesize Relative absolute differences between the $\zeta_i(t)$ and $\zeta_i^{(\star)}$. \label{tab:Ex2}}
\end{table}
\end{example}

In summary, the results presented in this section, particularly Propositions \ref{prop:v_learning_long_mem}--\ref{prop:convergence_mode_long_memory}, underscore the effectiveness of the market mechanism in facilitating knowledge acquisition under conditions of noisy production \citep{Vives1996589}. These findings should also be viewed in the context of the important contributions by \cite{collin2016parameter} and \cite{babiak2024parameter}, which highlight how parameter learning amplifies the effects of shocks and offers a valuable framework for better replicating certain empirical observations in asset pricing. The former study focuses on a simplified financial investment model with Epstein-Zin preferences, while the latter employs a standard real business-cycle model. In both scenarios, parameter learning enhances the representation of financial variables by intensifying the impact of shocks through the learning process.

Example \ref{ex:2} below examines how learning impacts the labor share of income (the portion of GDP allocated to workers versus capital) and demonstrates that uncertainty in returns to scale and learning increase its variance. A rich body of literature explores labor share fluctuations, attributing them to factors such as capital- or skill-based technological change, globalization (via the Stolper-Samuelson theorem), financialization boosting capital returns, and shifts in labor's bargaining power \citep{stockhammer2013have}. Our analysis introduces a new perspective about the labor share fluctuations: firms' learning processes about their technology. 

\begin{example}\label{ex:2} Following the same parametrization as in Example \ref{ex:1}, we examine the labor share of income (over a time series with $T = 100$ observations) for three different economies with one sector. These economies differs only with respect to the initial knowledge of the returns to scale parameter: $\zeta^{(0)} = 0.1$, $\zeta^{(0)} = 0.5$ and $\zeta^{(0)} = 0.9$.  We note that 
$$
GDP(t) = \underbrace{w^*(t) \delta(t)}_{\mbox{labor income}} +  ~~ \underbrace{l_0 + p_1^*(t) x_i(t) - w^*(t) \delta(t) }_{\mbox{profit}}.
$$
Figures \ref{fig:Example3} shows the histogram of $w^*(t) \delta(t)/GDP(t)$, over a time series with $100$ observations. The pink histograms depict the labor share of income when firms has exact knowledge of $v^{(\star)}$, as opposed to cases (blue histograms) when $\zeta_i(t)$ is dynamically updated based on the MAP estimation \eqref{eq:v_learning_n1_long_mem}. 

\begin{figure}[H]
        \centering
        \begin{subfigure}[b]{0.33\textwidth}     \includegraphics[scale=0.6]{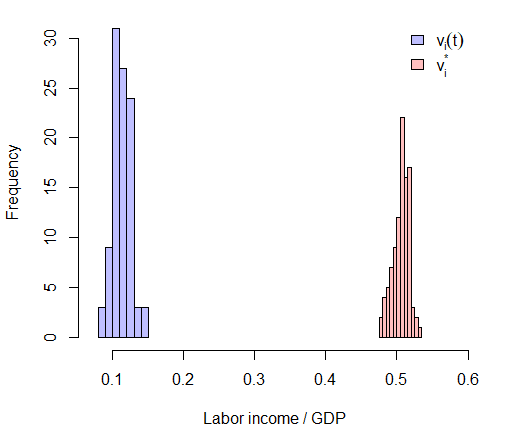}
        \caption{\footnotesize Setting $\zeta^{(0)} = 0.1$. \label{fig:Example3_1}}
        \end{subfigure}
        \begin{subfigure}[b]{0.32\textwidth}
        \includegraphics[scale=0.6]{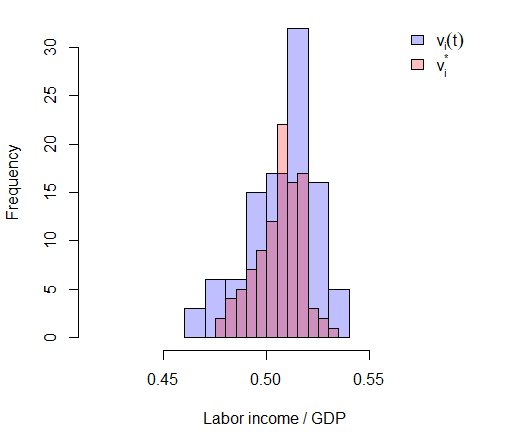}
        \caption{\footnotesize Setting $\zeta^{(0)} = 0.5$. \label{fig:Example3_2}}
        \end{subfigure}
        \begin{subfigure}[b]{0.33\textwidth}
        \includegraphics[scale=0.6]{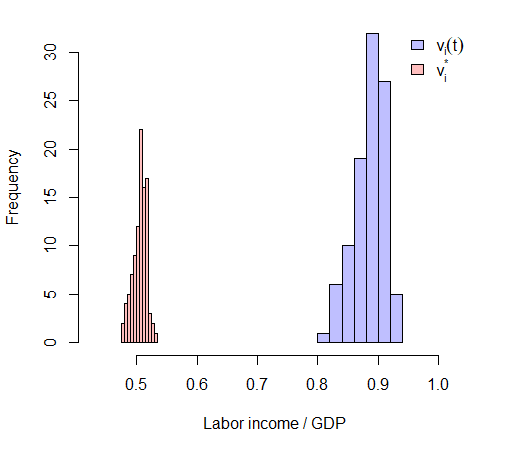}
        \caption{\footnotesize Setting $\zeta^{(0)} = 0.9$. \label{fig:Example3_3}}
        \end{subfigure}
        \caption{\footnotesize Histograms of $w(t) \delta(t)/GDP$ for three different economies. \label{fig:Example3} }
\end{figure}

Table \ref{tab:Ex3} summarizes the distributions in Figure \ref{fig:Example3} in terms of means and standard deviations.

\begin{table}[H]
\centering
\scalebox{0.98}{   
    \begin{tabular}{|l|lll|c|}
        \cline{2-5}
        \multicolumn{1}{c|}{} & $\zeta_{i}^{(0)}=0.1$ & $\zeta_{i}^{(0)}=0.5$ & $\zeta_{i}^{(0)}=0.9$ & $\zeta_{i}^{(\star)}=0.5$\\ \hline
        Mean & 0.1127378 & 0.5066034 & 0.8853479 & 0.5059747 \\
        Standard deviation   & 0.0122680 & 0.0171584 & 0.0258762 & 0.0114483 \\\hline
    \end{tabular}}
    \caption{\footnotesize Means and standard deviations for the histograms of $w(t) \delta(t)/GDP$ in Figure \ref{fig:Example3}. \label{tab:Ex3}}
\end{table}

We observe that uncertainty on the returns to scale and learning translate into larger standard deviations of the labor share of income, compared to the case in which firms have exact knowledge of $v^{(\star)}$. However, the average value of $w^*(t) \delta(t)/GDP$ is largely driven by over-estimation versus under-estimation of $\zeta_{i}^{(\star)}$. Hence, when firms' estimates align closely with true values, learning adds a source of fluctuation. However, when estimates deviate significantly, as with the advent of new technologies, learning drives a prolonged adjustment of factor shares toward their long-run equilibrium.

\end{example}


\section{Path-independent (Markovian) learning}\label{sec:learning_dynamics_short}


This section examines a path-independent learning method in which firms assimilate all past information into a new prior distribution but rely solely on the current production data as their information set. This framework encapsulates the concept of a short-memory process, where the system re-initializes around a new status-quo belief at each step of the decision-estimation dynamics. Formally, the point-wise estimation resulting from this path-independent learning method is derived by imposing the following assumptions on the information set and the initial knowledge.

\begin{assumption}\label{ass:4}
At period $t$, firm $i$'s information set is given by the last realized production: $\mathcal{I}_i(t-1) = \{x_i(t-1)\}$.  
\end{assumption}

\begin{assumption}\label{ass:5} At each period $t$, firms re-initialize their probability distribution $\pi_{t,i}(\zeta_i)$ as a zero-truncated Gaussian with parameters $\zeta_{i}(t-1)$ and $\tau_i$.
\end{assumption}

Departing from Assumptions \ref{ass:2} and \ref{ass:3}, the Bayesian update resulting from Assumptions \ref{ass:4} and \ref{ass:5} retains information from only the most recent period by re-initializing the prior distribution as a truncated Gaussian centered at the previous MAP estimation. This reflects a form of memory loss, where all past information is distilled into the prior update $\zeta_{i}(t-1)$. The likelihood function of $\zeta_i(t)$ reduces to the probability distribution over the period $t-1$ production, which is induced by the idiosyncratic shock $\eta_{i}(t-1)$:
\begin{equation}\label{eq:likelihood_PI}
    \mathscr{L}(\zeta_i(t); \, \mathcal{I}_i(t-1) ) = \displaystyle \frac{1}{(2\pi \sigma_i)^{\frac{1}{2}} } \exp \left(- \frac{1}{2 \sigma_i^{2}} \left( \log x_i(t-1) - \log \mu_i(t-1) \right)^2\right), \quad \mbox{ for } t \geq 2.
\end{equation}

The Bayesian posterior distribution of $\zeta_i(t)$ retains the general form given by \eqref{eq:T-posterior}, previously introduced in the context of path-dependent learning. By applying a procedure analogous to that outlined in Propositions \ref{prop:v_learning_long_mem},  \ref{prop:expect_v_long_memory} and \ref{prop:mode_v_long_memory}, with minor algebraic modifications, we derive the MAP estimator (Proposition \ref{prop:v_learning_short_mem} below), its expectation and variance (Proposition \ref{prop:expect_v_short_memory} below), and its mode (Proposition \ref{prop:mode_v_short_memory} below), for the path-independent case.

\begin{proposition}[MAP dynamics (path-independent)]\label{prop:v_learning_short_mem} For each $t \geq 1$, the MAP estimator is
\begin{equation}\label{eq:v_learning_n1_short_mem}
\zeta_i(t+1) = \phi(\zeta_i(t), x_i(t)) = \left\{
\begin{array}{ll}
      \left(\dfrac{ \zeta_i(t) + \gamma_i z_{i}(t)\left(\varepsilon_{i}(t) +  \zeta_i^{(*)} z_{i}(t) \right) }{1 + \gamma_i z_{i}(t)^2} \right)^{+}, & ~\Rew{\mbox{if } l_i(t) > 0},\\[0.6cm]
      \zeta_i(t) & ~ \mbox{otherwise. }
\end{array} \right.
\end{equation}
\end{proposition}

\begin{proposition}[Expectation and variance of the MAP estimator (path-independent)]\label{prop:expect_v_short_memory} For each $t \geq 1$, let us define
$$
\overline{\overline{v}}_{i,t}(L_i(t))  = \dfrac{ \zeta_i(t)  + \gamma_i \zeta_i^{(*)} z_i(t)^2}{1 + \gamma_i z_i(t)^2}, \quad \mbox{ and } \quad
   \overline{\overline{\varphi}}_i(L_i(t))  = \dfrac{\gamma_i  \sigma_i z_i(t)}{1 + \gamma_i z_i(t)^2}.
$$
Conditioned on the $i$-th input decisions up to period $t-1$, we have
\begin{eqnarray}
\mathbb{E}[\zeta_i(t+1)] & = & \overline{\overline{v}}_i(L_i(t)) \tilde{\tilde{F}}_i +  |\overline{\overline{\varphi}}_i(L_i(t))| \tilde{\tilde{f}}_i, \\[0.4cm]
\mathbb{E}[\zeta_i(t+1)^2] & = & \left(\overline{\overline{v}}_i(L_i(t))^2 +  |\overline{\overline{\varphi}}_i(L_i(t))|^2\right) \tilde{\tilde{F}}_i + \overline{\overline{v}}_i(L_i(t)) |\overline{\overline{\varphi}}_i(L_i(t))|\tilde{\tilde{f}}_i,
\end{eqnarray}
where 
$$
\begin{array}{ll}
    \tilde{\tilde{F}}_i & \equiv G_i(z_{i}(t),z_{i}(t)^2) = 1-F\left(-\left|\dfrac{\overline{\overline{v}}_i(L_i(t))}{\overline{\overline{\varphi}}_i(L_i(t))}\right|\right),\\[0.5cm]
    \tilde{\tilde{f}}_i & \equiv g_i(z_{i}(t),z_{i}(t)^2) = f\left(-\left|\dfrac{\overline{\overline{v}}_i(L_i(t))}{\overline{\overline{\varphi}}_i(L_i(t))}\right|\right) .
\end{array}
$$
\end{proposition}

\begin{proposition}[Mode of the MAP estimator (path-independent)]\label{prop:mode_v_short_memory}
\Rew{Let $\mathbb{M}[\zeta_i(t) ~|~ L_i(t) ]$ be the mode of the MAP estimator $\zeta_i(t)$, conditioned on the $i$-th input decisions up to period $t$. Based on Proposition \ref{prop:v_learning_short_mem}, we have}
$$
\mathbb{M}[\zeta_i(t+1) ~|~ L_i(t) ]  = 
\begin{cases}
  \overline{\overline{v}}_i(L_i(t)), & \quad \mbox{if } \displaystyle \dfrac{1}{|\overline{\overline{o}}_i(L_i(t))|\sqrt{ 2 \pi}} \geq  F\left(-\left|\dfrac{\overline{\overline{v}}_i(L_i(t))}{\overline{\overline{\varphi}}_i(L_i(t))}\right|\right)\\[0.3cm]
  0 & \quad \mbox{otherwise.}
\end{cases}
$$
\end{proposition}

Notably, the functional form of the MAP estimator provided in Proposition \ref{prop:v_learning_short_mem}, as well as its expectation, mode and variance, provided in Propositions \ref{prop:expect_v_short_memory} and \ref{prop:mode_v_short_memory}, are equivalent to the ones studied in the context of path-dependent learning (see Propositions \ref{prop:v_learning_long_mem}, \ref{prop:expect_v_long_memory} and \ref{prop:mode_v_long_memory}). However, their limit properties are much less restrictive and favor the emergence of scenarios in which the learning path converges to the true returns to scale.

\subsection{Limit properties and convergence under path-independent learning}\label{subsec:limit_short_mem}

This subsection mirrors the convergence analysis presented in Subsection \ref{subsec:limit_long_mem} (where we examined the convergence of the expectation and mode of the MAP estimator under path-dependent learning). As in the case of Propositions \ref{prop:long_memory_learning} and \ref{prop:convergence_mode_long_memory}, we again concentrate on the conditional distribution of the MAP estimator $\zeta_i(t) ~|~ \mathcal{I}_i(t-1)$. In doing so, we employ a mean-field approximation, where all instances of $\eta_i(1), \ldots, \eta_i(t-1)$ in $l_i(t)$ are replaced by $q_i$.

It will become increasingly clear that the path-independent (Markovian) learning is a much more favorable method for firms to uncover the true returns to scale parameter. As a starting point of this convergence analysis, Propositions \ref{prop:short_memory_learning} and \ref{prop:convergence_mode_short_memory} below provide the path-independent versions of Propositions \ref{prop:long_memory_learning} and \ref{prop:convergence_mode_long_memory}, respectively.

\begin{proposition}[Convergence of the expectation (path-independent learning)]\label{prop:short_memory_learning} 
We have the following convergence of the learning dynamics
$$
    \underset{t \rightarrow +\infty }\lim\mathbb{E}[\zeta_i(t) ~|~ l_i(t)]=\left\{
    \begin{array}{ll}
      \zeta_i^{(*)} & \text{if } \underset{t \rightarrow +\infty }\lim l_{i}(t) \neq 1 \\
      \zeta_i^{(*)}+\dfrac{\zeta_i(0)-\zeta_i^{(*)}}{\underset{h=1}{\overset{\infty}\prod}\left(1+\gamma_i z_{i}(h)^2\right)}  & \text{if } \underset{t \rightarrow +\infty }\lim l_{i}(t) = 1.
    \end{array}\right.
$$

\end{proposition}

\begin{proposition}[Convergence of the mode (path-independent)]\label{prop:convergence_mode_short_memory}
Let us define the ordered set 
$$
\Psi_i = \mathbb{N}/\left\{t \in \nset ~:~ |\overline{\overline{\varphi}}_i(L_i(t))|\sqrt{ 2 \pi} F\left(-\left|\dfrac{\overline{\overline{v}}_i(L_i(t))}{\overline{\overline{\varphi}}_i(L_i(t))}\right|\right)> 1 \right\}, \quad \mbox{ for } i \in \mathcal{N},
$$
along with the order function $\psi_i(t)$ that assigns to each $t \in \nset$ the element of $\Psi_i$ in the $t$-th position. We consider the indicator function $\tilde{\mathfrak{1}}_{i, t-1}$, taking value one if $t \in \Psi_i$, and zero, otherwise. 

Conditioned on the sequence of material input decisions $L_i(t)$, the mode of the MAP estimator follows a deterministic sequence $\left\{ \overline{v}_i(t) \right\}_{t\in \Psi}$, where 
\begin{equation}\label{eq:mode_shor_tmemory_general_form}
    \displaystyle \overline{v}_i(t+1)  = \zeta_i(t) \left(\prod_{h=0}^{t-1} \dfrac{\tilde{\mathfrak{1}}_{i,h}}{ (1 + \gamma_i z_i(h)^2)}\right) + \zeta_i^{(*)} \left(\gamma_i \sum_{h=0}^{t-1} z_i(h)^2 \prod_{s=h}^{t-1} \dfrac{\tilde{\mathfrak{1}}_{i,s}}{(1 + \gamma_i z_i(s)^2)} \right).
\end{equation}
We claim that there exists a sub-sequence $\{\overline{v}_i(\psi_i(t))\}_{t}$ such that 
$$
\begin{array}{ll}
\underset{t \rightarrow +\infty }\lim \mathbb{M}[\zeta_i(t) ~|~ L_i(t)] & = \underset{t \rightarrow +\infty }\lim\overline{v}_i(\psi_i(t)) \\[0.5cm]
& =\left\{
    \begin{array}{ll}
      0 & \text{if } \Psi_i\text{ is a finite set,}\\
      \zeta_i^{(*)} & \text{if } \Psi_i\text{ is an infinite set and } \underset{t \rightarrow +\infty }\lim l_{i}(t) \neq 1 ,\\
      \zeta_i^{(*)}+\dfrac{\zeta_i(\psi_i(0))-\zeta_i^{(*)}}{\underset{h=1}{\overset{\infty}\prod}\left(1+\gamma_i z_{i}(\psi_i(h))^2\right)}  & \text{if } \Psi_i\text{ is an infinite set and }\underset{t \rightarrow +\infty }\lim l_{i}(t)=1.
    \end{array}\right.
\end{array}    
$$
\end{proposition} 

Importantly, the learning conditions of Propositions \ref{prop:short_memory_learning} and \ref{prop:convergence_mode_short_memory} are less restrictive than those of Proposition \ref{prop:long_memory_learning} and \ref{prop:convergence_mode_long_memory}. This unveils the \emph{benefit of forgetfulness} to facilitate firms to learn their true returns to scale, as the only case for which $\mathbb{E}[\zeta_i(t) ~|~ l_i(t)]$ does not converge to $\zeta_i^{(*)}$ is $\underset{t \rightarrow +\infty }\lim l_{i}(t) = 1$. When it comes to the mode of the MAP estimator, Proposition \ref{prop:convergence_mode_short_memory} provides evidence of the fact that the learning dynamics can still converge to $v^{(*)}_i$ even if $\underset{t \rightarrow +\infty }\lim l_{i}(t) = 1$ when $\gamma_i = (\tau_i/\sigma_i)^2$ grows large. The next two propositions show that also in the case of path-independent learning, an equilibrium dynamics satisfying $\underset{t \rightarrow +\infty }\lim l_{i}(t) \neq 1$ is possible under the equilibrium conditions of the proposed multi-sector model.

\begin{proposition}[Demographic expansion (path-independent)]\label{prop:demography_short_mem} 
By the equilibrium conditions \eqref{eq:labor_eq}-\eqref{eq:equilibrium_w}, any increasing demographic path $\{\delta(t)\}_t$ (with $\delta(t) \geq \delta(t-1) + 1$) is such that $\underset{t \rightarrow +\infty }\lim l_{i}(t) \neq 1$. 
\end{proposition}

\begin{proposition}[Demographic requirement (path-independent)]\label{prop:l_i_bound} For each $i \in \mathcal{N}$, let us define 
$$
\kappa_{i,t} = - \dfrac{1 }{1 - \hat{\zeta}_{i}(t)\alpha } \quad \mbox{ and } \quad 
\chi_{i,t}  = \left(\mathbb{E}\left[ \eta_{i}(t)^{\alpha} \right] \hat{\zeta}_{i}(t) \right)^{\dfrac{1 }{1 - \hat{\zeta}_{i}(t)\alpha }}. 
$$
\RRew{For any $\lambda_{\min}$, $\lambda_{\max} > 0$, we have that
$$
\displaystyle \mbox{if } \quad \hat{\zeta}_{i}(t) ~\geq~ \dfrac{1}{\mathbb{E}\left[ \eta_{i}(t)^{\alpha} \right]}\max_{j}\left \{  \left( \dfrac{\delta(t) - l_0(t)}{n \chi_{j,t}}\right)^{\frac{1}{\kappa_{j,t}}} \right\} \left(\lambda_{\min}\right)^{-\frac{1}{\kappa_{i,t}}}, \quad \mbox{ then } \quad l_i(t) ~\geq~ \lambda_{\min},
$$
and,
$$
\displaystyle \mbox{if } \quad  \hat{\zeta}_{i}(t)  ~\leq~ \dfrac{1}{\mathbb{E}\left[ \eta_{i}(t)^{\alpha} \right]} \min_{j}\left \{  \left( \dfrac{\delta(t) - l_0(t)}{n \chi_{j,t}}\right)^{\frac{1}{\kappa_{j,t}}} \right\}\left(\lambda_{\max}\right)^{-\frac{1}{\kappa_{i,t}}}, \quad \mbox{ then } \quad l_i(t) ~\leq~ \lambda_{\max}.
$$}
\end{proposition}

Therefore, \RRew{by picking $\lambda_{\min} > 1$}, the sufficient condition in Proposition \ref{prop:l_i_bound} entails $\underset{t \rightarrow +\infty }\lim l_{i}(t) \neq 1$, for all $i \in \mathcal{N}$. Notably, both Propositions \ref{prop:demography_short_mem} and \ref{prop:l_i_bound} establish that an increasing labor supply can support learning under the specific equilibrium conditions of the proposed multi-sector model. 

While the overall picture of the path-independent MAP estimator to infer the returns to scale parameter is in line with the case of path-dependent learning (Example \ref{ex:2}), the following stylized example (Example \ref{ex:4}) clarifies that this convergence is much faster and more accurate in the context of path-independent learning.

\begin{example}\label{ex:4} We consider the illustrative case described in Example \ref{ex:1}, where $n=1$, $\zeta_{i}^{(\ast)}= 0.5$, $\alpha = 0.5$ and $m_i = 0$. Figures \ref{fig:Example4_1} and \ref{fig:Example4_2} show the dynamics of $v_1(t)$ along a time horizon $\{1, \ldots, 500\}$, for different values of $\tau_i$ and $\zeta_{i}^{(0)}$. While the convergence of $\zeta_{i}(t)$ to $\zeta_{i}^{(\ast)}$ results under all scenarios, its speed strongly depends on $\tau_i$.

\begin{figure}[H]
        \centering
        \begin{subfigure}[b]{0.4\textwidth}\includegraphics[scale=0.67]{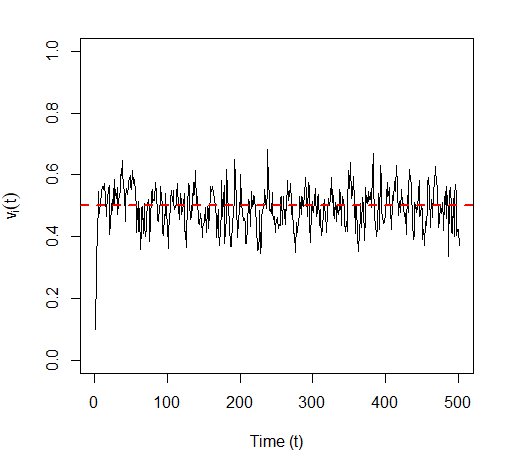}
        \caption{\footnotesize Setting $\zeta_{i}^{(0)} = 0.1$.}
        \end{subfigure}
        \begin{subfigure}[b]{0.4\textwidth}
        \includegraphics[scale=0.67]{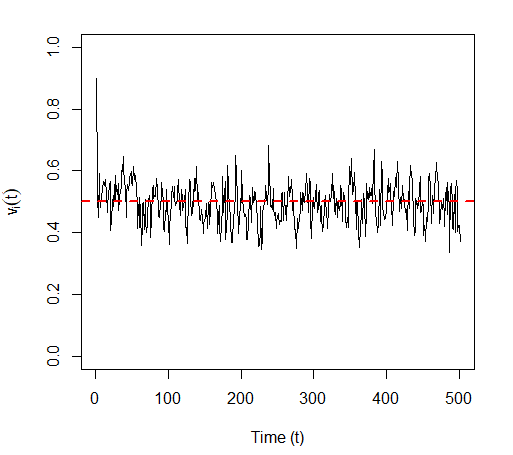}
        \caption{Setting $\zeta_{i}^{(0)} = 0.9$.}
        \end{subfigure} 
        \caption{\footnotesize Learning path $\{\zeta_i(t)\}_t$ with $\sigma_i = 0.1$ and $\tau=0.1$  \label{fig:Example4_1} }
\end{figure}

\begin{figure}[H]
        \centering
        ~~~\begin{subfigure}[b]{0.4\textwidth}     
        \includegraphics[scale=0.65]{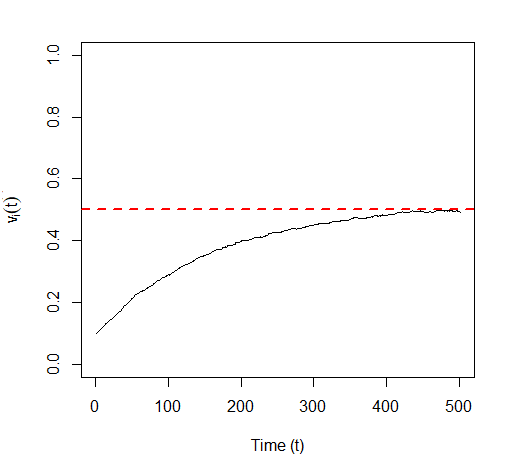}
        \caption{\footnotesize Setting $\zeta_{i}^{(0)} = 0.1$.}
        \end{subfigure} ~
        \begin{subfigure}[b]{0.4\textwidth}
        \includegraphics[scale=0.65]{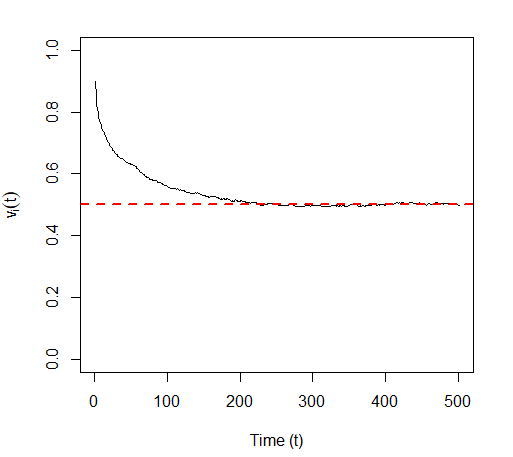}
        \caption{Setting $\zeta_{i}^{(0)} = 0.1$.}
        \end{subfigure} 
        \caption{\footnotesize Learning path $\{\zeta_i(t)\}_t$ with $\sigma_i = 0.1$ and $\tau=0.01$.  \label{fig:Example4_2} }
\end{figure}   

Table \ref{tab:Ex4} reports the difference between $\mathbb{E}[\zeta_i(1000)]$ and $\zeta_{i}^{(\ast)}$, for different values of $\sigma_i$ and $\tau_i$. The numerical values in Table \ref{tab:Ex4} support our theory about the convergence of the MAP estimator in the context of path-independent learning, and reaffirms that a major role is played by the variance of the idiosyncratic productivity shocks $\sigma_i$ and the degree of uncertainty of the prior distribution $\tau_i$. At the same time, sufficient conditions for the convergence of the learning path to $\zeta_i^{(\star)}$ are obtained by considering a minimum labor supply $\delta(t)$ (see Propositions \ref{prop:demography_short_mem} and \ref{prop:l_i_bound}).

\begin{table}[H]
\centering
\scalebox{0.98}{   
    \begin{tabular}{|cc|cc|cc|}
        \hline
        $\tau_1$ & $\sigma_1$ & \multicolumn{2}{c|}{$\zeta_{i}^{(\ast)}=0.4$} & \multicolumn{2}{c|}{$\zeta_{i}^{(\ast)}=0.6$} \\ \cline{3-6}
        & & $\zeta_{i}^{(0)} = 0.1$ & $\zeta_{i}^{(0)} = 0.9$ & $\zeta_{i}^{(0)} = 0.1$ & $\zeta_{i}^{(0)} = 0.9$ \\ 
        \hline
        $0.01$ & $0.01$ & 1.387e-16 & 1.388e-16 & 1.850e-16 & $<$1.00e-20 \\
        $0.01$ & $0.05$ & 0.0036624 & 0.0011406 & 0.0007456 & 0.0000265 \\
        $0.01$ & $0.10$ & 0.0594514 & 0.0594514 & 0.1126307 & 0.0075471 \\\hline
        $0.05$ & $0.01$ & $<$1.00e-20 & $<$1.00e-20 & $<$1.00e-20 & $<$1.00e-20 \\
        $0.05$ & $0.05$ & $<$1.00e-20 & $<$1.00e-20 & $<$1.00e-20 & $<$1.00e-20 \\
        $0.05$ & $0.10$ &   8.995e-11 & 1.683e-11 & 1.665e-15
 & 0.0049896 \\\hline
        $0.10$ & $0.01$ & $<$1.00e-20 & $<$1.00e-20 & $<$1.00e-20 & $<$1.00e-20 \\
        $0.10$ & $0.05$ & $<$1.00e-20 & $<$1.00e-20 & $<$1.00e-20 & $<$1.00e-20 \\
        $0.10$ & $0.10$ &   2.775e-16 & $<$1.00e-20 & 1.850e-16 & $<$1.00e-20
 \\\hline
    \end{tabular}}
    \caption{\footnotesize Relative absolute differences between the $\zeta_i(t)$ and $\zeta_i^{(\star)}$. \label{tab:Ex4}}
\end{table}

\end{example}

\section{A Bayesian (deterministic belief) filter with endogenous feedback}\label{subsec:BayesianFiltering}

We now turn our attention to providing a mathematically grounded explanation of the observed distinctions between the long-memory (path-dependent) and short-memory (path-independent) learning, based on the theory of Bayesian filtering. We first show that the alternating dynamics between the decision-making and learning stages define a class of Bayesian filters with endogenous feedback (more specifically, deterministic belief filters), where the evolution of the latent state is directly governed by the agent's optimization problem. To interpret our results from the Bayesian filtering perspective, we relate the general framework of recursive Bayesian estimation to the learning path formulation studied in this paper. Using an analogous notation as the one adopted so far, a recursive Bayesian filter is typically represented as a discrete-time system governed by the following equations \citep{DORE2009213}:
$$
v(t) = \nu(v(t-1), \omega(t-1)) \quad \mbox{ and } \quad x(t) = \mu(v(t), \eta(t)), 
$$
where \( v(t) \) is the latent state vector at time \( t \), \( x(t) \) is the observable signal, \( \omega(t) \) is the process noise (typically i.i.d.), and \( \eta(t) \) is the observation noise. The functions \( \nu \) and \( \mu \) represent the state transition and observation models, respectively. Estimation is carried out recursively in two steps: a prediction step using the system dynamics to propagate the prior, and an update step using the latest observation. A closed-form solution is available when the system is linear and the noise terms are Gaussian, in which case the optimal estimator is given by the Kalman filter. However, the probabilistic setting proposed in our paper departs from this classical setup in two crucial respects: (i) the transition function \( \nu \) is deterministic and observation-driven, and (ii) the observation function \( \mu \) is nonlinear and endogenous. First, as shown in Sections \ref{sec:learning_dynamics_long} and \ref{sec:learning_dynamics_short}, our learning dynamics are expressed as \( \zeta_i(t+1) = \nu(x_i(1), \ldots, x_i(t)) \) in the path-dependent case, and \( \zeta_i(t+1) = \nu(\zeta_i(t), x_i(t)) \) in the path-independent case. These are deterministic MAP estimators based on observed output, without any exogenous process noise \( \omega(t) \). This makes the learning process entirely observation-driven, unlike standard filtering where state transitions are governed by a stochastic law. Second, the observation function \( \mu \) is nonlinear in the latent state \( \zeta_i(t) \) due to the equilibrium feedback. In particular, from Proposition \ref{prop:equilibrium}, the observed production can be expressed as:
\[
x_i(t) = \eta_i(t) l_i(t)^{\zeta_i^*} = \eta_i(t) \left[ \left( \mathbb{E}[\eta_i(t)^\alpha] \right)^{1/(1 - \hat{\zeta}_i(t))} \left( \frac{\hat{\zeta}_i(t)}{w(t)} \right)^{1/[(1 - \hat{\zeta}_i(t))\alpha]} \right]^{\zeta_i^*},
\]
where \( \hat{\zeta}_i(t) = \max\left\{ \underline{\zeta}, \min\left\{ \overline{\zeta}, \zeta_i(t) \right\} \right\} \). This introduces an endogenous feedback loop: the signal \( x_i(t) \) depends on labor input, which is chosen based on the current belief \( \zeta_i(t) \), which in turn is updated based on \( x_i(t) \). Moreover, \( w(t) \), the equilibrium wage, depends on all firms' beliefs, further coupling the observation structure across agents.

As a result of these departures, standard convergence results for adaptive linear filters (e.g., those reviewed in \citet{liu2011kernel}) do not directly apply. These results typically assume exogenous, fixed regressors and linearity in both state and observation equations. In contrast, \RRew{$\zeta_i(t+1)$ is obtained recursively from $\zeta_i(t)$ and the newly observed signal $x_i(t)$, which itself depends on $\zeta_i(t)$, where the regressors \( z_i(t) = \log l_i(t) \) are themselves endogenous functions of \( \zeta_i(t) \). Hence, the learning rule is nonlinear and belief-dependent, exhibiting serial correlation. This dependence is not an additional assumption, but an intrinsic property of Bayesian updating.

Consequently, while our recursive update rules share a formal resemblance to adaptive linear filters, the convergence behavior must be understood through the lens of nonlinear Bayesian learning under endogenous feedback.} To investigate the convergence behavior of our learning rules (particularly the faster convergence of the path-independent update), we now turn to a unified characterization of both learning mechanisms.

\begin{proposition}\label{prop:comparison_learning} Let $p \in \{PD, PI\}$ and define
$$
A(u) = \frac{1}{1 - u} \log \mathbb{E}[\eta_i(t)^\alpha], \quad B(u, w) = \frac{1}{(1 - u)\alpha} \log\left( \frac{u}{w} \right), \quad \sigma^*_i(u,w) =  \left( \frac{\sigma_i}{ A(u) + B(u ,w)  } \right)^2. 
$$

Both path-dependent (PD) and path-independent (PI) learning equations \eqref{eq:v_learning_n1_long_mem} and \eqref{eq:v_learning_n1_short_mem} can be commonly expressed as:
$$
\zeta_i(t+1) = \Big( a_t^{p} \zeta_i(t) + (1-a_t^{p}) \varepsilon_i^*(t) \Big)^+, \quad \mbox{ where } \quad \varepsilon_i^*(t) \sim \mathcal{N}\Big( \zeta_i^{(*)},\sigma^*_i(\hat{\zeta}_i(t), w(t)) \Big),
$$
and
$$
a_t^{PD} =  \frac{1 + \gamma_i \sum_{l=1}^{t-1}  \left( A(\hat{\zeta}_i(l)) + B(\hat{\zeta}_i(l), w(l)) \right)^2 }{1 + \gamma_i \sum_{l=1}^t \left( A(\hat{\zeta}_i(l)) + B(\hat{\zeta}_i(l), w(l)) \right)^2} \quad \mbox{ and } \quad a_t^{PI} = \frac{1}{1 + \gamma_i \left( A(\hat{\zeta}_i(t)) + B(\hat{\zeta}_i(t), w(t)) \right)^2 }.
$$
\end{proposition}

This expression reveals that the estimator can be written as a convex combination of the previous estimate \( \zeta_i(t) \) and the current-period sufficient statistic (a Gaussian centered at $\zeta_i^{(*)}$), with time-varying weights. In the case of path-dependent learning, these weights increasingly favor \( \zeta_i(t) \) as \( t \to \infty \) under fairly mild conditions. Differently, the behavior of $a_t^{PI}$ is non-monotonic and state-dependent. To analyze the asymptotic behavior of $a_t^{PI}$ depending on the limiting behavior of $\hat{\zeta}_i(t)$, we consider thee cases.

One case is when $\hat{\zeta}_i(t) \to 1^-$, which implies $\frac{1}{1 - \hat{\zeta}_i(t)} \to \infty$. Consequently, both $A(\hat{\zeta}_i(t))$ and $B(\hat{\zeta}_i(t), w(t))$ diverge, and thus $a_t^{PI} \to 0$. In this case, $\zeta_i(t+1) \to \mathcal{N}\left( \zeta_i^{(*)}, \sigma^*_i(\hat{\zeta}_i(t), w(t)) \right)$, with $\sigma^*_i(\hat{\zeta}_i(t), w(t)) \to 0$. Another case is when $\hat{\zeta}_i(t) \to 0^+$, which implies that $A(\hat{\zeta}_i(t))$ remains bounded, but $B(\hat{\zeta}_i(t), w(t)) \to -\infty$ because $\log(\hat{\zeta}_i(t)/w(t)) \to -\infty$. Also in this case $a_t^{PI} \to 0$, so that $\zeta_i(t+1) \to \mathcal{N}\left( \zeta_i^{(*)}, \sigma^*_i(\hat{\zeta}_i(t), w(t)) \right)$, with $\sigma^*_i(\hat{\zeta}_i(t), w(t)) \to 0$. Finally, when $\hat{\zeta}_i(t)$ remains bounded away from $0$ and $1$, the functions $A(\hat{\zeta}_i(t))$ and $B(\hat{\zeta}_i(t), w(t))$ remain bounded too. Thus, $z_i(t)^2$ is bounded, and: $a_t^{PI} \in (0,1)$ which implies a balanced weighting between past belief and a Gaussian centered at $\zeta_i^{(*)}$. This regime corresponds to persistent, moderate learning behavior.

In conclusion, while the dynamics of the path-dependent learning resembles an auto-regressive process with a vanishing noise centered at $\zeta_i^{(*)}$, the path-independent learning mirrors a nonlinear dynamics in which the convergence to $\zeta_i^{(*)}$ accelerates when the sequence passes close to the boundaries zero and one.


\section{Extension to high-dimensional learning}\label{sec:Hight_dimension}


For the sake of generality of the proposed learning analysis, we conclude our study by investigating the possibility to extend our approach to the case of multi-input production with unknown input-output elasticities. Following \cite{horvath2000sectoral} and \cite{acemoglu2012network, acemoglu2017microeconomic}, each good in the economy can be either consumed or used in the next period by other sectors as an input for production. Hence, we let $y_{i,j}(t)$ denote the amount of production of sector $j$ used as an input for sector $i$ at period $t$, and establish the following clearing conditions:
$$
x_j(t) = c_j(t+1) + \sum_{i\in \mathcal{N}} y_{i,j}(t+1).
$$
The output of each sector is
\begin{equation}\label{eq:ExpectedProduction_multi}
x_i(t) = \eta_{i}(t)\mu_i(t), \quad \mbox{ with } \quad \mu_i(t) = (l_{i}(t))^{(1-\phi)} \prod_{j=1}^{n} (y_{i,j}(t))^{ \phi \beta_{i,j}} ,
\end{equation}
\noindent The exogenous parameter $\phi \in [0, 1]$ establishes the material versus labor intensity in production. The inter-sectoral input-output elasticities are summarized by $\phi \beta_{i,j}$, measuring the output response to a change in levels of production inputs, where $\beta_{i,j}$ is the input-specific factor of this elasticity. In this case, $\beta_{i,1}, \ldots, \beta_{i,n}$ represent the unknown payoff-relevant parameters firm $i$ wishes to learn by observing the outcomes (realized production) of its action (input decisions).\footnote{The multi-sector general equilibrium models of \cite{acemoglu2012network} and \cite{acemoglu2017microeconomic} rely upon constant returns to scale (i.e., $\sum_j \beta_{i,j} = 1$), whereas this assumption is not set in our modelling design.} We let $B$ be a $n \times n$ non-singular matrix with components $\beta_{i,j}$. Similarly, we let $B^{(*)} = \big( \beta_{i,j}^{(*)} \big)_{i,j=1}^n$ be the true input-output parameter matrix that firms wish to discover. Also, $Y(t)$ is used to denote a $n \times n$ matrix with components $y_{i,j}(t)$. The random factor $\eta_{i}(t) \sim \log N(m_i,\sigma_i )$ is defined as in Section \ref{sec:model}. 

Firm $i$ has an initial knowledge $\pi_{0,i}(\beta_{i,j})$ quantified as a zero-truncated Gaussian with parameters $\beta^{(0)}_{i,j}$ and $\tau_i$ (this is analogous to Assumption \ref{ass:3}, for the one-dimensional learning). We define:
\begin{equation}\label{eq:sr}
\displaystyle
s_i(t) = \sigma_i \varepsilon_i +  \sum_{j=1}^{n} \phi \beta_{i,j}^{(*)} \log y_{i,j}(t) \qquad \mbox{ and } \qquad z_{i,j}(t) = \phi \log y_{i,j}(t),
\end{equation}
where $\varepsilon_i \sim N(0,1)$. Further, for a given ordered subset of indices $\mathcal{U} \subseteq \mathcal{N}$, we define the $|\mathcal{U}|$-dimensional vector $\mathbf{z}_{i}(t, \mathcal{U})$, containing the components of $z_{i,j}(t)$, for every $j \in \mathcal{U}$, and the following collection of $|\mathcal{U}| \times |\mathcal{U}|$ matrices:
\begin{equation}\label{eq:H_matrix}
    H_i(t', t, \mathcal{U}) = \displaystyle \left(\sum_{\ell=t}^{t'} \frac{1}{\sigma_i^2} \mathbf{z}_{i,.}(\ell, \mathcal{U})\mathbf{z}_{i,.}(\ell, \mathcal{U})\T + \frac{1}{\tau_i^2} I_{|\mathcal{U}|} \right), \quad \mbox{ for } t' > t,
\end{equation}
where $I_{h}$ is the identity matrix of dimension $h$. Similarly, we consider the $|\mathcal{U}|$-dimensional vector  $\boldsymbol{\beta}_{i}^{(0)}(\mathcal{U})$, containing the components $\beta_{i,j}^{(0)}$, for every $j \in \mathcal{U}$, and define:
\begin{equation}\label{eq:B_vector}
\mathbf{b}_{i}(t', t, \mathcal{U}) = \displaystyle  \frac{1}{\tau_i^2}  \boldsymbol{\beta}_{i}(t-1, \mathcal{U}) + \frac{1}{\sigma_i^2} \sum_{\ell=t}^{t'} s_i(\ell) \mathbf{z}_{i,.}(\ell, \mathcal{U}), \quad \mbox{ for } t' > t,
\end{equation}
where $\boldsymbol{\beta}_{i}(t-1,\mathcal{U})$ is the $|\mathcal{U}|$-dimensional sub-vector of $\boldsymbol{\beta}_{i}(t-1)$, whose components correspond to the indices in $\mathcal{U}$ (in the same order), and set $\boldsymbol{\beta}_{i}(0,\mathcal{U}) = \boldsymbol{\beta}_{i}^{(0)}(\mathcal{U})$. We introduce the notation $\Omega_{i,t} = \left\{ \mathcal{U} \subseteq \mathcal{N} : H_i(t, 1, \mathcal{U})^{-1} \mathbf{b}_{i}(t, 1, \mathcal{U}) > 0 \right\}$. As for the one-dimensional learning, the MAP estimator $B(t)$ is constructed by maximizing the log-posterior distribution:
\begin{equation}\label{eq:MAP_W}
\begin{array}{ll}
    \displaystyle \max_{B(t) \geq 0} ~ \sum_{h=1}^{t} - \frac{1}{ \sigma_i^{2}} \left( \log x_i(h) - \log (l_{i}(h))^{(1-\phi)} \prod_{j=1}^{n} (y_{i,j}(h))^{ \phi \beta_{i,j}(t) }\right)^2- \frac{1}{\tau_i^{2}} \left(\beta_{i,j}(t) - \beta_{i,j}^{(0)} \right)^2\\[0.5cm]
    \displaystyle \max_{B(t) \geq 0} ~ - \frac{1}{ \sigma_i^{2}} \left( \log x_i(t) - \log (l_{i}(t))^{(1-\phi)} \prod_{j=1}^{n} (y_{i,j}(t))^{ \phi \beta_{i,j}(t) }\right)^2- \frac{1}{\tau_i^{2}} \left(\beta_{i,j}(t) - \beta_{i,j}(t-1) \right)^2
\end{array}
\end{equation}
for the path-dependent and path-independent case, respectively. Below we provide an algebraic characterization of the MAP estimators, consistently with the ones in Propositions \ref{prop:v_learning_long_mem} and \ref{prop:v_learning_short_mem}.

\begin{proposition}[MAP dynamics]\label{prop:multidim_MAP}
For each $i \in \mathcal{N}$ and $t \geq 1$, we have
$$
\boldsymbol{\beta}_i(t,\Omega_{i,t}) = H_i(t, 1, \Omega_{i,t})^{-1} \mathbf{b}_{i}(t, 1, \Omega_{i,t}) \qquad \mbox{ and } \qquad \boldsymbol{\beta}_i(t, \mathcal{N}/\Omega_{i,t}) = \mathbf{0},
$$
for the path-dependent learning, and
$$
\boldsymbol{\beta}_i(t,\Omega_{i,t}) = H_i(t, t, \Omega_{i,t})^{-1} \mathbf{b}_{i}(t, t, \Omega_{i,t}) \qquad \mbox{ and } \qquad \boldsymbol{\beta}_i(t, \mathcal{N}/\Omega_{i,t}) = \mathbf{0},
$$
for the path-independent learning.
\end{proposition}

This proposition has an important consequence. In fact, the complexity of the statistical learning in a parameters space of dimension $n$ from the realization of a noisy production of dimension one is reflected by the non-identifiablility of the log-likelihood function (corresponding to the first term in the expression \eqref{eq:MAP_W}), with respect to the whole input elasticities $\beta_{i,1}(t), \ldots \beta_{i,n}(t)$. This can be seen from the fact that $H_i(t', t, \mathcal{U})$ becomes singular when $\tau_i = \infty$ (improper uniform prior) and $t' - t$ is sufficiently small. In other words, in the absence of a prior knowledge, only a long series of realized productions can result in an identifiable learning model. Conversely, the existence of prior knowledge is sufficient for ensuring the identifiablility.

As established in Lemma \ref{lemma:H_conv} and Corollary \ref{cor:f_conditional} below, interesting properties can be deduced for two limit cases: when either $\gamma_i = (\tau_i/\sigma_i)^2$ or $\phi$ go to zero. 

\begin{lemma}[Diagonal inverse]\label{lemma:H_conv} For any $t \geq t' \geq 0$ and any $\Omega_{i,t} \subseteq \mathcal{N}$, we have 
$$
\displaystyle \lim_{\gamma_i \rightarrow 0} ~ H_{i,}^{-1}(t, t', \Omega_{i,t}) ~=~ 
\displaystyle \lim_{\phi \rightarrow 0} ~ H_{i,}^{-1}(t, t', \Omega_{i,t}) ~=~ \tau_i^2 I_{|\Omega_{i,t}|}.
$$
\end{lemma}

\begin{corollary}[Conditional distribution of $B(t)$]\label{cor:f_conditional} Based on Lemma \ref{lemma:H_conv}, as  $\gamma_i$ or $\phi$ approach zero, we have the following limit behavior
$$
\displaystyle \boldsymbol{\beta}_{i,.}(t+1, \Omega_{i,t}) ~ \rightsquigarrow~   \boldsymbol{\beta}_{i,.}^{(0)}(\Omega_{i,t}) + \gamma_i \sum_{\ell=1}^{t}  \left(\boldsymbol{\varepsilon}(\ell) + (B^{(*)}(\Omega_{i,t})\mathbf{z}_{i,.}(\ell, \Omega_{i,t})) \otimes \mathbf{z}_{i,.}(\ell, \Omega_{i,t})\right), \mbox{ and } \boldsymbol{\beta}_i(t, \mathcal{N}/\Omega_{i,t}) = \mathbf{0},
$$
for the path-dependent learning, and
$$
\displaystyle \boldsymbol{\beta}_{i,.}(t+1, \Omega_{i,t}) ~ \rightsquigarrow~   \boldsymbol{\beta}_{i,.}(t, \Omega_{i,t}) + \gamma_i \left(\boldsymbol{\varepsilon}(t) + (B^{(*)}(\Omega_{i,t})\mathbf{z}_{i,.}(t, \Omega_{i,t})) \otimes \mathbf{z}_{i,.}(t, \Omega_{i,t})\right), \mbox{ and } \boldsymbol{\beta}_i(t, \mathcal{N}/\Omega_{i,t}) = \mathbf{0},
$$
for the path-independent learning, where $\boldsymbol{\varepsilon}(\ell) \sim N(0,\sigma_i I_{|\Omega_{i,t}|} )$,  $\otimes$ is the element-wise matrix product, and $\rightsquigarrow$ denotes the limit behavior.\footnote{Formally, given functions $f(x)$ and $g(x)$, we have $f(x)\rightsquigarrow g(x)$ if and only if $\lim_{x \rightarrow \infty} \frac{f(x)}{g(x)} = 1$ \citep{de1981asymptotic}.} 
\end{corollary}

From Corollary \ref{cor:f_conditional}, $\boldsymbol{\beta}_{i,.}(t+1, \Omega_{i,t})$ behaves as a combination of $\boldsymbol{\beta}_{i,.}^{(0)}(\Omega_{i,t})$ (the initial belief) and $B^{(*)}$ (the true parameter matrix). Consequently, when the uncertainty of the prior distribution is much smaller than the uncertainty of the production (i.e., when $\gamma_i$ goes to zero), the distribution of $\boldsymbol{\beta}_{i,.}(t+1, \Omega_{i,t})$ is centered around $\boldsymbol{\beta}_{i,.}^{(0)}(\Omega_{i,t})$, so that firms cannot learn. This result reaffirms an analogous pattern observed in Propositions \ref{prop:long_memory_learning} and \ref{prop:convergence_mode_long_memory}.

As already studied for the one-dimensional case (the returns to scale learning), the short memory can also be analyzed when learning the input-output elasticity matrix. Formally, in the short memory method, \eqref{eq:H_matrix} and \eqref{eq:B_vector} become
\begin{eqnarray}\label{eq:H_matrix_red}
        H_i(t+1, t+1, \mathcal{U}) & = & \displaystyle  \frac{1}{\sigma_i^2} \mathbf{z}_{i,.}(t+1, \mathcal{U})\mathbf{z}_{i,.}(t+1, \mathcal{U})\T + \frac{1}{\tau_i^2} I_{|\mathcal{U}|},\\[0.2cm]
        \mathbf{b}_{i}(t+1, t+1, \mathcal{U}) & = &  \displaystyle  \frac{1}{\tau_i^2}  \boldsymbol{\delta}_{i}(t,\mathcal{U}) + \frac{1}{\sigma_i^2}  s_i(t+1, \mathcal{U}) \mathbf{z}_{i,.}(t+1, \mathcal{U}),       
\end{eqnarray}
so that Corollary \ref{cor:f_conditional} reduces to
$$
\displaystyle \boldsymbol{\beta}_{i,.}(t+1, \Omega_{i,t}) ~ \rightsquigarrow~   \boldsymbol{\beta}_{i,.}(t,\Omega_{i,t}) + \gamma_i \left(\boldsymbol{\varepsilon}(t) + (B^{(*)}(\Omega_{i,t})\mathbf{z}_{i,.}(t, \Omega_{i,t})) \otimes \mathbf{z}_{i,.}(t, \Omega_{i,t})\right),
$$
still behaving as a combination of the previous belief $\boldsymbol{\beta}_{i,.}(t,\Omega_{i,t})$ and the true input-output elasticity structure $B^{(*)}$. This high-dimensional extension thus suggests that the whole analysis of input-output elasticity learning mirrors the dynamics of the returns to scale case, despite the increased complexity induced by the matrix expressions. However, a full proof for the belief convergence in this high-dimensional context is still an open question that we leave for future work.

\section{Conclusions}\label{sec:conclusions}

This study advances the understanding of learning mechanisms in a multi-sector general equilibrium framework. By embedding firms' information processing and decision-making into this structure, the research addresses the critical economic question of how firms adaptively uncover their production capabilities. We modeled a collection of representative firms operating in uncertain production environments without full knowledge of profit-relevant parameters. Our investigation focused on returns to scale (a cornerstone of production theory) and utilized a dynamic Bayesian framework to analyze firms' learning paths. Central to our analysis was the MAP estimator, which integrates input decisions and production outcomes to iteratively refine firms' beliefs.

The findings underscore key insights into the relationship between economic dynamics and learning. We identified the essential conditions under which firms can asymptotically learn the true returns to scale, highlighting the influence of idiosyncratic productivity shocks and the informativeness of prior beliefs. Notably, robust economic growth coupled with increasing input quantities fosters learning under a path-dependent approach, while path-independent learning offers a more efficient trajectory by prioritizing recent observations over historical data. In this vein, we demonstrated that long-memory (path-dependent) learning dynamics that keep track of all past estimations end up having worse performance than a short-memory (path-independent) approach.

In conclusion, this study constitutes a theoretical foundation for analyzing learning by economic agents within a multi-sector equilibrium framework, where the data-generating process that pins down the estimation problem is itself generated by firms' optimal input choices and market-clearing. In this sense, the Bayesian updating rule is grounded in equilibrium conditions and its properties are assessed jointly with the induced dynamics of allocations and prices. This structural link is also what underlies our comparison between long-memory and short-memory learning, since the sequence of observations used for inference is endogenously shaped by firms' decisions. 

Building on this approach, several extensions are natural. A particularly promising and challenging direction is a formal analysis of belief convergence in the high-dimensional setting developed in Section \ref{sec:Hight_dimension}, where firms must learn the entire input-output elasticity structure. While our results suggest that the dynamics closely mirror those of the one-dimensional case, a complete proof of convergence in this high-dimensional context remains an open question that we identify as a key avenue for future research. Secondly, future research can focus on applying the proposed MAP estimator in the Bayesian general equilibrium of \cite{toda2015bayesian}, where non-optimizing agents respond to prices by setting a prior distribution on their demand. Finally, the inclusion of durable capital inputs (which implies inter-temporal decisions by firms) might have an impact on the learning dynamics and deserves to be part of a future analysis.

\bibliographystyle{plainnat} 
\bibliography{references}    

\newpage

\appendix
\setcounter{page}{1}

\section{: Mathematical proofs}\
\label{Section:appendix1}


\subsection*{\textbf{Proposition \ref{prop:equilibrium}}}

\begin{proof} Let $x_i(t)$ and $\Bar{x}_i(t)$ denote the true and the believed production in sector $i$, respectively. Similarly, let $p_i(t)$ and $\Bar{p}_i(t)$ denote the corresponding true and believed prices in sector $i$. Note that both $x_i(t)$ and $\Bar{x}_i(t)$ (as well as $p_i(t)$ and $\Bar{p}_i(t)$) are random variables induced by the shock $\eta_i(t)$. In accordance with step (i) of the decision--estimation dynamics, for each $i \in \mathcal{N}$, the optimal demand for labor is obtained by solving \eqref{eq:firm_problem}:
\begin{equation}\label{eq:labor_demand}
    l_i(t) = \left(\dfrac{\hat{\zeta}_{i}(t)\,\RRew{\mathbb{E}\!\left[ \Bar{p}_i(t)\eta_{i}(t) \;\big|\;\mathcal{I}_i(t) \right]}}{ w(t) }\right)^{\dfrac{1}{1-\hat{\zeta}_i(t)}},
\end{equation}
where $\hat{\zeta}_i(t) = \max\left\{ \underline{\zeta}, \min\left\{ \overline{\zeta}, \zeta_i(t) \right\} \right\}$. Using the \RRew{true} consumer demand $p_{i}(t) = c_{i}(t)^{\alpha-1}$ and the market-clearing conditions $x_{i}(t) = c_{i}(t)$, the true prices satisfy
\RRew{
\[
x_i(t)=\eta_{i}(t)\left(\dfrac{\hat{\zeta}_{i}(t)P_{i,t}}{ w(t)}\right)^{\dfrac{\zeta_{i}^{(\ast)}}{1-\hat{\zeta}_{i}(t)}} 
= \left(\dfrac{1}{p_{i}(t)}\right)^{\dfrac{1}{1-\alpha}}
= c_i(t),
\]
therefore
\[
p_{i}(t)\,\eta_{i}(t)^{(1-\alpha)}
\left(\dfrac{\hat{\zeta}_{i}(t)P_{i,t}}{ w(t)}\right)^{\dfrac{\zeta_{i}^{(\ast)}(1-\alpha)}{1-\hat{\zeta}_{i}(t)}}
=1,
\]
}
where $P_{i,t}=\RRew{\mathbb{E}\!\left[ \Bar{p}_i(t)\eta_{i}(t)\;\big|\;\mathcal{I}_i(t)\right]}$. Hence,
\begin{equation}\label{eq:p_it}
    p_{i}(t)
    =\eta_{i}(t)^{\alpha-1}\left(\dfrac{\hat{\zeta}_{i}(t)P_{i,t}}{ w(t)}\right)^{\dfrac{\zeta_{i}^{(\ast)}(\alpha-1)}{1-\hat{\zeta}_{i}(t)}} .
\end{equation}

In the same way, the $i$-th representative firm deduces $\Bar{p}_i(t)$ from the believed consumer demand $\Bar{p}_{i}(t) = c_{i}(t)^{\alpha-1}$ and the market-clearing conditions $\Bar{x}_{i}(t) = c_{i}(t)$, obtaining
\[
\Bar{x}_{i}(t)
=\eta_{i}(t)\left(\dfrac{\hat{\zeta}_{i}(t)P_{i,t}}{ w(t)}\right)^{\dfrac{\hat{\zeta}_{i}(t)}{1-\hat{\zeta}_{i}(t)}}
=\left(\dfrac{1}{\Bar{p}_{i}(t)}\right)^{\dfrac{1}{1-\alpha}} .
\]

Note that \eqref{eq:p_it} is the true price that satisfies market clearing. To obtain the price expected by the $i$-th firm, we replace $\zeta_{i}^{(\ast)}$ with $\hat{\zeta}_{i}(t)$ in \eqref{eq:p_it}. \RRew{Since $\mathbb{E}[ \eta_{i}(t)^{\alpha} ] = \mathbb{E}[ \eta_{i}(t)^{\alpha} \mid \mathcal{I}_i(t) ]$}, we have
\[
\begin{array}{ll}
P_{i,t}
= \RRew{\mathbb{E}\!\left[ \Bar{p}_i(t)\eta_{i}(t) \;\big|\;\mathcal{I}_i(t) \right]}
&= \mathbb{E}\!\left[ \eta_{i}(t)^{\alpha}
\left(\dfrac{\hat{\zeta}_{i}(t)P_{i,t}}{ w(t)}\right)^{\dfrac{\hat{\zeta}_{i}(t)(\alpha-1)}{1-\hat{\zeta}_{i}(t)}}
\Bigg|\; \mathcal{I}_i(t) \right] \\[0.5cm]
&= \mathbb{E}\!\left[ \eta_{i}(t)^{\alpha} \right]
\left(\dfrac{\hat{\zeta}_{i}(t)P_{i,t}}{ w(t)}\right)^{\dfrac{\hat{\zeta}_{i}(t)(\alpha-1)}{1-\hat{\zeta}_{i}(t)}} .
\end{array}
\]
Hence,
\[
\begin{array}{ll}
\left( P_{i,t} \right)^{\dfrac{1-\hat{\zeta}_{i}(t)\alpha}{1-\hat{\zeta}_{i}(t)}}
&= \mathbb{E}\!\left[  \eta_{i}(t)^{\alpha} \right]\left(\dfrac{\hat{\zeta}_{i}(t)}{ w(t)}\right)^{\dfrac{\hat{\zeta}_{i}(t)(\alpha-1)}{1-\hat{\zeta}_{i}(t)}} ,\\[0.5cm]
P_{i,t}
&= \mathbb{E}\!\left[  \eta_{i}(t)^{\alpha} \right]^{\dfrac{1-\hat{\zeta}_{i}(t)}{1-\hat{\zeta}_{i}(t)\alpha}}
\left(\dfrac{\hat{\zeta}_{i}(t)}{ w(t)}\right)^{\dfrac{\hat{\zeta}_{i}(t)(\alpha-1)}{1-\hat{\zeta}_{i}(t)\alpha}} .
\end{array}
\]

Using \eqref{eq:labor_demand}, we obtain
\[
\begin{array}{ll}
l_i(t)
&= \left[\dfrac{\hat{\zeta}_{i}(t) P_{i,t}}{w(t)}\right]^{\dfrac{1}{1-\hat{\zeta}_i(t)}} \\[0.5cm]
&= \left[ \mathbb{E}\!\left[ \eta_{i}(t)^{\alpha} \right] \right]^{\dfrac{1}{1-\hat{\zeta}_i(t)\alpha}}
\left(\dfrac{\hat{\zeta}_{i}(t)}{ w(t)}\right)^{\dfrac{1}{1-\hat{\zeta}_{i}(t)\alpha}} .
\end{array}
\]

Using \eqref{eq:p_it}, we obtain
\[
\begin{array}{ll}
p_{i}(t)
&= \eta_{i}(t)^{\alpha-1}\left[\dfrac{\hat{\zeta}_{i}(t) P_{i,t}}{ w(t)}\right]^{\dfrac{\zeta_{i}^{(\ast)}(\alpha-1)}{1-\hat{\zeta}_{i}(t)}} \\[0.5cm]
&= \eta_{i}(t)^{\alpha-1}
\left[\mathbb{E}\!\left[  \eta_{i}(t)^{\alpha} \right]\right]^{\dfrac{\zeta_{i}^{(\ast)}(\alpha-1)}{1-\hat{\zeta}_{i}(t)\alpha}}
\left(\dfrac{\hat{\zeta}_{i}(t)}{ w(t)}\right)^{\dfrac{\zeta_{i}^{(\ast)}(\alpha-1)}{1-\hat{\zeta}_{i}(t)\alpha}} \\[0.5cm]
&= \left(\eta_{i}(t)\,l_i(t)^{\zeta_{i}^{(\ast)}}\right)^{\alpha-1}.
\end{array}
\]

Therefore,
$$
l_i(t) = \displaystyle \left[ \mathbb{E}\!\left[ \eta_{i}(t)^{\alpha} \right] \right]^{\dfrac{1}{1-\hat{\zeta}_i(t)\alpha}}
\left(\dfrac{\hat{\zeta}_{i}(t)}{ w(t)}\right)^{\dfrac{1}{1-\hat{\zeta}_{i}(t)\alpha}} \quad \mbox{and} \quad 
p_i(t) = \left(\eta_{i}(t)\,l_i(t)^{\zeta_{i}^{(\ast)}}\right)^{\alpha-1}.
$$
\end{proof}


\subsection*{\textbf{Lemma \ref{lemma:L}}}

\begin{proof}
We first express the labor demand as
$l_i(t)=\chi_{i,t}\,w(t)^{\kappa_{i,t}}$,
where
\[
\kappa_{i,t} = - \dfrac{1 }{1 - \hat{\zeta}_{i}(t)\alpha } \quad \mbox{ and } \quad 
\chi_{i,t}  =  \left( \mathbb{E}\!\left[ \eta_{i}(t)^{\alpha}  \right] \hat{\zeta}_{i}(t) \right)^{\dfrac{1 }{1 - \hat{\zeta}_{i}(t)\alpha } } .
\]
Since $\hat{\zeta}_i(t) < 1$ and $\alpha < 1$, we have $\hat{\zeta}_{i}(t)\alpha < 1$ and hence $\kappa_{i,t}<0$.
Therefore, $L_{t}(w)$ is continuous (as it is a sum of continuous functions), and
\[
\frac{d}{dw} L_{t}(w) \;=\; \sum_{i=1}^{n} \chi_{i,t}\,\kappa_{i,t}\, w^{\kappa_{i,t} - 1} \;<\; 0.
\]
As a consequence, we have $\lim_{w \rightarrow + \infty} L_t(w) = 0$ and $\lim_{w \rightarrow 0} L_t(w) = +\infty$.
\end{proof}


\subsection*{\textbf{Proposition \ref{prop:v_learning_long_mem}}}

\begin{proof}
We drop the index $i$ and the time period $t$ from $\zeta_i(t)$, as this proof is valid for all firms and all periods. Based on Assumption~\ref{ass:3}, $\pi_{0,i}$ is the density of a zero-truncated Gaussian random variable with parameters $\zeta^{(0)}_i$ and $\tau_i$. As clarified in Section \ref{sec:learning_dynamics_long}, the likelihood is written in terms of log-output. Hence, defining $s_i(\ell):=\log x_i(\ell)-m_i$ and $z_i(\ell):=\log l_i(\ell)$ (hence $\log\mu_i(\ell)=\zeta z_i(\ell)$), the conditional density of $s_i(\ell)$ given $(\zeta,z_i(\ell))$ is Gaussian:
\[
f\big(s_i(\ell)\mid \zeta,z_i(\ell)\big)
=\frac{1}{\sigma_i\sqrt{2\pi}}
\exp\!\left(-\frac{\big(s_i(\ell)-\zeta z_i(\ell)\big)^2}{2\sigma_i^2}\right).
\]
Assuming conditional independence across $\ell$, we obtain:
\[
\mathscr{L}(\zeta;\mathcal{I}_i(t))
=\prod_{\ell=1}^{t}\frac{1}{\sigma_i\sqrt{2\pi}}
\exp\!\left(-\frac{\big(s_i(\ell)-\zeta z_i(\ell)\big)^2}{2\sigma_i^2}\right).
\]
For the case $l_i(\ell)>0$ for all $\ell=1,\ldots,t$, we have $z_i(\ell)\in\mathbb{R}$, so that the (unnormalized) posterior satisfies
\[
\mathscr{L}(\zeta;\mathcal{I}_i(t))\,\pi_{0,i}(\zeta)\ \propto\
\begin{cases}
\displaystyle \exp\!\left(-\frac{1}{2}\theta_i(\zeta;\mathcal{I}_i(t))\right) & \text{if }\zeta\ge 0,\\[0.2cm]
0 & \text{otherwise,}
\end{cases}
\]
where $\theta_i$ is defined as
\begin{equation*}
\theta_i(\zeta;\mathcal{I}_i(t))=
\begin{cases}
\displaystyle \sum_{\ell=1}^{t}\frac{1}{\sigma_i^2}\big(s_i(\ell)-\zeta z_i(\ell)\big)^2
+\frac{1}{\tau_i^2}\big(\zeta-\zeta_i^{(0)}\big)^2, & \text{if } t>1,\\[0.6cm]
\displaystyle \frac{1}{\tau_i^2}\big(\zeta-\zeta_i^{(0)}\big)^2, & \text{if } t=1.
\end{cases}
\end{equation*}
To maximize the posterior distribution for each period $t$, it is sufficient to solve
\[
\max_{\zeta\ge 0}\ \exp\!\left(-\frac{1}{2}\theta_i(\zeta;\mathcal{I}_i(t))\right),
\]
since the normalizing constant $\int \mathscr{L}(\zeta;\mathcal{I}_i(t))\,\pi_{0,i}(\zeta)\,d\zeta$ does not depend on $\zeta$. By the Karush--Kuhn--Tucker conditions, there exists a multiplier $\lambda \ge 0$ such that
\begin{equation*}
\frac{\partial}{\partial \zeta}\left[ \exp\!\left(-\frac{1}{2}\theta_i(\zeta; \mathcal{I}_i(t))\right) - \lambda \zeta\right]=0,
\qquad \lambda \zeta = 0.
\end{equation*}
Therefore,
\[
-\frac{1}{2}\exp\!\left(-\frac{1}{2}\theta_i(\zeta; \mathcal{I}_i(t))\right)\left(\frac{\partial}{\partial \zeta}\theta_i(\zeta; \mathcal{I}_i(t))\right)=\lambda .
\]
Due to complementarity, either $\lambda = 0$ or $\zeta = 0$. Therefore, if $\zeta>0$, then $\lambda=0$ and $\frac{\partial}{\partial v}\theta_i(\zeta; \mathcal{I}_i(t))=0$. This implies that for any $i\in\mathcal{N}$, $\zeta$ is the maximum between zero and the solution of
\begin{equation*}
\begin{cases}
\displaystyle  \sum_{\ell=1}^{t} \frac{1}{\sigma_i^2} z_{i}(\ell)\Big( s_i(\ell) - \zeta\, z_{i}(\ell)  \Big) = \frac{1}{\tau_i^2} \left( \zeta - v_{i}^{(0)} \right)  & \quad \mbox{if } t > 1,\\[0.3cm]
\displaystyle \zeta = v_{i}^{(0)} & \quad \mbox{if } t = 1.
\end{cases}
\end{equation*}
Therefore, we obtain \eqref{eq:v_learning_n1_long_mem}, for the case $l_i(\ell) > 0$ for all $\ell = 1, \ldots, t$ and $i \in \mathcal{N}$.
\end{proof}


\subsection*{\textbf{Proposition \ref{prop:conv_pointwise}}}

\begin{proof}
Let us define the following transformation:
\[
\zeta_i(L_i(t), E_i(t)) \;=\; \left( \dfrac{\zeta^{(0)}_i  + \gamma_i \sum_{k=1}^{t}  z_i(k)\, s_i(k)}{1 + \gamma_i \sum_{k=1}^{t} z_i(k)^2} \right)^{+},
\]
where $L_i(t)=\{l_i(k)\}_{k=1}^{t}$ is a deterministic sequence and $E_i(t)=\{\varepsilon_i(k)\}_{k=1}^{t}$ is a sequence of i.i.d.\ Gaussian random variables $\varepsilon_i(k)\sim \mathcal{N}(m_i,\sigma_i)$. Using $s_i(k)=\varepsilon_i(k)+\zeta_i^{(*)}z_i(k)$, we can write
\[
\sum_{k=1}^{t} z_i(k)s_i(k)
= \sum_{k=1}^{t} z_i(k)\varepsilon_i(k) + \zeta_i^{(*)}\sum_{k=1}^{t} z_i(k)^2.
\]
As $l_i(\ell)\to 0$, we have $z_i(\ell)=\log l_i(\ell)\to -\infty$, while $z_i(k)$ is fixed for $k\neq \ell$. Hence,
\[
\dfrac{\sum_{k=1}^{t} z_i(k)\varepsilon_i(k)}{\sum_{k=1}^{t} z_i(k)^2}
= \dfrac{z_i(\ell)\varepsilon_i(\ell) + \sum_{k\neq \ell} z_i(k)\varepsilon_i(k)}{z_i(\ell)^2 + \sum_{k\neq \ell} z_i(k)^2}
\longrightarrow 0,
\]
and also $\dfrac{v_i^{(0)}}{\gamma_i \sum_{k=1}^{t} z_i(k)^2}\to 0$ and $\dfrac{1}{\gamma_i \sum_{k=1}^{t} z_i(k)^2}\to 0$. Therefore, for each realization $E_i(t)$,
\[
\underset{l_i(\ell)\rightarrow 0}{\lim}\, \zeta_i(L_i(t),E_i(t))
= \left(\zeta_i^{(*)}\right)^{+}
= \zeta_i^{(*)},
\]
since $\zeta_i^{(*)}>0$. Because this limit holds for every realization of $E_i(t)$, it follows that
\[
\mathbb{P}\!\left(\underset{l_i(\ell)\rightarrow 0}{\lim}\, \zeta_i(t+1) = \zeta_i^{(*)}\right)=1.
\]
\end{proof}


\subsection*{\textbf{Proposition \ref{prop:expect_v_long_memory}}}

\begin{proof}
Let $L_i(t)=\{l_i(\ell)\}_{\ell=1}^{t}$. Throughout the proof we work \emph{conditionally on} $L_i(t)$ (equivalently, on the sequence $\{z_i(\ell)\}_{\ell\le t}$, since $z_i(\ell)$ is a deterministic function of $l_i(\ell)$). In particular, the quantities
\[
\tilde z_i^{(1)}(t)=\sum_{\ell=1}^{t-1} z_i(\ell)
\quad\text{and}\quad
\tilde z_i^{(2)}(t)=\sum_{\ell=1}^{t-1} z_i(\ell)^2
\]
are non-random given $L_i(t)$ (they depend only on past input choices), and hence so are
$\overline{v}_i(L_i(t))$, $\overline{\varphi}_i(L_i(t))$, and $b_i(t)$ as defined below. Under the log-output specification, the only source of randomness in $\zeta_i(t+1)$ conditional on $L_i(t)$ is the Gaussian term $\varepsilon\sim\mathcal{N}(0,1)$. For a given $i\in\mathcal{N}$, define
\[
b_i(t)=-\dfrac{\overline{v}_i(L_i(t))}{\overline{\varphi}_i(L_i(t))}
\quad\mbox{ and }\quad
\tilde{\zeta}_i(t)=\overline{v}_i(L_i(t))+\overline{\varphi}_i(L_i(t))\,\varepsilon,
\]
where $\varepsilon \sim \mathcal{N}(0,1)$. We distinguish two cases.

\begin{itemize}
    \item[-] Case 1. If $\sum_{\ell=1}^{t} z_i(\ell) > 0$, we have
\begin{equation}\label{eq:Exp_U}
\begin{array}{ll}
\mathbb{E}\!\left[\max(0,\tilde{\zeta}_i(t))\right]
&= \mathbb{E}\!\left[(\overline{v}_i(L_i(t)) + \overline{\varphi}_i(L_i(t)) \varepsilon)^{+}\right]\\[0.5cm]
& \displaystyle =  \overline{v}_i(L_i(t)) \int_{b_i(t)}^{\infty}  \frac{1}{\sqrt{2\pi}}e^{-\varepsilon^2/2} d \varepsilon
~+~ \overline{\varphi}_i(L_i(t))\int_{b_i(t)}^{\infty}  \varepsilon \frac{e^{-\varepsilon^2/2} }{\sqrt{2\pi}}d \varepsilon \\[0.5cm]
& \displaystyle =  \overline{v}_i(L_i(t))(1-F(b_i(t)))
~+~ \overline{\varphi}_i(L_i(t))\int_{b_i(t)}^{\infty}   \varepsilon \frac{e^{-\varepsilon^2/2}}{\sqrt{2\pi}} d \varepsilon.
\end{array}
\end{equation}
To compute the latter integral, we define
$$
B(b,m) = \int_{b}^{\infty}   \varepsilon^m \frac{1}{\sqrt{2\pi}}e^{-\varepsilon^2/2} d \varepsilon ,
$$
and note that using integration by parts
$$
\begin{array}{ll}
B(b,m-2)
& = \displaystyle \dfrac{1}{m-1}\left( \frac{1}{\sqrt{2\pi}} \lim_{x \rightarrow \infty} \left[ x^{m-1} e^{-x^2/2} \right]
- b^{m-1} f(b)  + B(b,m) \right)  \\[0.5cm]
& = \displaystyle \dfrac{1}{m-1}\left(B(b,m) - b^{m-1} f(b) \right),
\end{array}
$$
where the last equality comes from the fact that
$$
0 \leq \lim_{x \rightarrow \infty} \left[ x^{m-1} e^{-x^2/2} \right]
= \lim_{x \rightarrow \infty} \left[ x e^{-\frac{x^2}{2(m-1)}} \right]^{m-1}
\leq \lim_{x \rightarrow \infty} \left[ e^{-\frac{x^2}{2(m-1)} + x } \right]^{m-1} = 0.
$$
We obtain the recurrence relation $B(b,m) = (m-1)B(b,m-2) + b^{m-1} f(b)$, so that $B(b,1) = f(b)$.
Substituting it back into \eqref{eq:Exp_U}, we obtain
$$
\begin{array}{lll}
\mathbb{E}[\zeta_i(t+1) ~|~ L_i(t)]
& = \overline{v}_i(L_i(t))(1-F(b_i(t)))  ~+~ \overline{\varphi}_i(L_i(t)) f(b_i(t)) \\[0.5cm]
& =  \dfrac{\Big(\zeta^{(0)}_i  + \gamma_i \zeta_i^{(*)} \tilde{z}^{(2)}_{i}(t)\Big)\left(1-F\left(b_i(t)\right)\right)
+ \gamma_i \sigma_i \tilde{z}^{(1)}_{i}(t) f\left(b_i(t)\right)}{1 + \gamma_i \tilde{z}^{(2)}_{i}(t)}.
\end{array}
$$
Applying the same procedure for the second order moment, and using the recurrence relation
$B(b,m) = (m-1)B(b,m-2) + b^{m-1} f(b)$, with the initial conditions
$$
\left\{\begin{array}{lll}
B(b,0) & = \displaystyle \int_{b}^{\infty}   \frac{1}{\sqrt{2\pi}}e^{-\varepsilon^2/2} d \varepsilon & = 1-F(b),\\[0.5cm]
B(b,1) & = f(b), & 
\end{array}\right.
$$
we have $B(b,2) = B(b,0) + b f(b)$, so that
$$
\begin{array}{ll}
\mathbb{E}[((\tilde{\zeta}_i(t))^+)^2]
& = \mathbb{E}\!\left[((\overline{v}_i(L_i(t)) + \overline{\varphi}_i(L_i(t)) \varepsilon)^{+})^2\right]\\[0.5cm]
& \displaystyle =
\left[\overline{v}_i(L_i(t))^2+\overline{\varphi}_i(L_i(t))^2\right]\left[1-F(b_i(t))\right]
+ \overline{v}_i(L_i(t)) \overline{\varphi}_i(L_i(t))f(b_i(t)).
\end{array}
$$

\item[-] Case 2. If $\sum_{\ell=1}^{t} z_i(\ell) < 0$, we have
$$
\begin{array}{ll}
   \mathbb{E}\!\left[(\tilde{\zeta}_i(t))^+\right]
   & =  \displaystyle \int_{-\infty}^{b_i(t)}  (\overline{v}_i(L_i(t)) + \overline{\varphi}_i(L_i(t)) \varepsilon)\frac{1}{\sqrt{2\pi}}e^{-\varepsilon^2/2} d \varepsilon \\[0.5cm]
   & = \displaystyle  \overline{v}_i(L_i(t))F(b_i(t))  ~+~
   \overline{\varphi}_i(L_i(t))\int_{-\infty}^{b_i(t)}   \varepsilon \frac{1}{\sqrt{2\pi}}e^{-\varepsilon^2/2} d \varepsilon.
\end{array}
$$
Since $\int_{-\infty}^{b} \varepsilon \frac{1}{\sqrt{2\pi}}e^{-\varepsilon^2/2} d \varepsilon = -f(b)$, we obtain
$$
\begin{array}{lll}
\mathbb{E}[\zeta_i(t+1) ~|~ L_i(t)]
& = \overline{v}_i(L_i(t))F(b_i(t)) ~-~ \overline{\varphi}_i(L_i(t)) f(b_i(t)) \\[0.5cm]
& =  \dfrac{\Big(\zeta^{(0)}_i  + \gamma_i \zeta_i^{(*)} \tilde{z}^{(2)}_{i}(t)\Big)F\left(b_i(t)\right)
- \gamma_i \sigma_i \tilde{z}^{(1)}_{i}(t) f\left(b_i(t)\right)}{1 + \gamma_i \tilde{z}^{(2)}_{i}(t)}.
\end{array}
$$
To compute the second order moment, define
$$
\overline{B}(b,m) = \int_{-\infty}^{b}   \varepsilon^m \frac{1}{\sqrt{2\pi}}e^{-\varepsilon^2/2} d \varepsilon.
$$
By the same arguments, $\overline{B}(b,m) = (m-1)\overline{B}(b,m-2) - b^{m-1} f(b)$, and therefore
$$
\mathbb{E}[((\tilde{\zeta}_i(t))^+)^2] =
\left[\overline{v}_i(L_i(t))^2+\overline{\varphi}_i(L_i(t))^2\right]F(b_i(t))
- \overline{v}_i(L_i(t)) \overline{\varphi}_i(L_i(t))f(b_i(t)).
$$
\end{itemize}

\end{proof}


\subsection*{\textbf{Proposition \ref{prop:mode_v_long_memory}}}

\begin{proof}
Using \eqref{eq:v_learning_n1_long_mem} and \eqref{eq:v_bar}, and conditioning on the input decisions $L_i(t)$, we have
\[
\zeta_i(t+1)=\big(\overline{v}_i(L_i(t))+\overline{\varphi}_i(L_i(t))\,\varepsilon\big)^+,
\qquad \varepsilon\sim\mathcal{N}(0,1).
\]
Hence, the conditional distribution of $\zeta_i(t+1)$ is a zero-truncated Gaussian: it has a continuous density on $(0,\infty)$ and an atom at $0$. In particular, for $z>0$,
\[
f_{\zeta_i(t+1)\mid L_i(t)}(z)
=\frac{1}{|\overline{\varphi}_i(L_i(t))|\sqrt{2\pi}}
\exp\!\left(
-\frac{\big(z-\overline{v}_i(L_i(t))\big)^2}{2|\overline{\varphi}_i(L_i(t))|^2}
\right),
\]
while the point mass at $0$ is
\[
\mathbb{P}\!\left(\zeta_i(t+1)=0 \mid L_i(t)\right)
=\mathbb{P}\!\left(\overline{v}_i(L_i(t))+\overline{\varphi}_i(L_i(t))\,\varepsilon \le 0 \mid L_i(t)\right)
=F\!\left(-\frac{\overline{v}_i(L_i(t))}{\overline{\varphi}_i(L_i(t))}\right),
\]
where $F$ denotes the cdf of standard Gaussian random variable.
Equivalently, we can write
\[
\mathbb{P}\!\left(\zeta_i(t+1)\mid L_i(t)\right)=
\begin{cases}
\displaystyle
\frac{1}{|\overline{\varphi}_i(L_i(t))|\sqrt{2\pi}}
\exp\!\left(
-\frac{\big(\zeta_i(t+1)-\overline{v}_i(L_i(t))\big)^2}{2|\overline{\varphi}_i(L_i(t))|^2}
\right),
& \mbox{if } \zeta_i(t+1)>0,\\[0.5cm]
\displaystyle
F\!\left(-\frac{\overline{v}_i(L_i(t))}{\overline{\varphi}_i(L_i(t))}\right),
& \mbox{if } \zeta_i(t+1)=0,\\[0.2cm]
0, & \mbox{otherwise.}
\end{cases}
\]
For $\zeta_i(t+1)>0$, the density above is maximized at
$\zeta_i(t+1)=\overline{v}_i(L_i(t))$ if $\overline{v}_i(L_i(t))>0$, with peak value
$\big(|\overline{\varphi}_i(L_i(t))|\sqrt{2\pi}\big)^{-1}$.
At $\zeta_i(t+1)=0$, the probability mass is
$F\!\left(-\overline{v}_i(L_i(t))/\overline{\varphi}_i(L_i(t))\right)$.
Therefore,
\[
\mathbb{M}\!\left[\zeta_i(t+1)\mid L_i(t)\right]=
\begin{cases}
\overline{v}_i(L_i(t))
& \quad \mbox{if } \displaystyle \frac{1}{|\overline{\varphi}_i(L_i(t))|\sqrt{2\pi}}
\geq
F\!\left(-\dfrac{\overline{v}_i(L_i(t))}{\overline{\varphi}_i(L_i(t))}\right),\\[0.4cm]
0,
& \quad \mbox{otherwise.}
\end{cases}
\]
\end{proof}


\subsection*{\textbf{Proposition \ref{prop:long_memory_learning}}}

\begin{proof} We distinguish the following cases: Recall the following sequence (defined in \eqref{eq:v_bar}) 
$$
    \overline{v}_i(L_i(t))  = \dfrac{ \zeta^{(0)}_i  + \gamma_i \zeta_i^{(*)} \tilde{z}^{(2)}_{i}(t)}{1 + \gamma_i \tilde{z}^{(2)}_{i}(t)}, \quad 
    \overline{\varphi}_i(L_i(t))  = \dfrac{\gamma_i  \sigma_i \tilde{z}^{(1)}_{i}(t)}{1 + \gamma_i \tilde{z}^{(2)}_{i}(t)}, \quad -\left|\dfrac{\overline{v}_i(L_i(t))}{\overline{\varphi}_i(L_i(t))}\right| = -\dfrac{ \zeta^{(0)}_i  + \gamma_i \zeta_i^{(*)} \tilde{z}^{(2)}_{i}(t)}{\gamma_i |\tilde{z}^{(1)}_{i}(t)|}
$$
We have the following cases

\begin{itemize}
    \item[(1)] If $\underset{t \rightarrow +\infty }\lim |\tilde{z}_{i}^{(1)}(t)| = \underset{t \rightarrow +\infty }\lim \tilde{z}_{i}^{(2)}(t)=+\infty$ and 
    $$
    \underset{t \rightarrow +\infty }\lim  \dfrac{ \zeta^{(0)}_i  + \gamma_i \zeta_i^{(*)} \tilde{z}^{(2)}_{i}(t)}{\gamma_i |\tilde{z}^{(1)}_{i}(t)|} = L < +\infty,
    $$
    then 
    $$
    \underset{t \rightarrow +\infty }\lim  \dfrac{ \zeta^{(0)}_i  + \gamma_i \zeta_i^{(*)} \tilde{z}^{(2)}_{i}(t)}{\gamma_i |\tilde{z}^{(1)}_{i}(t)|} = L+o(t)
    $$
    with $\underset{t \rightarrow +\infty }\lim o(t)=0$. Then,
    $$
    \begin{array}{ll}
    \underset{t \rightarrow +\infty }\lim\mathbb{E}[\zeta_i(t)] & = \underset{t \rightarrow +\infty }\lim \dfrac{\Big( \zeta^{(0)}_i  + \gamma_i \zeta_i^{(*)}\tilde{z}_{i}^{(2)}(t)\Big)(1-F(-L-o(t))+     \sigma_i\gamma_i|\tilde{z}_{i}^{(1)}(t)|f(-L-o(t))}{1+\gamma_i\tilde{z}_{i}^{(2)}(t)}\\
    & = \zeta_i^{(*)}(1-F(-L)) +
    \sigma_i f(-L) \underset{t \rightarrow +\infty }\lim\dfrac{|\tilde{z}_{i}^{(2)}(t)|}{\tilde{z}_{i}^{(1)}(t)}.
    \end{array}
    $$ 
    \item[(2)] If $\underset{t \rightarrow +\infty }\lim \tilde{z}_{i}^{(1)}(t) =L_1 < +\infty$ and $\underset{t \rightarrow +\infty }\lim \tilde{z}_{i}^{(2)}(t) = L_2< +\infty$, let us consider the following function 
    $$
    \Upsilon(x,y)=\dfrac{\Big( \zeta^{(0)}_i  + \gamma_i \zeta_i^{(*)}y\Big)G_i(x,y)+ \sigma_i\gamma_i x g_i(x,y)}{1+\gamma_i y}.
    $$
    Since $\Upsilon$ is continuous in $\mathbb{R}_+^{*}\times\mathbb{R}_+$, then we need to distinguish the two following cases:
    \begin{itemize}
    \item[(2.1)] If $L_1\neq0$, then by continuity of $\Upsilon$, we obtain that $$
        \underset{t \rightarrow +\infty }\lim \mathbb{E}[\zeta_i(t)] =\dfrac{\Big( \zeta^{(0)}_i  + \gamma_i \zeta_i^{(*)}L_2\Big)G_i(L_1,L_2)+\sigma_i\gamma_i L_1 g_i(L_1,L_2)}{1+\gamma_i L_2}.
    $$
    \item[(2.2)] If $L_1=0$, then $g_i(\tilde{z}_{i}^{(1)}(t),\tilde{z}_{i}^{(2)}(t))=o_1(t)$ and  $G_i(\tilde{z}_{i}^{(1)}(t),\tilde{z}_{i}^{(2)}(t))=1+o_2(t)$ with $\underset{t \rightarrow +\infty }\lim o_1(t)=\underset{t \rightarrow +\infty }\lim o_2(t)=0$. Then,
    $$
    \mathbb{E}[\zeta_i(t)] = \dfrac{\Big( \zeta^{(0)}_i  + \gamma_i \zeta_i^{(*)}\tilde{z}_{i}^{(1)}(t)\Big)(1+o_2(t)) + \sigma_i \gamma_i |\tilde{z}_{i}^{(1)}(t)|o_1(t)}{1+\gamma_i\tilde{z}_{i}^{(2)}(t)}.
    $$
    Therefore, we obtain
    $$
        \begin{array}{ll}
        \underset{t \rightarrow +\infty }\lim \mathbb{E}[\zeta_i(t)] & =
        \underset{t \rightarrow +\infty }\lim \dfrac{\Big( \zeta^{(0)}_i  + \gamma_i \zeta_i^{(*)}\tilde{z}_{i}^{(2)}(t)\Big)             (1+o_2(t))+\sigma_i\gamma_i|\tilde{z}_{i}^{(1)}(t)|o_1(t)}{1+\gamma_i\tilde{z}_{i}^{(2)}(t)}\\
        &=\underset{t \rightarrow +\infty }\lim \dfrac{\Big( \zeta^{(0)}_i  + \gamma_i \zeta_i^{(*)}\tilde{z}_{i}^{(2)}(t) \Big)                (1+o_2(t))}{1+\gamma_i\tilde{z}_{i}^{(2)}(t)}\\  &=\dfrac{\zeta^{(0)}_i  + \gamma_i \zeta_i^{(*)}L_2}{1+\gamma_i L_2}. \end{array}
        $$
        \end{itemize}
        \item[(3)] If $\underset{t \rightarrow +\infty }\lim \tilde{z}_{i}^{(1)}(t) =L_1$ and $\underset{t \rightarrow +\infty }\lim \tilde{z}_{i}^{(2)}(t) =\infty$, then $g_i(\tilde{z}_{i}^{(1)}(t),\tilde{z}_{i}^{(2)}(t))=o_1(t)$ and $G_i(\tilde{z}_{i}^{(1)}(t),\tilde{z}_{i}^{(2)}(t))=1+o_2(t)$ with $\underset{t \rightarrow +\infty }\lim o_1(t) = \underset{t \rightarrow +\infty }\lim o_2(t)=0$. Then,
        $$
        \mathbb{E}[\zeta_i(t)] = \dfrac{\Big( \zeta^{(0)}_i  + \gamma_i \zeta_i^{(*)}\tilde{z}_{i}^{(2)}(t)\Big)(1+o_2(t))+ \sigma_i\gamma_i|\tilde{z}_{i}^{(1)}(t)|o_1(t)}{1+\gamma_i\tilde{z}_{i}^{(1)}(t)}.
        $$
        Therefore,
        $$
        \begin{array}{ll}
        \underset{t \rightarrow +\infty }\lim \mathbb{E}[\zeta_i(t)] & = \underset{t \rightarrow +\infty }\lim \dfrac{\Big( \zeta^{(0)}_i  + \gamma_i \zeta_i^{(*)}\tilde{z}_{i}^{(2)}(t)\Big)                (1+o_2(t))+\sigma_i\gamma_i|\tilde{z}_{i}^{(1)}(t)|o_{(1)}(t)}{1+\gamma_i\tilde{z}_{i}^{(2)}(t)}\\
        &= \underset{t \rightarrow +\infty }\lim \dfrac{\Big( \zeta^{(0)}_i  + \gamma_i \zeta_i^{(*)}\tilde{z}_{i}^{(2)}(t)\Big)}{1+\gamma_i\tilde{z}_{i}^{(2)}(t)}\\
        &=\zeta_i^{(*)}.
        \end{array}
        $$
    \item[(4)] If $\underset{t \rightarrow +\infty }\lim \tilde{z}_{i}^{(1)}(t)=\infty$ and $\underset{t \rightarrow +\infty }\lim \tilde{z}_{i}^{(2)}(t) =L_2$, then $g_i(\tilde{z}_{i}^{(1)}(t),\tilde{z}_{i}^{(2)}(t))=\dfrac{1}{\sqrt{2\pi}}+o_1(t)$ and $G_i(\tilde{z}_{i}^{(1)}(t),\tilde{z}_{i}^{(2)}(t)) = \dfrac{1}{2}+o_2(t)$ with $\underset{t \rightarrow +\infty }\lim o_1(t)= \underset{t \rightarrow +\infty }\lim o_2(t)=0$. Then,
    $$
    \mathbb{E}[\zeta_i(t)] = \dfrac{\Big( \zeta^{(0)}_i  + \gamma_i \zeta_i^{(*)}\tilde{z}_{i}^{(2)}(t)\Big)(\dfrac{1}{2}+o_2(t)) + \sigma_i \gamma_i |\tilde{z}_{i}^{(1)}(t)|\Big(\dfrac{1}{\sqrt{2\pi}}+o_1(t)\Big)}{1+\gamma_i\tilde{z}_{i}^{(2)}(t)}.
    $$
    Therefore,
    $$
    \begin{array}{ll}
    \underset{t \rightarrow +\infty }\lim \mathbb{E}[\zeta_i(t)]&=\dfrac{\Big( \zeta^{(0)}_i  + \gamma_i \zeta_i^{(*)}\tilde{z}_{i}^{(2)}(t)\Big)(\dfrac{1}{2}+o_2(t)) + \sigma_i \gamma_i|\tilde{z}_{i}^{(1)}(t)|\Big(\dfrac{1}{\sqrt{2\pi}}+o_1(t)\Big)}{1+\gamma_i\tilde{z}_{i}^{(2)}(t)}\\
    &=\dfrac{\sigma_i\gamma_i|\tilde{z}_{i}^{(1)}(t)|\dfrac{1}{\sqrt{2\pi}}}{1+\gamma_i\tilde{z}_{i}^{(2)}(t)}\\
    &=
    \underset{t \rightarrow +\infty }\lim |\tilde{z}_{i}^{(1)}(t)|\\
    &=+\infty.
    \end{array}
    $$
    \item[(5)] If $\underset{t \rightarrow +\infty }\lim |\tilde{z}_{i}^{(1)}(t)| = \underset{t \rightarrow +\infty }\lim \tilde{z}_{i}^{(2)}(t)=\underset{t \rightarrow +\infty }\lim\dfrac{\zeta^{(0)}_i  + \gamma_i \zeta_i^{(*)} \tilde{z}^{(2)}_{i}(t)}{\gamma_i \sigma_i \tilde{z}^{(1)}_{i}(t)} =  +\infty$, then $g_i(\tilde{z}_{i}^{(1)}(t),\tilde{z}_{i}^{(2)}(t))=o_1(t)$ and
    $G_i(\tilde{z}_{i}^{(1)}(t),\tilde{z}_{i}^{(2)}(t))=1+o_2(t)$ with $\underset{t \rightarrow +\infty }\lim o_1(t) = \underset{t \rightarrow +\infty }\lim o_2(t)=0$. Hence,
    $$
    \mathbb{E}[\zeta_i(t)] = \dfrac{\Big( \zeta^{(0)}_i  + \gamma_i \zeta_i^{(*)}\tilde{z}_{i}^{(2)}(t)\Big)(1+o_2(t))+\sigma_i\gamma_i|\tilde{z}_{i}^{(1)}(t)|o_1(t)}{1+\gamma_i\tilde{z}_{i}^{(2)}(t)}.
    $$
    Therefore,
    $$
    \begin{array}{ll}
    \underset{t \rightarrow +\infty }\lim \mathbb{E}[\zeta_i(t)]&=
    \underset{t \rightarrow +\infty }\lim \dfrac{\Big( \zeta^{(0)}_i  + \gamma_i \zeta_i^{(*)}\tilde{z}_{i}^{(2)}(t)\Big)(1+o_2(t))+\sigma_i \gamma_i|\tilde{z}_{i}^{(1)}(t)|o_1(t)}{1+\gamma_i\tilde{z}_{i}^{(2)}(t)}\\
    &=
    \zeta_i^{(*)}+\underset{t \rightarrow +\infty }\lim \dfrac{|\tilde{z}_{i}^{(1)}(t)|o_1(t)}{\tilde{z}_{i}^{(2)}(t)}
    \end{array}
    $$
    Since $\underset{t \rightarrow +\infty }\lim\dfrac{\zeta^{(0)}_i  + \gamma_i \zeta_i^{(*)} \tilde{z}^{(2)}_{i}(t)}{\gamma_i \sigma_i \tilde{z}^{(1)}_{i}(t)} =  +\infty$, then $\underset{t \rightarrow +\infty }\lim \tilde{z}^{(1)}_{i}(t) = \underset{t \rightarrow +\infty }\lim |\tilde{z}^{(1)}_{i}(t)| \geq 0$. Therefore,
    $$
    \underset{t \rightarrow +\infty }\lim\dfrac{\zeta^{(0)}_i  + \gamma_i \zeta_i^{(*)} \tilde{z}^{(2)}_{i}(t)}{\gamma_i \sigma_i \tilde{z}^{(1)}_{i}(t)} =  +\infty \quad \implies \quad \underset{t \rightarrow +\infty }\lim \dfrac{|\tilde{z}_{i}^{(1)}(t)|}{\tilde{z}_{i}^{(2)}(t)} < +\infty.
    $$
    Hence, 
    $$
    \underset{t \rightarrow +\infty }\lim \dfrac{|\tilde{z}_{i}^{(1)}(t)|o_1(t)}{\tilde{z}_{i}^{(2)}(t)} = 0,
    $$
    and 
    $$
    \underset{t \rightarrow +\infty }\lim \mathbb{E}[\zeta_i(t)] = \zeta_i^{(*)}.
    $$
\end{itemize}
\end{proof}


\subsection*{\textbf{Proposition \ref{prop:demography_long_mem}}}

\begin{proof}

Let $w_t*(\delta(t) - l_0(t))$ denote the unique solution of $\delta(t) - l_0(t) = L_t(w^*(t))$, whose existence is established in Lemma \ref{lemma:L}. To show that there exists a demographic path $\{\delta(t)\}_t$ for which 
$$
\underset{t \rightarrow +\infty }\lim |\tilde{z}_{i}^{(1)}(t)| = \underset{t \rightarrow +\infty }\lim \tilde{z}_{i}^{(2)}(t)=\underset{t \rightarrow +\infty }\lim\dfrac{\zeta^{(0)}_i  + \gamma_i \zeta_i^{(*)} \tilde{z}^{(2)}_{i}(t)}{\gamma_i \sigma_i \tilde{z}^{(1)}_{i}(t)} =  +\infty,
$$
we express the labor demand as    
$l_i(t) = \chi_{i,t} (w_t*(d(t)))^{\kappa_{i,t}}$, where $d(t) = \delta(t) - l_0(t)$ and
$$
\kappa_{i,t} = - \dfrac{1 }{1 - \hat{\zeta}_{i}(t)\alpha } \quad \mbox{ and } \quad 
\chi_{i,t}  =  \left( \mathbb{E}\left[ \eta_{i}(t)^{\alpha}  \right] \hat{\zeta}_{i}(t) \right)^{\dfrac{1 }{1 - \hat{\zeta}_{i}(t)\alpha } } .
$$
Since $\hat{\zeta}_i(t) < 1$ and $\alpha < 1$, then $\hat{\zeta}_{i}(t)\alpha < 1$ and $\kappa_{i,t} < -1$. Let us define 
    \begin{equation}\label{eq:Z_fun}
            \begin{array}{lll}
        Z_{i,t}^{(2)}(d(t) ) &\displaystyle  = \tilde{z}_{i}^{(2)}(t) & \displaystyle =  \left[\ln\left(\xi_{i,t} w^*(t)(d(t) )^{\kappa_{i,t}} \right)\right]^2 + \sum_{\ell = 1}^{t-1} \left[\ln\left(\xi_{i,t} w_{\ell}^*( d(\ell) )^{\kappa_{i,t}} \right)\right]^2,\\[0.5cm]
        Z_{i,t}^{(1)}( d(t) ) & = \tilde{z}_{i}^{(1)}(t) & \displaystyle = \ln\left(\xi_{i,t} w^*(t)(d(t) )^{\kappa_{i,t}} \right) + \sum_{\ell = 1}^{t-1} \ln\left(\xi_{i,t} w_{\ell}^*(d(\ell) )^{\kappa_{i,t}} \right).
    \end{array}
    \end{equation}

    We now show conditions for the monotonicity of $Z_{i,t}^{(2)}$ and $Z_{i,t}^{(1)}$ for all $\delta(t)$ greater than a fixed constant $\Bar{\delta}(t)$. 

    \paragraph{Derivative of $Z_{i,t}^{(2)}$. } To characterize the derivative of $Z_{i,t}^{(2)}$ with respect to $\delta(t)$, we consider that $w^*(t)(d(t))$ is an implicit function defined by the equation $L_{t}(w^*(t)(\delta(t) - l_0)) = \delta(t) - l_0$. We have    
    $$
    \begin{array}{lll}
    \frac{d }{d ~\delta(t)} Z_{i,t}^{(2)} & \displaystyle  = \frac{d}{d ~\delta(t)} \left[\ln\left(\xi_{i,t} \cdot w^*(t)(\delta(t) - l_0(t))^{\kappa_{i,t}}\right)\right]^2\\[0.5cm]
    & \displaystyle = 2 \ln\left(\xi_{i,t} \cdot w^*(t)(\delta(t) -l_0(t))^{\kappa_{i,t}}\right) \cdot \frac{d}{d ~\delta(t)} \left[\ln\left(\xi_{i,t} \cdot w^*(t)(\delta(t) -l_0(t))^{\kappa_{i,t}}\right)\right]\\[0.5cm]
    & \displaystyle =  \frac{2 \ln\left(\xi_{i,t} \cdot w^*(t)(\delta(t) -l_0(t))^{\kappa_{i,t}}\right)}{\xi_{i,t} \cdot w^*(t)(\delta(t) -l_0(t))^{\kappa_{i,t}}} \cdot \xi_{i,t} \cdot \kappa_{i,t} \cdot w^*(t)(\delta(t) -l_0(t))^{\kappa_{i,t}-1} \cdot \left(\frac{d }{d ~\delta(t)} w^*(t)( \delta(t) ) \right)\\[0.5cm]
    & = 2 \ln\left(\xi_{i,t} \cdot w^*(t)(\delta(t) -l_0)^{\kappa_{i,t}}\right) \cdot \frac{\kappa_{i,t}}{w^*(t)(\delta(t) - l_0(t) )} \cdot \left( \frac{d }{d ~ \delta(t) } w^*(t)(\delta(t) -l_0(t)) \right) .
    \end{array}
    $$
    Now, we need to find $\frac{d }{d ~\delta(t)} w^*(t)$. Since $w^*(t)(\delta(t) - l_0(t))$ is defined implicitly by the equation $L_{t}(w^*(t)(\delta- l_0(t))) = \delta - l_0(t)$, let's differentiate both sides with respect to $\delta_0$:
    $$
    \frac{d}{d ~\delta(t)} L_{t}(w^*(t)(\delta(t)- l_0(t))) = \frac{d}{d ~\delta(t)}(\delta(t) - l_0(t)).
    $$
    By Lemma \ref{lemma:L}, we have
    $$
    \displaystyle \sum_{i=1}^{n} \xi_{i,t} \cdot \kappa_{i,t} \cdot w^*(t)(\delta(t) - l_0(t))^{\kappa_{i,t} - 1} \cdot \dfrac{d}{d ~\delta(t)} w^*(t)(\delta(t) - l_0(t)) = 1,
    $$   
    Solving for $\frac{d}{d ~\delta} w^*(t)(\delta(t)- l_0(t))$ and noticing that $\kappa_{i,t} < -1$, we have:
    \begin{equation}\label{eq:d_w_d_l}
        \dfrac{d}{d ~\delta(t)} w^*(t)(\delta(t) - l_0(t)) = \dfrac{1}{\sum_{i=1}^{n} \xi_{i,t} \cdot \kappa_{i,t} \cdot w^*(t)(\delta(t) - l_0(t))^{\kappa_{i,t} - 1} } < 0.
    \end{equation}
    Now we can substitute $\frac{d}{d ~\delta(t)} w^*(t)(\delta(t)- l_0(t)) $ back into $ \frac{d~ Z_{i,t}^{(2)}}{d ~ \delta(t)}$:
    $$
    \frac{d }{d ~ \delta(t)} Z_{i,t}^{(2)} = 2 \ln\left(\xi_{i,t} \cdot w^*(t)(\delta(t) - l_0(t))^{\kappa_{i,t}}\right) \cdot \left(\dfrac{\kappa_{i,t}}{\sum_{j=1}^{n} \xi_{j,t} \cdot \kappa_{j,t} \cdot w^*(t)(\delta(t) - l_0(t))^{\kappa_{j,t}} } \right).
    $$
    We can establish sufficient conditions for the sign of this derivative: 
    $$
    \left\{\begin{array}{ll}
        \frac{d }{d ~ \delta(t)} Z_{i,t}^{(2)} > 0 &  \mbox{ if } \xi_{i,t} w^*(t)(\delta(t) - l_0(t))^{\kappa_{i,t}} > 1, \\
        \frac{d }{d ~ \delta(t)} Z_{i,t}^{(2)} < 0 &  \mbox{ if } \xi_{i,t} w^*(t)(\delta(t) - l_0(t))^{\kappa_{i,t}} < 1. \\
    \end{array}\right.
    $$
    Moreover, since $w^*(t)(\delta- l_0(t))$ is continuous (see Lemma \ref{lemma:L}), $\underset{d \rightarrow 0}{\lim} w^*(t)(d) = + \infty$ and $\underset{d \rightarrow + \infty}{\lim} w^*(t)(d) = 0$. Thus, there exists $\Bar{\delta}(t)$, such that for all $\delta(t) > \Bar{\delta}(t)$,
    $$
    \xi_{i,t} w^*(t)(\delta(t) - l_0(t))^{\kappa_{i,t}} > 1 \quad \mbox{ and } \quad \frac{d }{d ~\delta(t)} Z_{i,t}^{(2)} > 0.
    $$
     
    \paragraph{Derivative of $Z_{i,t}^{(1)}$. } We follow an analogous procedure and obtain: 
    $$
    \begin{array}{ll}
        \frac{dZ_{i,t}^{(1)}(l_0)}{d ~\delta(t)} & = \displaystyle \frac{d}{d ~\delta(t)} \left[ \ln\left(\xi_{i,t} w^*(t)(\delta(t) - l_0(t) )^{\kappa_{i,t}} \right) + \sum_{\ell = 1}^{t-1} \ln\left(\xi_{i,t} w^*(t)( l_0(\ell) )^{\kappa_{i,t}} \right) \right]\\
        & = \displaystyle \frac{d}{d ~\delta(t)}  \ln\left(\xi_{i,t} w^*(t)(\delta(t) - l_0(t) )^{\kappa_{i,t}} \right) \\[0.5cm]
        & = \displaystyle \frac{d}{d ~\delta(t)}  \kappa_{i,t} \ln\left( w^*(t)(\delta(t) - l_0(t) )\right) \\[0.5cm] 
        & \displaystyle = \frac{\kappa_{i,t}}{w^*(t)(\delta(t) - l_0(t) )}  \left(\frac{d ~ w^*(t)(\delta- l_0(t))}{d ~\delta(t)} \right).
    \end{array}
    $$
Using $\frac{d ~w^*(t)(\delta(t) - l_0(t) )}{d ~\delta(t)}$ in \eqref{eq:d_w_d_l}, we have
    $$
    \frac{d }{d ~\delta(t)} Z_{i,t}^{(1)} = \dfrac{\kappa_{i,t}}{\sum_{j=1}^{n} \xi_{j,t} \cdot \kappa_{j,t} \cdot w^*(t)(\delta(t) - l_0(t))^{\kappa_{j,t}} } = \dfrac{1}{2 \ln\left(\xi_{i,t} \cdot w^*(t)(\delta(t) - l_0(t))^{\kappa_{i,t}}\right)} \frac{d }{d ~\delta(t)} Z_{i,t}^{(2)}
    $$
    Again, there exists $\Bar{\delta}(t)$, such that for all $\delta(t) > \Bar{\delta}(t)$,
    $$
    \xi_{i,t} w^*(t)(\delta(t) - l_0(t))^{\kappa_{i,t}} > 1 \quad \mbox{ and } \quad  \frac{d }{d ~\delta(t)} Z_{i,t}^{(1)} > 0 
    $$

    \paragraph{Divergence. }  Since $\underset{d \rightarrow + \infty}{\lim} w^*(t)(d) = + \infty$, then
    $$
     \begin{array}{ll}
    \underset{d(t) \rightarrow + \infty}{\lim} ~ Z_{i,t}^{(2)}(d(t)) =  \underset{d(t) \rightarrow +\infty }\lim  \left[\ln\left(\xi_{i,t} w^*(t)(d(t) )^{\kappa_{i,t}} \right)\right]^2 + \sum_{\ell = 1}^{t-1} \left[\ln\left(\xi_{i,t} w_{\ell}^*( d(\ell) )^{\kappa_{i,t}} \right)\right]^2 = +\infty.\\[0.5cm]
    \underset{d(t) \rightarrow + \infty}{\lim} ~ Z_{i,t}^{(1)}(d(t)) =  \underset{d(t) \rightarrow +\infty }\lim  \ln\left(\xi_{i,t} w^*(t)(d(t) )^{\kappa_{i,t}} \right) + \sum_{\ell = 1}^{t-1} \ln\left(\xi_{i,t} w_{\ell}^*( d(\ell) )^{\kappa_{i,t}} \right) = +\infty.
    \end{array}
    $$
    Therefore, there exists a monotonic sequence $\{\delta(t)\}_t$, with $\delta(t)> \delta(t-1)$, such that the following uniform divergence holds:\footnote{Let $\{f_n\}_{n=1}^{\infty}$ be a sequence of real-valued functions defined on a domain $D \subseteq \mathbb{R}$. We say that the sequence $\{f_n(x)\}$ diverges uniformly on $D$ if for every $M > 0$, there exists a positive integer $N$ such that for all $n \geq N$ and for all $x \in D$, we have $|f_n(x)| > M$. In other words, beyond some index $N$, the magnitude of the functions $f_n(x)$ exceeds any given bound $M$ uniformly over the domain $D$.}
    $$
     \underset{t \rightarrow +\infty }\lim |\tilde{z}_{i}^{(1)}(t)| = \underset{t \rightarrow + \infty}{\lim} ~ Z_{i,t}^{(1)}(\delta(t)- l_0(t)) = +\infty \quad \mbox{and} \quad  \underset{t \rightarrow +\infty }\lim |\tilde{z}_{i}^{(2)}(t)| = \underset{t \rightarrow + \infty}{\lim} ~ Z_{i,t}^{(2)}(\delta(t)- l_0(t)) = +\infty.
    $$
    Finally, we note that, for all $\delta(t) > \Bar{\delta}(t)$,
    $$
    \frac{d }{d ~\delta(t)} Z_{i,t}^{(2)} > \frac{d }{d ~\delta(t)} Z_{i,t}^{(1)} \quad \mbox{ and } \quad Z_{i,t}^{(2)} > Z_{i,t}^{(1)}.
    $$
    By applying L'Hopital's rule, we have
    $$
    \begin{array}{ll}
    \underset{t \rightarrow +\infty }\lim\dfrac{\zeta^{(0)}_i  + \gamma_i \zeta_i^{(*)} \tilde{z}^{(2)}_{i}(t)}{\gamma_i \sigma_i \tilde{z}^{(1)}_{i}(t)} & =    \underset{t \rightarrow +\infty }\lim\dfrac{\gamma_i \zeta_i^{(*)} \frac{d }{d ~\delta(t)} Z_{i,t}^{(2)} }{\gamma_i \sigma_i \frac{d }{d ~\delta(t)} Z_{i,t}^{(1)} }\\
    & =    \underset{t \rightarrow +\infty }\lim\dfrac{\gamma_i \zeta_i^{(*)} \frac{d }{d ~\delta(t)} Z_{i,t}^{(2)} }{\gamma_i \sigma_i \frac{d }{d ~\delta(t)} Z_{i,t}^{(1)} } \\
    & =     \underset{t \rightarrow +\infty }\lim 2 \dfrac{\zeta_i^{(*)} }{\sigma_i  } \ln\left(\xi_{i,t} \cdot w^*(t)(\delta(t) - l_0(t))^{\kappa_{i,t}}\right)\\
    & = +\infty.
        \end{array}
    $$

\end{proof}


\subsection*{\textbf{Proposition \ref{prop:convergence_mode_long_memory}}}

\begin{proof}
\Rew{We invoke Proposition \ref{prop:mode_v_long_memory} to characterize the mode of the MAP estimator:
$$
\displaystyle \mathbb{M}[\zeta_i(t) ~|~ L_i(t) ] = \underset{\zeta_i(t+1)}{\mbox{ argmax }} ~ \mathbb{P}( \zeta_i(t+1) ~|~ L_i(t) ) = 
\begin{cases}
  \overline{v}_i(L_i(t)) & \quad \mbox{if } \displaystyle \dfrac{1}{|\overline{\varphi}_i(L_i(t))|\sqrt{ 2 \pi}} \geq  F\left(-\left|\dfrac{\overline{v}_i(L_i(t))}{\overline{\varphi}_i(L_i(t))}\right|\right),\\[0.4cm]
  0 & \quad \mbox{otherwise.}
\end{cases}
$$
We construct the ordered set $\overline{\Psi}_i$ (as defined in \eqref{eq:Psi_mode}) as the collection of time periods for which $\mathbb{M}[\zeta_i(t) ~|~ L_i(t) ] > 0$. This allows constructing a sub-sequence by indexing the elements of $\overline{\Psi}_i$ by $\varrho(t)$, for $t \in \mathbb{N}$. Hence,
$$
\left\{ ~\mathbb{M}[\zeta_i(t) ~|~ L_i(t) ] ~ \right\}_{t \in \overline{\Psi}_i} \equiv \left\{ ~\overline{v}_i(\varrho(t)) ~\right\}_{t\geq0}
$$
Focusing on the sub-sequence $\{\overline{v}_i(\varrho(t))\}_{t\geq0}$, we distinguish the following cases:}
\begin{itemize}
    \item[(1)] If $\overline{\Psi}_i$ is a finite set, then there is a $t_0\in \nset$ such that for each $t\geq t_0$, we have $\mathbb{M}[\zeta_i(t) ~|~ L_i(t) ]=0$ and therefore $\left\{ \mathbb{M}[\zeta_i(t) ~|~ L_i(t) ] \right\}_{t}$  converges to 0.
    
    \item[(2)] Assume that $\overline{\Psi}_i$ is an infinite set. Then, without loss of generality, let $\overline{\Psi}_i\equiv \nset$ (in other words, if $\overline{\Psi}_i\not\equiv\nset$, then there exists a bijection $\varphi:\nset\rightarrow\varphi(\nset)=\overline{\Psi}_i$). 
    
    \begin{itemize}
        \item[(2.1)] If $\underset{t\rightarrow\infty}\lim \tilde{z}^{(2)}_{i}(t) = +\infty$, then $\mathbb{M}[\zeta_i(t) ~|~ L_i(t) ] = \dfrac{ \zeta^{(0)}_i  + \gamma_i \zeta_i^{(*)} \tilde{z}^{(2)}_{i}(t)}{1 + \gamma_i \tilde{z}^{(2)}_{i}(t)}$. As a consequence, we obtain
          $$
          \begin{array}{ll}  
          \underset{t \rightarrow +\infty }\lim
          \mathbb{M}[\zeta_i(t) ~|~ L_i(t) ] & = 
          \underset{t \rightarrow +\infty }\lim
          \dfrac{ \zeta^{(0)}_i  + \gamma_i \zeta_i^{(*)} \tilde{z}^{(2)}_{i}(t)}{1 + \gamma_i \tilde{z}^{(2)}_{i}(t)}\\
          &=\underset{t \rightarrow +\infty }\lim
          \dfrac{ \gamma_i \zeta_i^{(*)} \tilde{z}^{(2)}_{i}(t)}{1 + \gamma \tilde{z}^{(2)}_{i}(t)}\\
          &=\zeta_i^{(*)}.
          \end{array}
          $$
        \item[(2.2)] If $\underset{t\rightarrow\infty}\lim \tilde{z}^{(2)}_{i}(t) =L_2 < +\infty$, then there exists a function $o(t)$ such that $\tilde{z}^{(2)}_{i}(t) = L_2 + o(t)$ with $\underset{t \rightarrow +\infty }\lim o(t)=0$. Then we obtain that 
        $$
          \begin{array}{ll} 
          \underset{t \rightarrow +\infty }\lim
          \mathbb{M}[\zeta_i(t) ~|~ L_i(t) ] &= 
          \underset{t \rightarrow +\infty }\lim
          \dfrac{ \zeta^{(0)}_i  + \gamma_i \zeta_i^{(*)} (L_2+o(t))}{1 + \gamma_i (L_2+o(t))}\\[0.2cm]
         &= \dfrac{ \zeta^{(0)}_i+\gamma_i \zeta_i^{(*)} L_2}{1 + \gamma_i L_2}.
        \end{array}
        $$
    \end{itemize}
\end{itemize}    
\end{proof}


\subsection*{\textbf{Proposition \ref{prop:v_learning_short_mem} and Proposition \ref{prop:expect_v_short_memory}}}

\begin{proof}

Proposition \ref{prop:v_learning_short_mem} and Proposition \ref{prop:expect_v_short_memory} are obtained by replacing $\zeta_i^{(0)}$ with $\zeta_i(t)$ in Proposition \ref{prop:v_learning_long_mem}, and $\tilde{z}^{(1)}_{i}(t)$ and $\tilde{z}^{(2)}_{i}(t)$ with $z_i(t)$ and $z_i(t)^2$, respectively, in Proposition \ref{prop:expect_v_long_memory}. To lighten the exposition, we avoid reporting whole expressions extensively.

\end{proof}


\subsection*{\textbf{Proposition \ref{prop:mode_v_short_memory}}}

\begin{proof}
We consider the probability density function of $\zeta_i(t)$, conditioned on the input decisions
$$
\mathbb{P}( \zeta_i(t+1) ~|~ L_i(t) ) = 
\begin{cases}
  \dfrac{\exp\left(-\dfrac{(\zeta_i(t)-\overline{\overline{v}}_i(L_i(t)))^2}{2|\overline{\overline{\varphi}}_i(L_i(t))|^2 }\right)}{|\overline{\overline{\varphi}}_i(L_i(t))| \sqrt{2 \pi}} , & \quad \mbox{if } \displaystyle \zeta_i(t) > 0 \\[0.3cm]
  F\left(-\overline{\overline{v}}_i(L_i(t))/\overline{\overline{\varphi}}_i(L_i(t)) \right), & \quad \mbox{if } \displaystyle \zeta_i(t) = 0 \\[0.2cm]
  0, & \quad \mbox{otherwise.}
\end{cases}
$$
Note that for $\zeta_i(t) > 0$, if $\overline{\overline{v}}_i(L_i(t))) > 0$, then $\mathbb{P}(\zeta_i(t) ~|~ L_i(t) )$ is maximized when $\zeta_i(t) = \overline{\overline{v}}_i(L_i(t)))$, with $\mathbb{P}( \overline{\overline{v}}_i(L_i(t))) ~|~ L_i(t) ) = \dfrac{1}{|\overline{\overline{\varphi}}_i(L_i(t))| \sqrt{2 \pi}}$. Likewise,  $\mathbb{P}( ~ 0 ~|~ L_i(t) ) = F\left(-\overline{\overline{v}}_i(L_i(t))/\overline{\overline{\varphi}}_i(L_i(t)) \right)$. Therefore,
$$
\max_{\zeta_i(t)} ~ \mathbb{P}( \zeta_i(t) ~|~ L_i(t) ) = 
\begin{cases}
  \overline{\overline{v}}_i(L_i(t)) & \quad \mbox{if } \displaystyle \dfrac{1}{|\overline{\overline{\varphi}}_i(L_i(t))|\sqrt{ 2 \pi}} \geq  F\left(-\overline{\overline{v}}_i(L_i(t))/\overline{\overline{\varphi}}_i(L_i(t)) \right),\\[0.2cm]
  0 & \quad \mbox{otherwise.}
\end{cases}
$$
\end{proof}


\subsection*{\textbf{Proposition \ref{prop:short_memory_learning}}}

\begin{proof} Let us prove that the following sequence
$$
\mathbb{E}_{\varepsilon_{i}(t)}[\zeta_i(t+1)] = \dfrac{ \mathbb{E}_{\varepsilon_{i}(t)}[\zeta_i(t)]  + \gamma_i \zeta_i^{(*)} z_i(t)^2}{1 + \gamma_i z_i(t)^2} G_i(z_{i}(t),z_{i}(t)^2) +  \dfrac{\gamma_i  \sigma_i |z_i(t)|}{1 + \gamma_i z_i(t)^2} g_i(z_{i}(t),z_{t}(t)^2)
$$
is convergent. For this, let us consider the following sequence 
$$B_i(t)=\dfrac{\zeta_i{(t)}  + \gamma_i \zeta_i^{(*)} z_i(t)^2}{\gamma_i  \sigma_i z_{i}(t)}.$$

We distinguish two cases: 

\begin{itemize}
    \item[(1)] If $\underset{t \rightarrow +\infty }\lim |z_{i}(t)| =+\infty$, then $\underset{t \rightarrow +\infty }\lim B_i(t) =+\infty$ and consequently,    $g_i(z_{i}(t),z_{i}(t)^2)=o_1(t)$ and 
          $G_i(z_{i}(t),z_{i}(t)^2)=1+o_2(t)$ with $\underset{t \rightarrow +\infty }\lim o_1(t)=
          \underset{t \rightarrow +\infty }\lim o_2(t)=0$. Then
            $$
                \mathbb{E}_{\varepsilon_{i}(t)}[\zeta_i(t+1)] = \dfrac{\Big( \mathbb{E}_{\varepsilon_{i}(t)}[\zeta_i(t)]  + \gamma_i \zeta_i^{(*)}z_{i}(t)\Big)(1+o_2(t))+
                \sigma_i\gamma_i|z_{i}(t)|o_1(t)}{1+\gamma_i z_{i}(t)^2}.
            $$
            Therefore, we obtain that
            $$
            \begin{array}{ll}
              \underset{t \rightarrow +\infty }\lim\mathbb{E}_{\varepsilon_{i}(t)}[\zeta_i(t+1)] &=
                \underset{t \rightarrow +\infty }\lim \dfrac{\Big( \mathbb{E}_{\varepsilon_{i}(t)}[\zeta_i(t)]  + \gamma_i \zeta_i^{(*)}z_{i}(t)^2
                \Big)
                (1+o_2(t))+\sigma_i\gamma_i|z_{i}(t)|o_1(t)}{1+\gamma_i z_{i}(t)^2}\\
              &=
                \underset{t \rightarrow +\infty }\lim \dfrac{ \mathbb{E}_{\varepsilon_{i}(t)}[\zeta_i(t)]  + \gamma_i \zeta_i^{(*)}z_{i}(t)^2}{1+\gamma_iz_i(t)^2}.
            \end{array}
         $$
        Define the following sequence as follows by 
        $$
        \Bar{V}_{i}(t) = \mathbb{E}_{\varepsilon_{i}(t)}[\zeta_i(t)], \text{ and } 
        \Bar{V}_{i}(t+1)=\dfrac{ \Bar{V}_{i}(t)  + \gamma_i \zeta_i^{(*)}z_{i}(t)^2}{1+\gamma_i z_{i}(t)^2}.
        $$
        We have the following implications:
    \begin{itemize}
      \item[I.1.] If $\Bar{V}_{i}(0)\geq \zeta_i^{(*)}$, then the sequence $\left\{\Bar{V}_{i}(t) \right\}_{t\geq 0}$ is decreasing and lower bounded by $\zeta_i^{(*)}$.
      \item[I.2.]  If $\Bar{V}_{i}(0)\leq \zeta_i^{(*)}$, then the sequence $\left\{\Bar{V}_{i}(t) \right\}_{t\geq 0}$ is increasing and upper bounded by $\zeta_i^{(*)}$.
    \end{itemize}
    For the first case, by induction we have $\Bar{V}_{i}(0)\geq \zeta_i^{(*)}$ and $X_{1}\geq \zeta_i^{(*)}$. Assume that $\Bar{V}_{i}(t)\geq \zeta_i^{(*)}$. If $\Bar{V}_{i}(t+1)< \zeta_i^{(*)}$, then $$
    \dfrac{ \Bar{V}_{i}(t) + \gamma_i \zeta_i^{(*)} z_{i}(t)^2 }{\left(1 + \gamma_i z_{i}(t)^2\right)}<\zeta_i^{(*)},
    $$ which implies that $\Bar{V}_{i}(t)< \zeta_i^{(*)}$ which is a contradiction. Hence,
    $$
    \begin{array}{ll}
    \Bar{V}_{i}(t+1)-\Bar{V}_{i}(t) &= \dfrac{ \Bar{V}_{i}(t) + \gamma_i \zeta_i^{(*)} z_{i}(t)^2 }{\left(1 + \gamma_i z_{i}(t)^2\right)}- \Bar{V}_{i}(t)\\
    &= \gamma_i \zeta_i^{(*)} z_{i}(t)^2\dfrac{ \zeta_i^{(*)} - \Bar{V}_{i}(t)}{\left(1 + \gamma_i z_{i}(t)^2\right)}\\
    &\leq0.
    \end{array}
    $$
    Then the sequence $\left\{ \Bar{V}_{i}(t) \right\}_{t\geq 0}$ is decreasing. We can prove the second case by an analogous reasoning. We then conclude that the sequence $\left\{ \Bar{V}_{i}(t) \right\}_{t\geq 0}$ is convergent. 

    \item[(2)] For the case $\underset{t \rightarrow +\infty }\lim |z_{i}(t)| =L < +\infty$, let us assume that the sequence $\{\mathbb{E}_{\varepsilon_{i}(t)}[\zeta_i(t)]\}_t$ is divergent. Then, 
    $\underset{t \rightarrow +\infty }\lim B_i(t) =+\infty$ and consequently,    $g_i(z_{i}(t),z_{i}(t)^2)=o_1(t)$ and 
          $G_i(z_{i}(t),z_{i}(t)^2)=1+o_2(t)$ with $\underset{t \rightarrow +\infty }\lim o_1(t)=
          \underset{t \rightarrow +\infty }\lim o_2(t)=0$. Therefore, we obtain
            $$ \mathbb{E}_{\varepsilon_{i}(t)}[\zeta_i(t+1)] = \dfrac{\Big( \mathbb{E}_{\varepsilon_{i}(t)}[\zeta_i(t)]  + \gamma_i \zeta_i^{(*)}z_{i}(t)\Big)(1+o_2(t))+
            \sigma_i\gamma_i |z_{i}(t)|o_1(t)}{1+\gamma_i z_{i}(t)^2},
            $$
            which implies
            $$
            \begin{array}{ll}
              \underset{t \rightarrow +\infty }\lim\mathbb{E}_{\varepsilon_{i}(t)}[\zeta_i(t+1)] &=
                \underset{t \rightarrow +\infty }\lim \dfrac{\Big( \mathbb{E}_{\varepsilon_{i}(t)}[\zeta_i(t)]  + \gamma_i \zeta_i^{(*)}z_{i}(t)^2
                \Big)
                (1+o_2(t))+\sigma_i \gamma_i |z_{i}(t)|o_1(t)}{1+\gamma_i z_i(t)^2}\\
              &=
                \underset{t \rightarrow +\infty }\lim \dfrac{ \mathbb{E}_{\varepsilon_{i}(t)}[\zeta_i(t)]  + \gamma_i \zeta_i^{(*)}z_{i}(t)^2}{1+\gamma_i z_i(t)^2}.
            \end{array}
         $$
        Next, we define the following sequence: 
        $$\Bar{V}_{i}(t)= \mathbb{E}_{\varepsilon_{i}(t)}[\zeta_i(t)], \text{ and } 
        \Bar{V}_{i}(t+1)=\dfrac{ \Bar{V}_{i}(t)  + \gamma_i \zeta_i^{(*)}z_{i}(t)^2}{1+\gamma_i z_i(t)^2}.
        $$
        By I.1 and I.2, the sequence $\{\Bar{V}_{i}(t)\}_t$ is convergent, which is a contradiction with the divergence of $\{\mathbb{E}_{\varepsilon_{i}(t)}[\zeta_i(t)]\}_t$. Therefore, the sequence $\left\{ \Bar{V}_{i}(t) \right\}_{t\geq 0}$ converges.
\end{itemize}
By using the fact that in both cases (i) and (ii), $\{\mathbb{E}_{\varepsilon_{i}(t)}[\zeta_i(t)]\}_t$ is a convergent sequence, since 
$$
\Bar{V}_{i}(t+1)=\dfrac{ \Bar{V}_{i}(t)  + \gamma_i \zeta_i^{(*)}z_{i}(t)^2}{1+\gamma_i z_i(t)^2}.
$$
we have
$$
\Bar{V}_{i}(t+1)-\zeta_i^{(*)}=\dfrac{ \Bar{V}_{i}(t)  - \zeta_i^{(*)}}{1+\gamma_i z_i(t)^2}, \qquad \mbox{ and } \qquad 
\Bar{V}_{i}(t)-\zeta_i^{(*)}=\dfrac{ \Bar{V}_{i}(0)  - \zeta_i^{(*)}}{\underset{h=1}{\overset{t-1}\prod}\left(1+\gamma_i z_i(t)^2\right)}.
$$
Hence, we obtain that 
$$
\underset{t \rightarrow +\infty }\lim\mathbb{E}_{\varepsilon_{i}(t)}[\zeta_i(t)] = \left\{
    \begin{array}{ll}
      \zeta_i^{(*)}, & \text{if } \underset{t \rightarrow +\infty }\lim z_{i}(t) \neq0 \\
      \zeta_i^{(*)}+\dfrac{\zeta_i(0)-\zeta_i^{(*)}}{\underset{h=1}{\overset{\infty}\prod}\left(1+\gamma_i z_{i}(t)^2\right)},  & \text{if } \underset{t \rightarrow +\infty }\lim z_{i}(t)=0.
\end{array}\right.
$$
\end{proof}


\subsection*{\textbf{Proposition \ref{prop:convergence_mode_short_memory}}}

\begin{proof}
Conditioned on the input decisions $L_i(t)$, the MAP estimator follows a  deterministic sequence:
\begin{equation}\label{eq:v_learning_n1_onestep_expect}
    \overline{v}_i(t+1)  =
    \tilde{\mathfrak{1}}_{i, t}
  \dfrac{ \overline{v}_i(t) + \gamma_i \zeta_i^{(*)} z_{i}(t)^2 }{\left(1 + \gamma_i z_{i}(t)^2\right)}.
\end{equation}
By recursively substituting $\overline{v}_i(t)$ into $\overline{v}_i(t+1)$, we obtain:
$$
\begin{array}{ll}
    \overline{v}_i(1) &  = \tilde{\mathfrak{1}}_{i,0} \dfrac{ \zeta^{(0)}_i + \gamma_i v^{(*)}_i z_i(0)^2 }{R_{i,0}},\\
    \overline{v}_i(2) &  = \tilde{\mathfrak{1}}_{i,0} \tilde{\mathfrak{1}}_{i,1}  \dfrac{ \zeta^{(0)}_i + \gamma v^{(*)}_i z_i(0)^2 }{R_{i,0} R_{i,1}} + \tilde{\mathfrak{1}}_{i,1} \dfrac{ \gamma_i v^{(*)}_i z_i(1)^2 }{R_{i,1}},\\
    \overline{v}_i(3) &  = \tilde{\mathfrak{1}}_{i,0} \tilde{\mathfrak{1}}_{i,1} \tilde{\mathfrak{1}}_{i,2} \dfrac{ \zeta^{(0)}_i + \gamma_i v^{(*)}_i z_i(0)^2 }{R_{i,0} R_{i,1} R_{i,2}} + \tilde{\mathfrak{1}}_{i,1} \tilde{\mathfrak{1}}_{i,2} \dfrac{ \gamma_i v^{(*)}_i z_i(1)^2 }{R_{i,1} R_{i,2}} + \tilde{\mathfrak{1}}_{i,2} \dfrac{ \gamma_i v^{(*)}_i z_i(2)^2 }{R_{i,2}}, \\ 
    \overline{v}_i(4) &  = \tilde{\mathfrak{1}}_{i,0} \tilde{\mathfrak{1}}_{i,1} \tilde{\mathfrak{1}}_{i,2} \tilde{\mathfrak{1}}_{i,3} \dfrac{ \zeta^{(0)}_i + \gamma v^{(*)}_i z_i(0)^2 }{R_{i,0} R_{i,1} R_{i,2} R_{i,3}} + \tilde{\mathfrak{1}}_{i,1} \tilde{\mathfrak{1}}_{i,2} \tilde{\mathfrak{1}}_{i,3}\dfrac{ \gamma_i v^{(*)}_i z_i(1)^2 }{R_{i,1} R_{i,2} R_{i,3}} + \tilde{\mathfrak{1}}_{i,2} \tilde{\mathfrak{1}}_{i,3} \dfrac{ \gamma_i v^{(*)}_i z_i(2)^2 }{R_{i,2} R_{i,3}} + \tilde{\mathfrak{1}}_{i,3} \dfrac{ \gamma_i v^{(*)}_i z_i(3)^2 }{R_{i,3}},\\     
    \vdots & \vdots \\
    \overline{v}_i(t) & = \displaystyle \zeta^{(0)}_i \left(\prod_{h=0}^{t-1} \dfrac{\tilde{\mathfrak{1}}_{i,h}}{R_{i,h}}\right) + v^{(*)}_i \left(\gamma_i \sum_{h=0}^{t-1} z(h)^2 \prod_{s=h}^{t-1} \dfrac{\tilde{\mathfrak{1}}_{i,h} }{ R_{i,s}}\right),
\end{array}
$$
where $R_{i,t} = (1 + \gamma_i z_i(t)^2)$. This proves \eqref{eq:mode_shor_tmemory_general_form}. We distinguish the following cases:
\begin{itemize}
    \item[(1)] If $\Psi$ is a finite set, then there is a $t_0\in \nset$ such that for each $t\geq t_0$, we have $\overline{v}_i(t+1)=0$ and therefore $\left\{ \overline{v}_i(t) \right\}_{t}$  converges to 0.
    
    \item[(2)] Assume that $\Psi_i$ is an infinite set and $\underset{t\rightarrow\infty}\lim z_{i}(t) = 0$. Then, without loss of generality, let $\Psi_i\equiv \nset$ (in other words, if $\Psi_i\not\equiv\nset$, then there exists a bijection $\varphi:\nset\rightarrow\varphi(\nset)=\Psi_i$). 
    
    We now note that the convergence of $\left\{ \overline{v}_i(t) \right\}_{t}$ is a direct consequence of the following result:
    \begin{itemize}
      \item[I.3] If $\overline{v}_i(0)\geq \zeta_i^{(*)}$, then the sequence $\left\{ \overline{v}_i(t) \right\}_{t}$ is decreasing and lower bounded by $\zeta_i^{(*)}$.
      \item[I.4] If $\overline{v}_i(0)\leq \zeta_i^{(*)}$, then the sequence $\left\{ \overline{v}_i(t) \right\}_{t}$ is increasing and upper bounded by $\zeta_i^{(*)}$.
    \end{itemize}
    In the case of I.3, we first prove that $\left\{ \overline{v}_i(t) \right\}_{t}$ is lower bounded by $\zeta_i^{(*)}$. We proceed by induction and let $\overline{v}_i(0)\geq \zeta_i^{(*)}$. We show that if $\overline{v}_i(t)\geq \zeta_i^{(*)}$, then $\overline{v}_i(t+1)\geq \zeta_i^{(*)}$. By contradiction, let us assume that $\overline{v}_i(t+1)< \zeta_i^{(*)}$. Then, we have 
    $$
    \dfrac{ \overline{v}_i(t) + \gamma_i \zeta_i^{(*)} z_{i}(t)^2 }{\left(1 + \gamma_i z_{i}(t)^2\right)}<\zeta_i^{(*)}.
    $$ 
    This implies that $\overline{v}_i(t)< \zeta_i^{(*)}$, which is a contradiction. Still in I.3, we can now prove that the sequence $\left\{ \overline{v}_i(t) \right\}_{t}$ is decreasing:
    $$
    \begin{array}{ll}
    \overline{v}_i(t+1)-\overline{v}_i(t) &= \dfrac{ \overline{v}_i(t) + \gamma_i \zeta_i^{(*)} z_{i}(t)^2 }{\left(1 + \gamma_i z_{i}(t)^2\right)}- \overline{v}_i(t)\\
    &= \gamma_i \zeta_i^{(*)} z_{i}(t)^2\dfrac{ \zeta_i^{(*)} - \overline{v}_i(t)}{\left(1 + \gamma_i z_{i}(t)^2\right)}\\
    &\leq0.
    \end{array}
    $$
    In the case of I.4, we can use the same argument to prove that $\left\{ \overline{v}_i(t) \right\}_{t}$ is increasing and upper bounded by $\zeta_i^{(*)}$.
    
    Using the proof of Proposition \ref{prop:short_memory_learning}, we know that $\{\mathbb{E}_{\varepsilon_{i}(t)}[\zeta_i(t)]\}_t$ is a convergent sequence, and that in the case $\underset{t \rightarrow +\infty }\lim z_{i}(t)=0$ we have
    $$
    \underset{t \rightarrow +\infty }\lim \overline{v}_i(t)=\zeta_i^{(*)}+\dfrac{\zeta_i(0)-\zeta_i^{(*)}}{\underset{h=1}{\overset{\infty}\prod}\left(1+\gamma_i z_{i}(h)^2\right)}.
    $$
    
    \item[(3)] If $\Psi$ is an infinite set and $\underset{t\rightarrow\infty}\lim z_{i}(t) \neq 0$, then we show that $\left\{ \overline{v}_i(t) \right\}_{t}$ converges to $v^{(*)}_i$. In doing so, we first rewrite \eqref{eq:v_learning_n1_onestep_expect} as
    $$
    \overline{v}_i(t+1)=\left(\frac{1}{1+\gamma_i z_i(t)^2}\right)\left[\overline{v}_i(t)+\gamma_i v^{(*)}_i z_i(t)^2\right].
    $$
    Next, we note that
    $$
    \overline{v}_i(t+1) - v^{(*)}_i = \frac{\overline{v}_i(t)-v^{(*)}_i}{1+\gamma_i z_i(t)^2}
    $$
    implies the monotonicity of the sequence (namely, if $\overline{v}_i(0)\geq v^{(*)}_i $ then $\overline{v}_i(t+1)\leq \overline{v}_i(t)$ and if $\overline{v}_i(0)\leq v^{(*)}_i$ then $\overline{v}_i(t+1)\geq \overline{v}_i(t)$). Next, we have
$$
\frac{\overline{v}_i(t+1)-v^{(*)}_i}{\overline{v}_i(t)-v^{(*)}_i}=\frac{1}{1+\gamma_i z_i(t)^2}\leq 1.
$$

\begin{description}
\item Case 1) If the sequence $\{z_i(t)^2\}_t$ is divergent, then $\forall M>0$, there exists $t_0>0$, such that $\forall t>t_0$, $z_i(t)^2>M$. Therefore, for $M=1$, $\exists t_0>0$ , $\forall t>t_0$, $z_i(t)^2>1$. Then 
$$
\frac{1}{1+\gamma_i z_i(t)^2}\leq \frac{1}{1+\gamma_i},\text{ for } t>t_0.
$$
Let $\overline{\gamma}_i = \frac{1}{1+\gamma_i}<1$, then we obtain
$$
|\overline{v}_i(t+1)-v^{(*)}_i|\leq \overline{\gamma}_i |\overline{v}_i(t)-v^{(*)}_i|.
$$ 
This implies 
$$
|\overline{v}_i(t)-v^{(*)}_i|\leq \overline{\gamma}_i^t |\zeta_i(0)-v^{(*)}_i|\underset{t\rightarrow\infty}\longrightarrow 0.
$$
\item Case 2) If the sequence $\{z_i(t)^2\}_t$ converges to $M$, then by hypothesis 
$$
\underset{t\rightarrow\infty}\lim z_i(t) = M \neq 0.
$$
As a consequence, $\forall \epsilon>0$, $\exists t_0>0$ such that $\forall t>t_0$, $M-\epsilon\leq r^2(t)\leq M+\epsilon$. Choose an arbitrary $\epsilon>0$ with $M-\epsilon>0$. Then, $\exists t_0>0$ such that $\forall t>t_0$, $0< M-\epsilon\leq r^2(t)$. Therefore 
$$\frac{1}{1+\gamma_i z_i(t)^2}\leq \frac{1}{1+(M-\epsilon)\gamma_i}.
$$
Let $\overline{\gamma}_i=\frac{1}{1+(M-\epsilon)\gamma_i}<1$, then we obtain
$$
|\overline{v}_i(t+1)-v^{(*)}_i|\leq \overline{\gamma}_i |\overline{v}_i(t)-v^{(*)}_i|.
$$ 
This implies that 
$$
|\overline{v}_i(t)-v^{(*)}_i|\leq \overline{\gamma}_i^t |\zeta_i(0)-v^{(*)}_i|\underset{t\rightarrow\infty}\longrightarrow 0.
$$
\end{description}
\end{itemize}
\end{proof}


\subsection*{\textbf{Proposition \ref{prop:demography_short_mem}}}

\begin{proof}
Following an analogous procedure as in the proof of Proposition \ref{prop:demography_long_mem}, we obtain: 
    $$
    \begin{array}{ll}
        \frac{d ~ \ln( l_i(t) )}{d ~\delta(t)} & = \displaystyle \frac{d}{d ~\delta(t)} \left[ \ln\left(\xi_{i,t} w^*(t)(\delta(t) - l_0(t) )^{\kappa_{i,t}} \right)  \right]\\
        & \displaystyle = \frac{\kappa_{i,t}}{w^*(t)(\delta(t) - l_0(t) )}  \left(\frac{d ~ w^*(t)(\delta(t) - l_0(t))}{d ~\delta(t)} \right).
    \end{array}
    $$
Using $\frac{d ~w^*(t)(\delta(t) - l_0(t) )}{d ~\delta(t)}$ in \eqref{eq:d_w_d_l}, we have
    $$
   \frac{d ~ \ln( l_i(t) )}{d ~\delta(t)} = \dfrac{\kappa_{i,t}}{\sum_{j=1}^{n} \xi_{j,t} \cdot \kappa_{j,t} \cdot w^*(t)(\delta(t) - l_0(t))^{\kappa_{j,t}} } 
    $$
    Hence, there exists $\Bar{\delta}(t)$, such that for all $\delta(t) > \Bar{\delta}(t)$,
    $$
    \xi_{i,t} w^*(t)(\delta(t) - l_0(t))^{\kappa_{i,t}} > 1 \quad \mbox{ and } \quad  \frac{d }{d ~\delta(t)} \ln( l_i(t) )  > 0 
    $$
        Therefore, there exists a monotonic sequence $\{\delta(t)\}_t$, with $\delta(t)> \delta(t-1)$, such that the following uniform divergence holds:\footnote{Let $\{f_n\}_{n=1}^{\infty}$ be a sequence of real-valued functions defined on a domain $D \subseteq \mathbb{R}$. We say that the sequence $\{f_n(x)\}$ diverges uniformly on $D$ if for every $M > 0$, there exists a positive integer $N$ such that for all $n \geq N$ and for all $x \in D$, we have $|f_n(x)| > M$. In other words, beyond some index $N$, the magnitude of the functions $f_n(x)$ exceeds any given bound $M$ uniformly over the domain $D$.}
    $$
     \underset{t \rightarrow +\infty }\lim \ln( l_i(t) )= +\infty ,
    $$
    so  that $\underset{t \rightarrow +\infty }\lim l_{i}(t) \neq 1$.
\end{proof}


\subsection*{\textbf{Proposition \ref{prop:l_i_bound}}}

\begin{proof}
Let us $w^*(\delta(t) - l_0(t))$ be the collection of solutions of $L_{t}(w) = \delta(t) - l_0(t)$, and define as $\overline{w}^*$ and $\underline{w}^*$ the maximum and minimum element in $w^*(\delta(t) - l_0(t))$. Since $\kappa_i < -1$, we have
$$
\begin{array}{ll}
     \overline{w}^* & \displaystyle = \max\left \{w \geq 0 ~:~ \sum_{i=1}^{n} \chi_{i,t}  \cdot w^{\kappa_{i,t} } = \delta(t) - l_0(t)\right\}  \\[0.5cm]
     & \displaystyle \leq \max\left \{w \geq 0 ~:~ \sum_{i=1}^{n} \chi_{i,t}  \cdot w^{\kappa_{i,t} } \geq \delta(t) - l_0(t)\right\}  \\[0.5cm]
     & \displaystyle \leq \max\left \{w \geq 0  ~:~ \dfrac{1}{n}\sum_{i=1}^{n}  \chi_{i,t} w^{\kappa_{i,t} } \geq \dfrac{\delta(t) - l_0(t)}{n}\right\}  \\[0.5cm]
    & \displaystyle \leq \max\left \{w \geq 0  ~:~ \max_{i=1} \{ \chi_{i,t} w^{\kappa_{i,t} } \} \geq \dfrac{\delta(t) - l_0(t)}{n}\right\}  \\[0.5cm]
    & \displaystyle \leq \max_{i}\left \{  \left( \dfrac{\delta(t) - l_0(t)}{n \chi_{i,t}}\right)^{\frac{1}{\kappa_{i,t}}} \right\}  \\[0.5cm]
    & \displaystyle = \overline{\overline{w}}
\end{array} 
$$
and
$$
\begin{array}{ll}
     \underline{w}^* & \displaystyle = \min\left \{w \geq 0 ~:~ \sum_{i=1}^{n} \chi_{i,t}  \cdot w^{\kappa_{i,t} } = \delta(t) - l_0(t)\right\}  \\[0.5cm]
     & \displaystyle \geq \min\left \{w \geq 0 ~:~ \sum_{i=1}^{n} \chi_{i,t}  \cdot w^{\kappa_{i,t} } \leq \delta(t) - l_0(t)\right\}  \\[0.5cm]
    & \displaystyle \geq \min\left \{w \geq 0  ~:~  n\min_{i}\left\{\chi_{i,t} w^{\kappa_{i,t} }\right\}
    \leq \delta(t) - l_0(t)\right\}  \\[0.5cm]
    & \displaystyle \geq \min_{i}\left \{  \left( \dfrac{\delta(t) - l_0(t)}{n \chi_{i,t}}\right)^{\frac{1}{\kappa_{i,t}}} \right\}  \\[0.5cm]
    & \displaystyle = \underline{\underline{w}}
\end{array} 
$$
We have
$$
\begin{array}{lll}
     l_i(t) & = \displaystyle \left(\mathbb{E}\left[ \eta_{i}(t)^{\alpha} \right]\dfrac{\hat{\zeta}_{i}(t) }{ w(t)  }\right)^{\dfrac{1 }{1 - \hat{\zeta}_{i}(t)\alpha } } \\
     & \geq \displaystyle \left(\mathbb{E}\left[ \eta_{i}(t)^{\alpha} \right]\dfrac{\hat{\zeta}_{i}(t) }{ \overline{\overline{w}}  }\right)^{\dfrac{1 }{1 - \hat{\zeta}_{i}(t)\alpha } }\\
     & = \displaystyle \left[\mathbb{E}\left[ \eta_{i}(t)^{\alpha} \right] \dfrac{\hat{\zeta}_{i}(t) }{ \displaystyle \max_{j}\left \{  \left( \dfrac{\delta(t) - l_0(t)}{n \chi_{j,t}}\right)^{\frac{1}{\kappa_{j,t}}} \right\}   }\right]^{-\kappa_{i,t}}
\end{array}
$$
Therefore, for any $\lambda_{\min}, \lambda_{\max} > 0$
$$
\displaystyle \mbox{if } \quad \hat{\zeta}_{i}(t) ~\geq~ \dfrac{1}{\mathbb{E}\left[ \eta_{i}(t)^{\alpha} \right]}\max_{j}\left \{  \left( \dfrac{\delta(t) - l_0(t)}{n \chi_{j,t}}\right)^{\frac{1}{\kappa_{j,t}}} \right\} \left(\lambda_{\min}\right)^{-\frac{1}{\kappa_{i,t}}} \quad \mbox{ then } \quad l_i(t) ~\geq~ \lambda_{\min},
$$
and
$$
\displaystyle \mbox{if } \quad  \hat{\zeta}_{i}(t)  ~\leq~ \dfrac{1}{\mathbb{E}\left[ \eta_{i}(t)^{\alpha} \right]} \min_{j}\left \{  \left( \dfrac{\delta(t) - l_0(t)}{n \chi_{j,t}}\right)^{\frac{1}{\kappa_{j,t}}} \right\}\left(\lambda_{\max}\right)^{-\frac{1}{\kappa_{i,t}}} \quad \mbox{ then } \quad l_i(t) ~\leq~ \lambda_{\max}.
$$

\end{proof}


\subsection*{\textbf{Proposition \ref{prop:comparison_learning}}}

\begin{proof}

We distinguish the two learning cases. 

\paragraph{Derivation for the path-dependent case.} To investigate the reason behind the faster convergence observed in Example~\ref{ex:4}, we notice that the path-dependent learning update in Equation~\eqref{eq:v_learning_n1_long_mem} can be equivalently rewritten in a recursive form. 

Recall that the update rule is given by:
\[
\zeta_i(t+1) = \left( \frac{\zeta_i^{(0)} + \gamma_i \sum_{\ell=1}^{t} z_i(\ell) s_i(\ell)}{1 + \gamma_i \sum_{\ell=1}^{t} z_i(\ell)^2} \right)^{+}.
\]

Now, let us denote by \( \zeta_i(t) \) the MAP estimator computed up to period \( t-1 \), i.e.,
\[
\zeta_i(t) = \left( \frac{\zeta_i^{(0)} + \gamma_i \sum_{\ell=1}^{t-1} z_i(\ell) s_i(\ell)}{1 + \gamma_i \sum_{\ell=1}^{t-1} z_i(\ell)^2} \right)^{+}.
\]

Multiplying both sides by the denominator yields:
\[
(1 + \gamma_i \sum_{\ell=1}^{t-1} z_i(\ell)^2) \cdot \zeta_i(t) = \zeta_i^{(0)} + \gamma_i \sum_{\ell=1}^{t-1} z_i(\ell) s_i(\ell).
\]

We now add the new term \( \gamma_i z_i(t)^2 s_i(t) \) to the numerator and \( \gamma_i z_i(t)^2 \) to the denominator. This gives the update at time \( t+1 \) as:
\[
\zeta_i(t+1) = \left( \frac{(1 + \gamma_i \sum_{\ell=1}^{t-1} z_i(\ell)^2) \cdot \zeta_i(t) + \gamma_i z_i(t)^2 s_i(t)}{1 + \gamma_i \sum_{\ell=1}^{t} z_i(\ell)^2} \right)^{+}.
\]
Therefore,
$$
\zeta_i(t+1) = a_t^{PD} \zeta_i(t) + (1-a_t^{PD}) \frac{s_i(t)}{z_i(t)},
$$
using Proposition  \ref{prop:equilibrium} to replace $z_i(t)$, we have
$$
\RRew{a_t^{PD} = \left( \frac{(1 + \gamma_i \sum_{\ell=1}^{t-1} z_i(\ell)^2) }{1 + \gamma_i \sum_{\ell=1}^{t} z_i(\ell)^2} \right)^{+}.}
$$

\paragraph{Derivation for the path-independent case.}
We now show that the path-independent update rule given in Equation~\eqref{eq:v_learning_n1_short_mem} can also be rewritten in the form stated in the proposition.

Recall that the path-independent update is:
\[
\zeta_i(t+1) = \left( \frac{\zeta_i(t) + \gamma_i z_i(t)(\varepsilon_i(t) + \zeta_i^{(*)} z_i(t))}{1 + \gamma_i z_i(t)^2} \right)^{+}.
\]

Using the definition \( s_i(t) = \varepsilon_i(t) + \zeta_i^{(*)} z_i(t) \), this becomes:
\[
\zeta_i(t+1) = \left( \frac{\zeta_i(t) + \gamma_i z_i(t)^2 \cdot \frac{s_i(t)}{z_i(t)} }{1 + \gamma_i z_i(t)^2} \right)^{+}.
\]

Now observe that this expression is equivalent to:
\[
\zeta_i(t+1) = \left( \frac{1}{1 + \gamma_i z_i(t)^2} \zeta_i(t) + \frac{\gamma_i z_i(t)^2}{1 + \gamma_i z_i(t)^2} \cdot \frac{s_i(t)}{z_i(t)} \right)^{+}.
\]

This confirms that the update is a convex combination of the previous estimate \( \zeta_i(t) \) and the moment-based statistic \( s_i(t)/z_i(t) \), with weight:
\[
a_t^{PI} = \frac{1}{1 + \gamma_i z_i(t)^2}.
\]

Hence, the recursive update takes the form:
\[
\zeta_i(t+1) = a_t^{PI} \zeta_i(t) + (1 - a_t^{PI}) \frac{s_i(t)}{z_i(t)},
\]
which aligns with the unified representation provided in the proposition.

\end{proof}

\paragraph{Proposition \ref{prop:multidim_MAP}}

\begin{proof} 

We distinguish the two cases of path-dependent and path-independent learning.

\paragraph{Path-dependent learning.} Based on Assumption \ref{ass:3}, $\pi_{0,i}$ is the density function of a zero-truncated Gaussian random variable with parameters $\boldsymbol{\beta}^{(0)}_i$ and $\tau_i I$. We can write
$$
\begin{array}{ll}
\mathscr{L}(\boldsymbol{\beta}_i(t); \, \mathcal{I}_i(t) ) \pi_{0,i}(\boldsymbol{\beta}_i(t)) & \propto \displaystyle \prod_{\ell=1}^{t} \exp \left(- \frac{1}{2 \sigma_i^{2}} \left( \log x_i(h) - \log \mu_i(h) \right)^2\right) \pi_{0,i}(\boldsymbol{\beta}_i(t))\\[0.9cm]
&= \begin{cases}
    \displaystyle \exp \left(- \frac{1}{2} \sum_{i = 1}^{n} \theta_i(\boldsymbol{\beta}_i(t), \mathcal{I}_i(t)) \right) & \quad \mbox{ if } \boldsymbol{\beta}_i(t) \geq 0\\
    0 & \quad \mbox{ otherwise},
    \end{cases}
    \end{array}
$$
where $\boldsymbol{\beta}_{i,.}(t)$ is the $i$-th row of $\beta(t)$ and $\theta_i$ is defined as
\begin{equation*}
\theta_i(\boldsymbol{\beta}_i(t), \mathcal{I}(t)) ~=~
\begin{cases}
\displaystyle  \sum_{\ell=1}^{t} \frac{1}{\sigma_i^2} \left( s_i(\ell) - \sum_{j=1}^{n} \beta_{i,j}(t) z_{i,j}(\ell) \right)^2 + \frac{1}{\tau_i^2} \sum_{j=1}^{n} \left( \beta_{i,j}(t) - \beta_{i,j}^{0} \right)^2  & \quad \mbox{if } t > 1\\[0.8cm]
\displaystyle \frac{1}{\tau_i^2} \sum_{j=1}^{n} \left( \beta_{i,j} - \beta_{i,j}^{0} \right)^2 , & \quad \mbox{if } t = 1,
\end{cases}
\end{equation*}
with $s_i(\ell)$ and $z_{i,j}(\ell)$ defined in \eqref{eq:sr}. To maximize the posterior distribution for each period $t$, it is sufficient to solve
$$
\displaystyle \max_{\beta(t)} \quad 
\exp \left(- \frac{1}{2} \sum_{i = 1}^{n} \theta_i(\boldsymbol{\beta}_{i,.}(t), \mathcal{I}(t)) \right), \quad \mbox{ subject to }  \beta_{i,j}(t) \geq 0,
$$
as $\int \mathscr{L}(\boldsymbol{\beta}_i(t); \, \mathcal{I}_i(t) ) \pi_{0,i}(\boldsymbol{\beta}_i(t))~d \boldsymbol{\beta}_i(t) $ is constant with respect to $\boldsymbol{\beta}_i(t)$. We apply the Karush-Kuhn-Tucker conditions:
\begin{equation*}
\begin{array}{rl}
    \displaystyle \frac{\partial}{\partial \beta_{i,j}(t)}  \left[ \exp \left(- \frac{1}{2} \sum_{s = 1}^{n} \theta_s(\boldsymbol{\beta}_s(t), \mathcal{I}(t)) \right) - q_{i,j}\beta_{i,j}(t)\right]
    ~=~ 0 \quad \mbox{ and } \quad q_{i,j}\beta_{i,j}(t) = 0,
\end{array}
\end{equation*}
so that
$$
\displaystyle  \exp \left(- \frac{1}{2} \sum_{i = 1}^{n} \theta_i(\boldsymbol{\beta}_i(t), \mathcal{I}(t)) \right) \left(\frac{\partial}{\partial \beta_{i,j}(t)} \theta_i(\boldsymbol{\beta}_i(t), \mathcal{I}(t) ) \right) = q_{i,j}
$$

We fix any arbitrary index set $\Omega_{i,t}$. Due to the complementarity either $q_{i,j} = 0$ or $\beta_{i,j}(t) = 0$. Therefore, for all $(i, j)$, if $\beta_{i,j}(t) > 0$, then $q_{i,j} = 0$ and $\frac{\partial}{\partial \beta_{i,j}(t)} \theta_i(\boldsymbol{\beta}_i(t), \mathcal{I}(t) ) = 0$. This implies that for any $i \in \mathcal{N}$, the components of the vector $\boldsymbol{\beta}_i(t)$ can be partitioned in two sets $\Omega_{i,t}$ (associated to $\beta_{i,j}(t) > 0$) and $\mathcal{N}/\Omega_{i,t}$ associated to $\beta_{i,j}(t) = 0$). Then, for any $j \in \Omega_{i,t}$, we have
\begin{equation*}
\begin{cases}
\displaystyle  - \sum_{\ell=1}^t \frac{2}{\sigma_i^2} z_{i,j}(\ell) \left( s_i(\ell) - \left(\sum_{h \in \Omega_{i,t}} \beta_{i,h}(t)z_{i,h}(\ell)\right)\right) = \frac{2}{\tau_i^2} \left(\beta_{i,j}^{(0)} - \beta_{i,j}(t) \right)  & \quad \mbox{if } t > 1\\[0.6cm]
\displaystyle - \frac{2}{\tau_i^2} \left(\beta_{i,j}^{(0)} - \beta_{i,j}(t) \right) = 0, & \quad \mbox{if } t = 1,
\end{cases}
\end{equation*}
and for any $j \in \mathcal{N}/\Omega_{i,t}$, we have $\beta_{i,j}(t) = 0$.  Now, for all $j \in \Omega_{i,t}$, we can rearrange the first order conditions to get
\begin{equation*}
\begin{cases}
\displaystyle \sum_{\ell=1}^t \frac{1}{\sigma_i^2} z_{i,j}(\ell)\left(\sum_{h \in \Omega_{i,t}} \beta_{i,h}(t)z_{i,h}(t)\right) + \frac{1}{\tau_i^2}  \beta_{i,j}(t) = b_{i,j}(t)  & \quad \mbox{if } t > 1\\[0.5cm]
\displaystyle \beta_{i,j}(t)  = \beta_{i,j}^{(0)} , & \quad \mbox{if } t = 1,
\end{cases}
\end{equation*}
\noindent where $b_{i,j}(t) = (2/\tau_i^2)\beta_{i,j}^{(0)} - (2/\sigma_i^2) \sum_{\ell=1}^t s_i(\ell) z_{i,j}(\ell)$. Using the notation \eqref{eq:H_matrix} and \eqref{eq:B_vector}, we write
\begin{equation*}
\begin{cases}
\displaystyle  \left(\sum_{\ell=1}^{t} \frac{1}{\sigma_i^2} \mathbf{z}_{i,.}(\ell, \Omega_{i,t})\mathbf{z}_{i,.}(\ell, \Omega_{i,t})\T + \frac{1}{\tau_i^2} I_{|\Omega_{i,t}|} \right)  \boldsymbol{\beta}_i(t, \Omega_{i,t}) =  \mathbf{b}_{i}(t, 0, \Omega_{i,t})   & \quad \mbox{if } t > 1\\[0.5cm]
\displaystyle \boldsymbol{\beta}_i(t, \Omega_{i,t}) = \boldsymbol{\beta}_i^{(0)}(\Omega_{i,t}) , & \quad \mbox{if } t = 1,
\end{cases}
\end{equation*}
\noindent where $\boldsymbol{\beta}_{i}(t,\mathcal{U})$ is the $|\mathcal{U}|$-dimensional sub-vector of $\boldsymbol{\beta}_{i}(t)$, whose components correspond to the indices in $\mathcal{U}$ (in the same order). Therefore, all $i \in \Omega_{i,t}$ must verify $H_i(t, 0, \Omega_{i,t})^{-1} \mathbf{b}_{i}(t, 0, \Omega_{i,t}) \geq 0$. We now consider the elements $j \in \mathcal{N}/\Omega_{i,t} $. We have
$$
\displaystyle  \exp \left(- \frac{1}{2} \sum_{i = 1}^{n} \theta_i(\boldsymbol{\beta}_i(t), \mathcal{I}(t)) \right) \left(\frac{\partial}{\partial \beta_{i,j}(t)} \theta_i(\boldsymbol{\beta}_i(t), \mathcal{I}(t) ) \right) \geq 0, \mbox{ which implies } \frac{\partial}{\partial \beta_{i,j}(t)} \theta_i(\boldsymbol{\beta}_i(t), \mathcal{I}(t) )  \geq 0.
$$
Therefore, 
\begin{equation*}
\begin{cases}
\displaystyle   \sum_{\ell=1}^t \frac{1}{\sigma_i^2} z_{i,j}(\ell)  \left(\sum_{h \in \Omega_{i,t}} \beta_{i,h}(t)z_{i,h}(\ell)-s_i(\ell)\right) \geq \frac{1}{\tau_i^2} \beta_{i,j}^{(0)}  & \quad \mbox{if } t > 1\\[0.3cm]
\displaystyle \beta_{i,j}^{(0)}  \geq 0, & \quad \mbox{if } t = 1,
\end{cases}
\end{equation*}
and
\begin{equation*}
\displaystyle  \left(\sum_{\ell=1}^{t} \frac{1}{\sigma_i^2} \mathbf{z}_{i,.}(\ell, \mathcal{N}/\Omega_{i,t})\mathbf{z}_{i,.}(\ell, \Omega_{i,t})\T \right)  \boldsymbol{\beta}_i(t, \Omega_{i,t}) \geq  \boldsymbol{\beta}_{i}^{(0)} +   \frac{\tau_i^2}{\sigma_i^2} s_i(\ell) \sum_{\ell=1}^{t} \mathbf{z}_{i,.}(\ell, \mathcal{N}/\Omega_{i,t}). 
\end{equation*}
Using the notation \eqref{eq:H_matrix} and \eqref{eq:B_vector}, for every $t \geq 1$, let us define $\Omega_{i,t}(t,0)$ as follows: for any $\mathcal{U} \subseteq \mathcal{N}$ if $H_i(t, 0, \mathcal{U})^{-1} \mathbf{b}_{i}(t, 0, \mathcal{U}) > 0$ and
$H_i(t, 0,\mathcal{U})^{-1} \mathbf{b}_{i}(t, 0, \mathcal{U}) \geq \mathbf{b}_i(t,0,\mathcal{N}/\mathcal{U})$, then $\mathcal{U}\subseteq\Omega_{i,t}(t,0)$. Therefore, 
$$
\boldsymbol{\beta}_i(t,\Omega_{i,t}) = H_i(t, 1, \Omega_{i,t})^{-1} \mathbf{b}_{i}(t, 1, \Omega_{i,t}) \qquad \mbox{ and } \qquad \boldsymbol{\beta}_i(t, \mathcal{N}/\Omega_{i,t}) = \mathbf{0}.
$$

\paragraph{Path-independent learning.} 

To establish the distinction with respect to the proof of the path-dependent case, it is sufficient to redefine 
$$
\begin{array}{ll}
\mathscr{L}(\boldsymbol{\beta}_i(t); \, \mathcal{I}_i(t) ) \pi_{0,i}(\boldsymbol{\beta}_i(t)) & \propto \displaystyle \exp \left(- \frac{1}{2 \sigma_i^{2}} \left( \log x_i(t) - \log \mu_i(t) \right)^2\right) \pi_{0,i}(\boldsymbol{\beta}_i(t))\\[0.9cm]
&= \begin{cases}
    \displaystyle \exp \left(- \frac{1}{2} \sum_{i = 1}^{n} \theta_i(\boldsymbol{\beta}_i(t), \mathcal{I}_i(t)) \right) & \quad \mbox{ if } \boldsymbol{\beta}_i(t) \geq 0\\
    0 & \quad \mbox{ otherwise},
    \end{cases}
    \end{array}
$$
where $\boldsymbol{\beta}_{i,.}(t)$ is the $i$-th row of $\beta(t)$ and $\theta_i$ is defined as
\begin{equation*}
\theta_i(\boldsymbol{\beta}_i(t), \mathcal{I}(t)) ~=~
\begin{cases}
\displaystyle  \frac{1}{\sigma_i^2} \left( s_i(t) - \sum_{j=1}^{n} \beta_{i,j}(t) z_{i,j}(\ell) \right)^2 + \frac{1}{\tau_i^2} \sum_{j=1}^{n} \left( \beta_{i,j}(t) - \beta_{i,j}(t-1) \right)^2  & \quad \mbox{if } t > 1\\[0.8cm]
\displaystyle \frac{1}{\tau_i^2} \sum_{j=1}^{n} \left( \beta_{i,j} - \beta_{i,j}^{(0)} \right)^2 , & \quad \mbox{if } t = 1.
\end{cases}
\end{equation*}
Therefore, 
\begin{equation*}
\begin{cases}
\displaystyle    \frac{1}{\sigma_i^2} z_{i,j}(t)  \left(\sum_{h \in \Omega_{i,t}} \beta_{i,h}(t)z_{i,h}(t)-s_i(t)\right) \geq \frac{1}{\tau_i^2} \beta_{i,j}(t-1)  & \quad \mbox{if } t > 1\\[0.3cm]
\displaystyle \beta_{i,j}^{(0)}  \geq 0, & \quad \mbox{if } t = 1,
\end{cases}
\end{equation*}
and
\begin{equation*}
\displaystyle  \left( \frac{1}{\sigma_i^2} \mathbf{z}_{i,.}(t, \mathcal{N}/\Omega_{i,t})\mathbf{z}_{i,.}(t, \Omega_{i,t})\T \right)  \boldsymbol{\beta}_i(t, \Omega_{i,t}) \geq  \boldsymbol{\beta}_{i}(t-1, \Omega_{i,t}) +   \frac{\tau_i^2}{\sigma_i^2} s_i(t)  \mathbf{z}_{i,.}(t, \mathcal{N}/\Omega_{i,t}). 
\end{equation*}
so that
$$
\boldsymbol{\beta}_i(t,\Omega_{i,t}) = H_i(t, t, \Omega_{i,t})^{-1} \mathbf{b}_{i}(t, t \Omega_{i,t}) \qquad \mbox{ and } \qquad \boldsymbol{\beta}_i(t, \mathcal{N}/\Omega_{i,t}) = \mathbf{0}.
$$

\end{proof}


\paragraph{Lemma \ref{lemma:H_conv}}



\begin{proof} Based on \cite{miller1981inverse} (who generalizes the theory developed by \cite{sherman1950adjustment}), we know that for a given square matrix $H(t+1) = A + B(1) + \ldots + B(t)$, such that $B(i)$ is rank one and $A$ is non-singular, we have 
\begin{equation*}
H(t+1)^{-1} = H(t)^{-1} - \frac{1}{1 + \mbox{Tr}[H(t)^{-1}B(t)]} H(t)^{-1}B(t)H(t)^{-1},
\end{equation*}
with $H(1)^{-1} = A^{-1}$ and $\mbox{Tr}[.]$ denoting the trace of a matrix inside brackets. To keep notation short, $H_i(t',t)$ and $\mathbf{z}_{i,.}(t')$ shall be used instead of $H_{i}(t',t, \mathcal{U})$ and $\mathbf{z}_{i,.}(t', \mathcal{U})$ in the rest of this proof, as the argument is valid for all $\mathcal{U}$. Using \eqref{eq:H_matrix}, we define
$$
\kappa_{i, t', t} = 1 + \mbox{Tr}\left[ H_{i}(t',t)^{-1}  \left(\mathbf{z}_{i,.}(t'+1)\mathbf{z}_{i,.}(t'+1)\T \right)  \right],
$$
and apply the generalized Sherman-Morrison formula and for all $i \in \mathcal{N}$:
\begin{equation}\label{eq:ShermanMorrison}
H_{i}(t'+1, t)^{-1} =
\begin{cases}
H_{i}(t', t)^{-1} - \frac{1}{\kappa_{i, t', t}} H_{i}^{-1}(t', t)  \left(\mathbf{z}_{i,.}(t')\mathbf{z}_{i,.}(t')\T \right) H_{i}^{-1}(t', t) & \mbox{ if } t' > 0 \\[0.1cm]
\tau_i I & \mbox{ if } t' = 0
\end{cases}
\end{equation}
From this formula, we can compute the first-period matrix as
\begin{equation*}
\begin{array}{ll}
H_{i}^{-1}(1 , 1) & = \displaystyle \tau_i^2 I - \frac{\tau_i^4}{\sigma_i^2 + \tau_i^2 \mathbf{z}_{i,.}(1) \T \mathbf{z}_{i,.}(1)} \mathbf{z}_{i,.}(1) \mathbf{z}_{i,.}(1)\T\\[0.8cm]
      & = \displaystyle  \left[ \begin{array}{ccc}
      \tau_i^2 - \displaystyle \frac{\tau_i^4 \phi^2 \log(y_{i,1}(1))^2}{\sigma_i^2 + \tau_i^2  \phi^2 \sum_{j} \log(y_{i,j}(1))^2 }    & \ldots & - \displaystyle \frac{\tau_i^4 \phi^2 \log(y_{i,1}(1))\log(y_{i,n}(1))}{\sigma_i^2 + \tau_i^2  \phi^2  \sum_{j} \log(y_{i,j}(1))^2 } \\
       \vdots  & \ddots  & \vdots \\
      - \displaystyle \frac{\tau_i^4 \phi^2  \log(y_{i,1}(1))\log(y_{i,n}(1))}{\sigma_i^2 + \tau_i^2  \phi^2  \sum_{j} \log(y_{i,j}(1))^2 }    & \ldots & \tau_i^2 - \displaystyle \frac{\tau_i^4 \phi^2 \log(y_{i,n}(1))^2}{\sigma_i^2 + \tau_i^2   \phi^2  \sum_{j} \log(y_{i,j}(1))^2 }
    \end{array}
    \right]
\end{array}
\end{equation*}
Note that for any $a_q, a_s \in \rset$, the $(s,q)$ element of $H_{i}(1, 1)^{-1}$ can be written as

$$
\frac{\gamma_i \phi^2 \log(y_{i,s}(1))\log(y_{i,q}(1))}{1 + \gamma_i  \phi^2 \sum_{j} \log(y_{i,j}(1))^2 } = \frac{\gamma_i \phi^2 a_s a_q}{1 + \gamma_i \phi^2(a_s^2 + a_q^2 + \sum_{j}a_j^2) } \in [-1, 1],
$$
\noindent Thus, since $\phi \in [0, 1]$, when either $\gamma_i$ or $\phi$ go to zero, $H_{i}^{-1}(1,1)$ converges in distribution to a constant $\tau_i^2 I$. This means that applying \eqref{eq:ShermanMorrison}, the second period matrix becomes
$$
\begin{array}{ll}
     H_{i}^{-1}(2,1) & = H_{i}^{-1}(1,1) - \frac{1}{\kappa_{i,1,1}} H_{i}^{-1}(1,1)  \left(\mathbf{z}_{i,.}(2)\mathbf{z}_{i,.}(2)\T \right) H_{i}^{-1}(1,1) \\[0.4cm]
     & = \displaystyle \tau_i^2 I - \frac{\gamma_i \phi^2}{1 + \gamma_i \phi^2 \log \mathbf{y}_{i,.}(2) \T \log \mathbf{y}_{i,.}(2)} \log \mathbf{y}_{i,.}(2) \log \mathbf{y}_{i,.}(2)\T     
\end{array}
$$
which converges again to $\tau_i^2 I$ when either $\gamma_i$ or $\phi$ go to zero. Therefore, we proceed by induction from the first period and observe that for every $\ell$ we can write
$$
\displaystyle \lim_{\gamma_i \rightarrow 0} ~ H_{i,}^{-1}(\ell, 1)  \displaystyle = \lim_{\gamma_i \rightarrow 0}
H_{i}^{-1}(\ell-1, 1) - \lim_{\gamma_i \rightarrow 0} \frac{1}{\kappa_{i,\ell-1,1}} H_{i}^{-1}(\ell-1, 1)  \left(\mathbf{z}_{i,.}(\ell)\mathbf{z}_{i,.}(\ell)\T \right) H_{i}^{-1}(\ell-1, 1).
$$
and
$$
\displaystyle \lim_{\phi \rightarrow 0} ~ H_{i,}^{-1}(\ell, 1)  \displaystyle = \lim_{\phi \rightarrow 0}
H_{i}^{-1}(\ell-1, 1) - \lim_{\phi \rightarrow 0} \frac{1}{\kappa_{i,\ell-1,1}} H_{i}^{-1}(\ell-1, 1)  \left(\mathbf{z}_{i,.}(\ell)\mathbf{z}_{i,.}(\ell)\T \right) H_{i}^{-1}(\ell-1, 1).
$$
Therefore,
$$
\displaystyle \lim_{\gamma_i \rightarrow 0} ~ H_{i,}^{-1}(\ell, 1) ~=~ 
\displaystyle \lim_{\phi \rightarrow 0} ~ H_{i,}^{-1}(\ell, 1) ~=~ \tau_i^2 I
$$
\end{proof}


\subsection*{\textbf{Corollary \ref{cor:f_conditional}}}

\begin{proof}
To keep notation short, $H_i(t',t)$, $\mathbf{b}_{i}(t', t)$, $\boldsymbol{\beta}_{i,.}^{(0)}$ and $\mathbf{z}_{i,.}(t')$ shall be used instead of $H_{i}(t',t, \mathcal{U})$, $\mathbf{b}_{i}(t', t, \mathcal{U})$, $\boldsymbol{\beta}_{i,.}^{(0)}(\mathcal{U})$ and $\mathbf{z}_{i,.}(t', \mathcal{U})$ in the rest of this proof, as the argument is valid for all $\mathcal{U}$. Using \eqref{eq:H_matrix}, we can write
\begin{equation*}
\mathbf{b}_{i}(t', t) = \displaystyle  \frac{1}{\tau_i^2}  \boldsymbol{\beta}_{i,.}^{(0)} + \sum_{\ell=t}^{t'} \frac{1}{\sigma_i^2}  (\boldsymbol{\varepsilon} + B^{(*)}(\mathbf{z}_{i,.}(\ell)) \otimes \mathbf{z}_{i,.}(\ell),
\end{equation*}
where $\varepsilon_i = \log \eta_{i}(t) \sim N(m_i, \sigma_i )$ and $\boldsymbol{\varepsilon} = [\varepsilon_1, \ldots, \varepsilon_n]^{\top}$. Given functions $f(x)$ and $g(x)$, we define the limit behavior $\rightsquigarrow$ by establishing that $f(x)\rightsquigarrow g(x)$  if and only if $\lim_{x \rightarrow \infty} \frac{f(x)}{g(x)} = 1$ \citep{de1981asymptotic}. From Lemma \ref{lemma:H_conv}, we know that $H_{i}(t', t, \mathcal{U})$ converges to $\tau_i^2 I$ in two limit cases. Therefore, for any realization of $\boldsymbol{\varepsilon}$, we have
$$
\displaystyle \boldsymbol{\beta}_{i,.}(t+1, \Omega_{i,t}) ~ \rightsquigarrow~   \boldsymbol{\beta}_{i,.}^{(0)}(\Omega_{i,t}) + \gamma_i \sum_{\ell=1}^{t}  \left(\boldsymbol{\varepsilon}(\ell) + (B^{(*)}(\Omega_{i,t})\mathbf{z}_{i,.}(\ell, \Omega_{i,t})) \otimes \mathbf{z}_{i,.}(\ell, \Omega_{i,t})\right), \mbox{ and } \boldsymbol{\beta}_i(t, \mathcal{N}/\Omega_{i,t}) = \mathbf{0},
$$
for the path-dependent learning, and
$$
\displaystyle \boldsymbol{\beta}_{i,.}(t+1, \Omega_{i,t}) ~ \rightsquigarrow~   \boldsymbol{\beta}_{i,.}(t, \Omega_{i,t}) + \gamma_i \left(\boldsymbol{\varepsilon}(t) + (B^{(*)}(\Omega_{i,t})\mathbf{z}_{i,.}(t, \Omega_{i,t})) \otimes \mathbf{z}_{i,.}(t, \Omega_{i,t})\right), \mbox{ and } \boldsymbol{\beta}_i(t, \mathcal{N}/\Omega_{i,t}) = \mathbf{0},
$$
for the path-inddependent learning, where $\boldsymbol{\varepsilon}(\ell) \sim N(0,\sigma_i I_{|\Omega_{i,t}|} )$ and  $\otimes$ is the element-wise matrix product.

\end{proof}

\newpage


\section{: Summary of variables and parameters}\label{app:notation}

\small
\setlength{\LTpre}{0pt}
\setlength{\LTpost}{0pt}

\begin{longtable}{p{0.18\textwidth} p{0.75\textwidth}}
\caption{Main symbols used in the paper (excluding symbols introduced only for auxiliary purposes in the appendices).}
\label{tab:notation}\\
\hline
\textbf{Symbol} & \textbf{Meaning / definition} \\
\hline
\endfirsthead

\hline
\textbf{Symbol} & \textbf{Meaning / definition} \\
\hline
\endhead

\hline
\multicolumn{2}{r}{\emph{(continued on next page)}}\\
\endfoot

\hline
\endlastfoot

$\mathcal{N}$ & Set of sectors, $|\mathcal{N}|=n$.\\
$n$ & Number of sectors in $\mathcal{N}$.\\
$x_i(t)$ & Output (production) of sector $i$ at time $t$.\\
$c_i(t)$ & Consumption of good $i$ at time $t$, for $i\in\mathcal{N}$.\\
$l_i(t)$ & Labor demand (labor input) used by sector $i$ at time $t$.\\
$l_0(t)$ & Labor employed in sector $0$; treated as exogenous.\\
$p_i(t)$ & Equilibrium price of good $i$ at time $t$; $p_0(t)=1$ (numeraire).\\
$w(t)$ & Equilibrium wage (unit salary) at time $t$.\\
$\pi_i(t)$ & Profit of sector $i\in\mathcal{N}$ at time $t$.\\
$\pi(t)$ & Aggregate profit from sectors in $\mathcal{N}$, $\pi(t)=\sum_{i\in\mathcal{N}}\pi_i(t)$.\\
$\pi_0(t)$ & Profit of sector $0$, $\pi_0(t)=l_0(t)\,(1-w(t))$.\\
$E(t)$ & Household income/expenditure at time $t$, $E(t)=w(t)\delta(t)+\pi_0(t)+\pi(t)$.\\
$GDP(t)$ & GDP expression used in Example~\ref{ex:2}.\\
$w^*(t)$, $p_i^*(t)$ & Equilibrium wage and prices when emphasizing equilibrium values.\\

$\rho$ & Discount factor in lifetime utility.\\
$\alpha$ & Preference parameter in quasi-linear utility.\\
$\delta(t)$ & Aggregate labor supply / active population at time $t$ (inelastic).\\

$\eta_i(t)$ & Idiosyncratic productivity shock for sector $i$ at time $t$.\\
$m_i$ & Mean of $\log\eta_i(t)$ under $\log\eta_i(t)\sim\mathcal{N}(m_i,\sigma_i^2)$.\\
$\sigma_i$ & Standard deviation of $\log\eta_i(t)$ (shock magnitude).\\
$q_i$ & Mean productivity, $q_i=\mathbb{E}[\eta_i(t)]=\exp(m_i+\sigma_i^2/2)$.\\
$\mu_i(t)$ & Deterministic production component in $x_i(t)=\eta_i(t)\mu_i(t)$.\\
$\zeta_i^{(\ast)}$ & True (unknown) returns-to-scale parameter, $\zeta_i^{(\ast)}\in[\underline{\zeta},\overline{\zeta}]$.\\
$\zeta_i(t)$ & Firm $i$'s belief/estimate of $\zeta_i^{(\ast)}$ at time $t$ (MAP point estimate).\\
$\underline{\zeta},\overline{\zeta}$ & Exogenous bounds with $\underline{\zeta}>0$ and $\overline{\zeta}<1$.\\
$\hat{\zeta}_i(t)$ & Truncated belief used in decisions:
$\hat{\zeta}_i(t)=\max\{\underline{\zeta},\min\{\overline{\zeta},\zeta_i(t)\}\}$.\\
$\mathcal{I}_i(t)$ & Information set of firm $i$ at date $t$.\\
$\bar{x}_i(t)$ & ``Believed'' output constructed under belief $\zeta_i(t)$.\\
$\bar{p}_i(t)$ & ``Believed'' equilibrium price constructed under belief $\zeta_i(t)$.\\
$l_i(t)$  & Labor choice solving \eqref{eq:firm_problem} (within-period optimization).\\
$L_t(w)$ & Aggregate labor demand function in \eqref{eq:equilibrium_w}.\\
$\varepsilon_i(t)$ & Log shock: $\varepsilon_i(t)=\log(\eta_i(t))$.\\
$z_i(t)$ & Log labor: $z_i(t)=\log(l_i(t))$.\\
$s_i(\ell)$ & Term used in learning section: $s_i(\ell)=\varepsilon_i(\ell)+\zeta_i^{(\ast)}z_i(\ell)$.\\
$\phi(\cdot)$ & Deterministic belief update map.\\
$(\cdot)^+$ & Positive-part operator: $(u)^+=\max\{0,u\}$.\\
$\pi_{0,i}(\zeta)$ & Prior density of $\zeta$ for firm $i$ (zero-truncated Gaussian, Assumption~\ref{ass:3}).\\
$\pi_{t,i}(\zeta)$ & Posterior density at time $t$ (see \eqref{eq:T-posterior}).\\
$\zeta_i^{(0)}$ & Prior location/mean parameter for $\zeta$.\\
$\tau_i$ & Prior uncertainty (standard deviation) parameter.\\
$\gamma_i$ & Precision ratio, $\gamma_i=(\tau_i/\sigma_i)^2$.\\
$\mathscr{L}(\zeta;\mathcal{I})$ & Likelihood of $\zeta$ given information $\mathcal{I}$ (see \eqref{eq:likelihood}).\\

$L_i(t)$ & History of firm $i$'s labor inputs up to $t$: $L_i(t)=\{l_i(\ell)\}_{\ell=1}^{t}$.\\
$\tilde{z}_i^{(1)}(t)$ & $\tilde{z}_i^{(1)}(t)=\sum_{\ell=1}^{t-1} z_i(\ell)$.\\
$\tilde{z}_i^{(2)}(t)$ & $\tilde{z}_i^{(2)}(t)=\sum_{\ell=1}^{t-1} z_i(\ell)^2$.\\
$\overline{v}_i(\cdot)$, $\overline{\varphi}_i(\cdot)$ & Auxiliary terms defined in \eqref{eq:v_bar} for conditional moments (PD).\\
$\tilde{F}_i(t)$, $\tilde{f}_i(t)$ & Shorthands in Proposition~\ref{prop:expect_v_long_memory}.\\
$F(\cdot)$, $f(\cdot)$ & Standard normal CDF and PDF.\\
$G_i(\cdot,\cdot)$, $g_i(\cdot,\cdot)$ & CDF/PDF-based shorthands used to write $\tilde{F}_i(t),\tilde{f}_i(t)$ compactly.\\
$\mathbb{E}[\cdot\mid\cdot]$ & Conditional expectation operator.\\
$\mathbb{M}[\cdot\mid\cdot]$ & Conditional mode operator.\\
$\overline{\Psi}_i$ & Index set used in path-dependent mode analysis (see \eqref{eq:Psi_mode}).\\
$\varrho(\cdot)$ & Ordering/indexing map for $\overline{\Psi}_i$.\\

$\overline{\overline{v}}_{i,t}(\cdot)$, $\overline{\overline{\varphi}}_i(\cdot)$
& Auxiliary terms in Proposition~\ref{prop:expect_v_short_memory} for conditional moments (PI).\\
$\Psi_i$ & Index set used in PI mode analysis (Proposition~\ref{prop:convergence_mode_short_memory}).\\
$\psi_i(\cdot)$ & Ordering/indexing map for $\Psi_i$.\\
$\tilde{\mathfrak{1}}_{i,t}$ & Indicator used in Proposition~\ref{prop:convergence_mode_short_memory}.\\

$\lambda_{\min},\lambda_{\max}$ & Exogenous labor bounds in Proposition~\ref{prop:l_i_bound}.\\
$\kappa_{i,t}$ & $\kappa_{i,t} = -\left(1-\hat{\zeta}_i(t)\alpha \right)^{-1}$.\\
$\chi_{i,t}$ & $\chi_{i,t}=\left(\mathbb{E}[\eta_i(t)^\alpha]\hat{\zeta}_i(t)\right)^{-\kappa_{i,t}}$.\\

$v(t)$ & Generic latent state in the Bayesian-filter illustration (context-specific).\\
$x(t)$ & Generic observation in the Bayesian-filter illustration (context-specific).\\
$\omega(t)$ & Generic process noise in the Bayesian-filter illustration.\\
$\nu(\cdot)$ & Generic state transition map in the Bayesian-filter illustration.\\
$\mu(\cdot)$ & Generic observation map in the Bayesian-filter illustration.\\
$A(u)$ & $A(u)=\dfrac{1}{1-u}\log \mathbb{E}[\eta_i(t)^\alpha]$ (Proposition~\ref{prop:comparison_learning}).\\
$B(u,w)$ & $B(u,w)=\dfrac{1}{(1-u)\alpha}\log\!\left(\dfrac{u}{w}\right)$ (Proposition~\ref{prop:comparison_learning}).\\
$\sigma_i^*(u,w)$ & $\sigma_i^*(u,w)=\left(\dfrac{\sigma_i}{A(u)+B(u,w)}\right)^2$ (Proposition~\ref{prop:comparison_learning}).\\
$a_t^{PD},a_t^{PI}$ & Time-varying weights in 
Proposition~\ref{prop:comparison_learning}.\\
$\varepsilon_i^*(t)$ & Gaussian ``sufficient statistic'' in Proposition~\ref{prop:comparison_learning}.\\

$y_{i,j}(t)$ & Amount of good $j$ used as intermediate input by sector $i$ at time $t$.\\
$Y(t)$ & Matrix with entries $y_{i,j}(t)$.\\
$\phi\in[0,1]$ & Material-input intensity parameter in the multi-input technology (distinct from update map $\phi(\cdot)$ by context).\\
$\beta_{i,j}$ & Input-output elasticity parameter for input $j$ in sector $i$ (unknown and learned).\\
$\beta_{i,j}^{(\ast)}$ & True (unknown) input-output elasticity.\\
$\beta_{i,j}^{(0)}$ & Prior location/mean for $\beta_{i,j}$.\\
$B$ & $n\times n$ matrix collecting $\beta_{i,j}$.\\
$B^{(\ast)}$ & True elasticity matrix with entries $\beta_{i,j}^{(\ast)}$.\\
$s_i(t)$ & (High-dimensional section) $s_i(t)=\sigma_i\varepsilon_i+\sum_{j=1}^n \phi \beta_{i,j}^{(\ast)}\log y_{i,j}(t)$.\\
$z_{i,j}(t)$ & (High-dimensional section) $z_{i,j}(t)=\phi\log y_{i,j}(t)$.\\
$\mathcal{U}$ & Ordered subset of indices used to form sub-vectors/matrices.\\
$\mathbf{z}_i(t,\mathcal{U})$ & Sub-vector collecting $\{z_{i,j}(t)\}_{j\in\mathcal{U}}$.\\
$I_h$ & Identity matrix of dimension $h$.\\
$H_i(t',t,\mathcal{U})$ & Matrix defined in \eqref{eq:H_matrix}.\\
$\mathbf{b}_i(t',t,\mathcal{U})$ & Vector defined in \eqref{eq:B_vector}.\\
$\Omega_{i,t}$ & Collection of index sets defined before Proposition~\ref{prop:multidim_MAP}.\\
$\boldsymbol{\beta}_i(t,\mathcal{U})$ & Sub-vector of firm $i$'s MAP estimate at time $t$ for indices in $\mathcal{U}$.\\
$\mathbf{0}$ & Zero vector (dimension clear from context).\\
$\otimes$ & Element-wise (Hadamard) product.\\
$\rightsquigarrow$ & Asymptotic-equivalence notation used in Corollary~\ref{cor:f_conditional}.\\

\end{longtable}

\normalsize

\newpage


\section{: Probability of applying a rule of thumb}\
\label{Section:appendix2}


In this appendix, the inquiry pertains to the likelihood of the analyzed dynamic system hitting a state where firms dismiss the empirical estimation of $\zeta_{i}(t)$ and adopt a rule of thumb. Using \eqref{prop:v_learning_long_mem} and noticing that $\varepsilon_{i}(t) = \mathcal{N}(m_i, \sigma_i^2)$, we have
$$
\begin{array}{ll}
\mathbb{P}\left(\zeta_i(t) \leq \underline{\zeta} ~|~ \mathcal{I}_i(t-1) \right)
    & = \mathbb{P}\left( \dfrac{\zeta^{(0)}_i  + \gamma_i \sum_{\ell=1}^{t} z_i(\ell) s_i(\ell)}{\left(1 + \gamma_i \sum_{\ell=1}^{t} z_i(\ell)^2\right)} \leq \underline{\zeta} ~\Big|~ \mathcal{I}_i(t-1) \right) \\[0.6cm]
    & = \mathbb{P}\left( \zeta^{(0)}_i  + \gamma_i \sum_{\ell=1}^{t} z_i(\ell) s_i(\ell) \leq \underline{\zeta}\left(1 + \gamma_i \sum_{\ell=1}^{t} z_i(\ell)^2\right) ~\Big|~ \mathcal{I}_i(t-1) \right) \\[0.6cm]
    & = \mathbb{P}\left( (\zeta^{(0)}_i - \underline{\zeta})  + \gamma_i (\zeta_i^{(*)} - \underline{\zeta}) \sum_{\ell=1}^{t} z_i(\ell)^2  \leq - \gamma_i \sum_{\ell=1}^{t} z_i(\ell) \varepsilon_{i}(\ell)
     ~\Big|~ \mathcal{I}_i(t-1) \right) \\[0.6cm]  
    & = \mathbb{P}\left( (\zeta^{(0)}_i - \underline{\zeta})  + \gamma_i (\zeta_i^{(*)} - \underline{\zeta}) \sum_{\ell=1}^{t} z_i(\ell)^2  \leq Q_{i,t} ~\Big|~ \mathcal{I}_i(t-1) \right) 
\end{array}
$$
where 
$$
Q_{i,t} = - \gamma_i \sum_{\ell=1}^{t} z_i(\ell) \varepsilon_{i}(\ell) ~\sim~ \mathcal{N}\left( - m_i \gamma_i \sum_{\ell=1}^{t} z_i(\ell), ~ \sigma_i^2 \gamma_i^2 \sum_{\ell=1}^{t} z_i(\ell)^2  \right)
$$
Therefore,
$$
\begin{array}{ll}
\mathbb{P}\left(\zeta_i(t) \leq \underline{\zeta} ~|~ \mathcal{I}_i(t-1) \right)
    & \mathbb{P}\left( (\zeta^{(0)}_i - \underline{\zeta})  + \gamma_i (\zeta_i^{(*)} - \underline{\zeta}) \sum_{\ell=1}^{t} z_i(\ell)^2  \leq Q_{i,t} ~\Big|~ \mathcal{I}_i(t-1) \right)  \\[0.6cm]
    & = 1 - F\left( \dfrac{\sigma^2(\zeta^{(0)}_i - \underline{\zeta})  + \tau_i^2 (\zeta_i^{(*)} - \underline{\zeta}) \sum_{\ell=1}^{t} z_i(\ell)^2 + m_i \tau_i^2 \sum_{\ell=1}^{t} z_i(\ell)}{\sigma_i \tau_i^2 \sqrt{\sum_{\ell=1}^{t} z_i(\ell)^2}} \right) \\[0.6cm]
     & = 1 - F\left( \tilde{t}(\sigma_i, \tau_i)  \right),
\end{array}
$$
where $Z$ is a standardized Gaussian and $F$ is the corresponding probability distribution function. We have, 
$$
\tilde{t}(\sigma_i, \tau_i) = \dfrac{\sigma^2(\zeta^{(0)}_i - \underline{\zeta})}{\sigma_i \tau_i^2 \sqrt{\sum_{\ell=1}^{t} z_i(\ell)^2}} + \dfrac{ \tau_i^2 (\zeta_i^{(*)} - \underline{\zeta}) \sum_{\ell=1}^{t} z_i(\ell)^2 + m_i \tau_i^2 \sum_{\ell=1}^{t} z_i(\ell)}{\sigma_i \tau_i^2 \sqrt{\sum_{\ell=1}^{t} z_i(\ell)^2}}.
$$
Notably, $\mathbb{P}\left(\zeta_i(t) \leq \underline{\zeta} ~|~ \mathcal{I}_i(t-1) \right)$ is decreasing with respect to $\tau_i$ (the uncertainty of firm $i$'s prior knowledge) and $\sigma_i$ (the magnitude of idiosyncratic production shocks). We have
$$
\underset{\sigma_i \rightarrow 0}{\lim} ~ \tilde{t}(\sigma_i, \tau_i) =  +\infty \quad \mbox{ and } \quad \underset{\tau_i \rightarrow 0}{\lim} ~ \tilde{t}(\sigma_i, \tau_i) =  + \infty 
$$
Since $F$ is continuous, increasing and $\lim_{t \rightarrow -\infty} F(t) = 0$, then 
$$
\left\{\begin{array}{ll}
\underset{\sigma_i \rightarrow 0}{\lim} ~ \mathbb{P}\left(\zeta_i(t) = 0 ~|~ \mathcal{I}_i(t-1) \right) & = 0, \\[0.5cm]
\underset{\tau_i \rightarrow 0}{\lim} ~ \mathbb{P}\left(\zeta_i(t) = 0 ~|~ \mathcal{I}_i(t-1) \right) & = 0.
\end{array}\right.
$$

Therefore, the learning path remains invariant with respect to $\underline{\zeta}$ when a sector has strong prior knowledge and experiences only minor productivity fluctuations.

\newpage

\RRew{
\section{: Inclusion of capital market}\label{Section:appendix_C}

In this appendix, we extend the baseline model by incorporating capital as an additional production input. We show that this extension does not alter the structure of the learning dynamics analyzed in the paper. In particular, even when capital is introduced, the production function can be re-expressed in the same reduced form used throughout the main text, so that the characterization of firms' learning paths remains unchanged.

\paragraph{Consumer side.} 
Let prices for non-numeraire goods be denoted by \(p_i(t)\), with \(p_0(t) \equiv 1\). The representative consumer can trade a one-period asset \(a(t+1)\) (interpreted as a bond or riskless capital) that pays one unit of the numeraire in period \(t+1\) and is sold at price $1/R_t$ in period $t$. Income is given by $E(t)$. With an initial asset endowment $a(1)$ and the standard no-Ponzi/transversality condition, the budget constraint reads
\begin{equation}\label{eq:budget_constr}
    c_0(t)+\sum_{i=1}^n p_i(t)c_i(t) + a(t+1)
= E(t) + R_t\,a(t) \quad \text{for all } t \geq 1.
\end{equation}
The representative consumer maximizes lifetime utility
$$
\max_{\{c_0(t),c_i(t),a_{t+1}\}_{t\ge1}} 
\sum_{t=1}^{\infty}\rho^t \left[ c_0(t)+\sum_{i=1}^n \frac{c_i(t)^{\alpha}}{\alpha} \right],
$$
subject to the budget constraints \eqref{eq:budget_constr} and the transversality condition \(\lim_{T\to\infty}\rho^T a(T+1)=0\).  
Introducing Lagrange multipliers \(\{\lambda_t\}_{t\ge1}\), the Lagrangian is
$$
\mathcal{L}=\sum_{t=1}^{\infty}\Bigg\{
\rho^t\!\left[c_0(t)+\sum_{i=1}^n\frac{c_i(t)^{\alpha}}{\alpha}\right]
+\lambda_t\Big(E(t)+R_t a(t) - c_0(t)-\sum_{i=1}^n p_i(t)c_i(t)-a(t+1)\Big)
\Bigg\}.
$$
The first-order conditions are
$$
\left\{
\begin{array}{ll}
\displaystyle \frac{\partial\mathcal{L}}{\partial c_0(t)}: & \rho^t - \lambda_t = 0 \quad \Rightarrow\quad \lambda_t=\rho^t,\\[0.4cm]
\displaystyle \frac{\partial\mathcal{L}}{\partial c_i(t)}: & \rho^t c_i(t)^{\alpha-1} - \lambda_t p_i(t)=0 
\quad \Rightarrow\quad p_i(t)=c_i(t)^{\alpha-1},\\[0.4cm]
\displaystyle \frac{\partial\mathcal{L}}{\partial a(t+1)}: & -\lambda_t+\lambda_{t+1}R_{t+1}=0
\quad \Rightarrow\quad \lambda_t=\lambda_{t+1}R_{t+1}.
\end{array}\right.
$$
Hence $\rho^t=\rho^{t+1}R_{t+1}$, which implies \(1=\rho R_{t+1}\).

\paragraph{Firm side.}
Assume that the good $0$ can be used both as consumption and
capital good. Assume that capital has full depreciation in one period. Each firm $i\in\mathcal{N}$ chooses labor $l_{i}(t)$ and capital $k_{i}(t)$ as inputs according to the technology
$$
x_{i}(t)=e_{i}(t)l_{i}(t)^{\Upsilon_{i}^{(\ast)}}k_{i}(t)^{\iota_{i}}%
$$
where $e_{i}(t)$ is a firm-specific productivity term; $\iota_{i}\in\left(
0,1\right)$ is known by the firm; $\Upsilon_{i}^{(\ast)}$ is the parameter to be learned. The perceived production technology is
\[
x_{i}(t)=e_{i}(t)l_{i}(t)^{\Upsilon_{i}(t)}k_{i}(t)^{\iota_{i}}.
\]
With gross capital return $R_{t}=1/\rho$, profit maximization becomes:%
\[
\underset{\left(  l_{i}(t),k_{i}(t)\right)  }{\max}\mathbb{E}\left[  \left.
p_{i}(t)e_{i}(t)\right\vert \mathcal{I}_{i}(t)\right]  l_{i}(t)^{\Upsilon
_{i}(t)}k_{i}(t)^{\iota_{i}}-w(t)l_{i}(t)-\frac{k_{i}(t)}{\rho}.
\]
The term $\mathbb{E}\left[  \left.  p_{i}(t)e_{i}(t)\right\vert \mathcal{I}%
_{i}(t)\right]$ is denoted $P_{i,t}$. We have:
$$
\Upsilon_{i}(t)P_{i,t}l_{i}(t)^{\Upsilon_{i}(t)-1}k_{i}(t)^{\iota_{i}} 
=w(t)\quad \mbox{ and } \quad \iota_{i}(t)P_{i,t}l_{i}(t)^{\Upsilon_{i}(t)}k_{i}(t)^{\iota_{i}-1} 
=\frac{1}{\rho}.
$$
The ratio of these two conditions gives:%
\[
\frac{k_{i}(t)}{l_{i}(t)}=\frac{w(t)\rho\iota_{i}}{\Upsilon_{i}(t)}.
\]
Replacing $k_{i}(t)$ in the first condition leads to:%
\[
\Upsilon_{i}(t)P_{i,t}l_{i}(t)^{\Upsilon_{i}(t)-1}\left(  \frac{l_{i}%
(t)w(t)\rho\iota_{i}}{\Upsilon_{i}(t)}\right)  ^{\iota_{i}}=w(t)
\]
or%
\[
l_{i}(t)=P_{i,t}^{\frac{1}{1-\Upsilon_{i}(t)-\iota_{i}}}\left(  \rho\iota
_{i}\right)  ^{\frac{\iota_{i}}{1-\Upsilon_{i}(t)-\iota_{i}}}\left[
\frac{\Upsilon_{i}(t)}{w(t)}\right]  ^{\frac{1-\iota_{i}}{1-\Upsilon
_{i}(t)-\iota_{i}}}.
\]
Then we get:
\[
k_{i}(t)=P_{i,t}^{\frac{1}{1-\Upsilon_{i}(t)-\iota_{i}}}\left(  \rho\iota
_{i}\right)  ^{\frac{1-\Upsilon_{i}(t)}{1-\Upsilon_{i}(t)-\iota_{i}}}\left[
\frac{\Upsilon_{i}(t)}{w(t)}\right]  ^{\frac{\Upsilon_{i}(t)}{1-\Upsilon
_{i}(t)-\iota_{i}}}.
\]
The expected production of the firm is:%
\[
\bar{x}_{i}(t)=e_{i}(t)P_{i,t}^{\frac{\Upsilon_{i}(t)+\iota_{i}}%
{1-\Upsilon_{i}(t)-\iota_{i}}}\left(  \rho\iota_{i}\right)  ^{\frac{\iota_{i}%
}{1-\Upsilon_{i}(t)-\iota_{i}}}\left[  \frac{\Upsilon_{i}(t)}{w(t)}\right]
^{\frac{\Upsilon_{i}(t)}{1-\Upsilon_{i}(t)-\iota_{i}}},
\]
which is related to the consumer demand: $\bar{x}_{i}(t)=\left[  \bar{p}_{i}(t)\right]  ^{\frac{1}{1-\alpha}}$. Thus,
\[
\bar{p}_{i}(t)=\bar{x}_{i}(t)^{\alpha-1}=e_{i}(t)^{\alpha-1}P_{i,t}%
^{\frac{\left[  \Upsilon_{i}(t)+\iota_{i}\right]  \left(  \alpha-1\right)
}{1-\Upsilon_{i}(t)-\iota_{i}}}\left(  \rho\iota_{i}\right)  ^{\frac{\iota
_{i}\left(  \alpha-1\right)  }{1-\Upsilon_{i}(t)-\iota_{i}}}\left[
\frac{\Upsilon_{i}(t)}{w(t)}\right]  ^{\frac{\Upsilon_{i}(t)\left(
\alpha-1\right)  }{1-\Upsilon_{i}(t)-\iota_{i}}}.
\]
Finally,
\begin{align*}
\mathbb{E}\left[  \bar{p}_{i}(t)e_{i}(t)\right]    & =P_{i,t}=\mathbb{E}%
\left[  e_{i}(t)^{\alpha}\right]  P_{i,t}^{\frac{\left[  \Upsilon_{i}%
(t)+\iota_{i}\right]  \left(  \alpha-1\right)  }{1-\Upsilon_{i}(t)-\iota_{i}}%
}\left(  \rho\iota_{i}\right)  ^{\frac{\iota_{i}\left(  \alpha-1\right)
}{1-\Upsilon_{i}(t)-\iota_{i}}}\left[  \frac{\Upsilon_{i}(t)}{w(t)}\right]
^{\frac{\Upsilon_{i}(t)\left(  \alpha-1\right)  }{1-\Upsilon_{i}(t)-\iota_{i}%
}}\\
P_{i,t}^{1+\frac{\left[  \Upsilon_{i}(t)+\iota_{i}\right]  \left(
1-\alpha\right)  }{1-\Upsilon_{i}(t)-\iota_{i}}}  & =\mathbb{E}\left[
e_{i}(t)^{\alpha}\right]  \left(  \rho\iota_{i}\right)  ^{\frac{\iota
_{i}\left(  \alpha-1\right)  }{1-\Upsilon_{i}(t)-\iota_{i}}}\left[
\frac{\Upsilon_{i}(t)}{w(t)}\right]  ^{\frac{\Upsilon_{i}(t)\left(
\alpha-1\right)  }{1-\Upsilon_{i}(t)-\iota_{i}}}\\
P_{i,t}^{\frac{1}{1-\Upsilon_{i}(t)-\iota_{i}}}  & =\mathbb{E}\left[
e_{i}(t)^{\alpha}\right]  ^{\frac{1}{1-\alpha\left[  \Upsilon_{i}(t)+\iota
_{i}\right]  }}\left\{  \left(  \rho\iota_{i}\right)  ^{\frac{\iota_{i}%
}{1-\Upsilon_{i}(t)-\iota_{i}}}\left[  \frac{\Upsilon_{i}(t)}{w(t)}\right]
^{\frac{\Upsilon_{i}(t)}{1-\Upsilon_{i}(t)-\iota_{i}}}\right\}  ^{\frac
{\left(  \alpha-1\right)  }{1-\alpha\left[  \Upsilon_{i}(t)+\iota_{i}\right]
}}.
\end{align*}
Labor and capital expressions are then:%
\begin{align*}
l_{i}(t)  & =\mathbb{E}\left[  e_{i}(t)^{\alpha}\right]  ^{\frac{1}%
{1-\alpha\left[  \Upsilon_{i}(t)+\iota_{i}\right]  }}\left(  \rho\iota
_{i}\right)  ^{\frac{\iota_{i}\alpha}{1-\alpha\left[  \Upsilon_{i}%
(t)+\iota_{i}\right]  }}\left[  \frac{\Upsilon_{i}(t)}{w(t)}\right]
^{\frac{1-\iota_{i}\alpha}{1-\alpha\left[  \Upsilon_{i}(t)+\iota_{i}\right]
}},\\
k_{i}(t)  & =\mathbb{E}\left[  e_{i}(t)^{\alpha}\right]  ^{\frac{1}%
{1-\alpha\left[  \Upsilon_{i}(t)+\iota_{i}\right]  }}\left(  \rho\iota
_{i}\right)  ^{\frac{1-\alpha \Upsilon_{i}(t)}{1-\alpha\left[  \Upsilon
_{i}(t)+\iota_{i}\right]  }}\left[  \frac{\Upsilon_{i}(t)}{w(t)}\right]
^{\frac{\alpha \Upsilon_{i}(t)}{1-\alpha\left[  \Upsilon_{i}(t)+\iota
_{i}\right]  }}.
\end{align*}
The expression of $l_{i}(t)$ has the same form as \eqref{eq:labor_eq}. In order to show this, it is possible to use the change of variable:%
\[
\hat{\zeta}_{i}(t)=\frac{\Upsilon_{i}(t)}{1-\alpha\iota_{i}}.
\]
Therefore, $l_{i}(t)$ becomes
\[
l_{i}(t)=\left\{  \mathbb{E}\left[  e_{i}(t)^{\alpha}\right]  ^{\frac
{1}{\left[  1-\alpha \hat{\zeta}_{i}(t)\right]  }}\left(  \rho\iota_{i}\right)
^{\frac{\iota_{i}\alpha}{\left[  1-\alpha \hat{\zeta}_{i}(t)\right]  }}\left(
1-\alpha\iota_{i}\right)  \right\}  ^{\frac{1}{\left(  1-\alpha\iota
_{i}\right)  }}\left[  \frac{\hat{\zeta}_{i}(t)}{w(t)}\right]  ^{\frac{1}{\left[
1-\alpha \hat{\zeta}_{i}(t)\right]  }},
\]
which shows that the analysis of the labor market remains the same. The equilibrium on the capital market will not add any new condition. As capital can be traded between two periods at a constant price $1/\rho,$ the capital supply directly follows the value of capital demand. Thus, the inclusion of capital simply rescales productivity and the effective return to scale on labor, while preserving the reduced-form representation used in the main analysis. Consequently, the learning dynamics over \(\zeta_i(t)\) remain structurally identical to those in the baseline model. }


\newpage

\RRew{
\section{: Endogenous labor supply}\label{Section:appendix_D}

To endogenize the labor supply, we modify the consumer problem \eqref{eq:consumer_model_endogenous labor} and assume that at each period $t$, the representative consumer solves
\begin{equation}\label{eq:consumer_model_endogenous labor}
\max_{\{c_0(t),\,c_i(t),\,\delta(t)\}} 
\rho^{t}\!\left(c_{0}(t)+\sum_{i=1}^{n}\frac{c_{i}(t)^{\alpha}}{\alpha}-\delta(t)^{r}\right)
\quad\text{s.t.}\quad
c_{0}(t)+\sum_{i=1}^{n}p_{i}(t)c_{i}(t)=E(t),
\end{equation}
with \(E(t)=w(t)\delta(t)+\pi_{0}(t)+\pi(t)\), \(p_{0}(t)=1\), \(0<\alpha<1\), and \(r>1\).
(The budget constraint binds by monotonicity of utility in $c_0$, $c_i$, and $\delta(t)$.) The objective function of problem \eqref{eq:consumer_model_endogenous labor} includes a disutility of labor. Let us introduce the multiplier \(\lambda_t\) and consider the Lagrangian
\[
\mathcal{L}_t=\rho^{t}\!\left(c_{0}(t)+\sum_{i=1}^{n}\frac{c_{i}(t)^{\alpha}}{\alpha}-\delta(t)^{r}\right)
+\lambda_t\!\left(w(t)\delta(t)+\pi_{0}(t)+\pi(t)-c_{0}(t)-\sum_{i=1}^{n}p_{i}(t)c_{i}(t)\right).
\]
First-order conditions are
$$
\left\{
\begin{array}{lllll}
\displaystyle \frac{\partial\mathcal{L}_t}{\partial c_0(t)}&:\quad \rho^{t}-\lambda_t=0 &\Rightarrow& \lambda_t&=\rho^{t},\\[0.4cm]
\displaystyle \frac{\partial\mathcal{L}_t}{\partial c_i(t)}&:\quad \rho^{t}c_i(t)^{\alpha-1}-\lambda_t p_i(t)=0
&\Rightarrow& p_i(t)&=c_i(t)^{\alpha-1},\\[0.4cm]
\displaystyle \frac{\partial\mathcal{L}_t}{\partial \delta(t)}&:\quad -\rho^{t} r\,\delta(t)^{\,r-1}+\lambda_t w(t)=0
&\Rightarrow& r\,\delta(t)^{\,r-1}&=w(t).
\end{array}\right.
$$
Using \(\lambda_t=\rho^{t}\), the optimal labor supply is therefore
\begin{equation}\label{eq:delta_endogenous_labor}
    \delta^{\ast}(t)=\left(\frac{w(t)}{r}\right)^{\!\frac{1}{\,r-1\,}}\;\qquad (r>1,\; w(t)\ge 0).
\end{equation}

We now reproduce the existence and uniqueness of the equilibrium wage, presented in Lemma \ref{lemma:L}, for the case of endogenous labor supply.

\begin{lemma}\label{lemma:L_endog}
Let us define the aggregate private-sector labor demand
$L_t(w)\;=\;\sum_{i=1}^n l_i(t)$ (defined in \eqref{eq:equilibrium_w}), where each $l_i(t)$ is given by the firms' labor demand (defined in \eqref{eq:labor_eq}). Then, there exists a unique \(w^*(t)>0\) such that 
$\delta^{\ast}(t)-l_0(t) = L_t\big(w^*(t)\big)$, where $\delta^{\ast}(t)$ is endogenously defined in \eqref{eq:delta_endogenous_labor}.
\end{lemma}

\begin{proof}
We start by recalling the functional form of each firm's labor demand from the equilibrium conditions. As in the exogenous-supply case, we may write
$l_i(t) = \chi_{i,t} w(t)^{\kappa_{i,t}}$, 
where
$$
\kappa_{i,t} = - \dfrac{1 }{1 - \hat{\zeta}_{i}(t)\alpha } \quad \mbox{ and } \quad 
\chi_{i,t}  = \left(\mathbb{E}\left[ \eta_{i}(t)^{\alpha}  \right] \hat{\zeta}_{i}(t) \right)^{\dfrac{1 }{1 - \hat{\zeta}_{i}(t)\alpha } } .
$$
Using \(0<\alpha<1\) and \(\hat{\zeta}_i(t)\alpha<1\), each \(l_i(t)\) is continuous in \(w\) and strictly decreasing (since \(\kappa_{i,t}<0\)), and the aggregate demand
\[
L_t(w)=\sum_{i=1}^n \chi_{i,t}\,w^{\kappa_{i,t}}
\]
is continuous and strictly decreasing on $\mathbb{R}_+$. In particular \(L_t(w)\to +\infty\) as \(w\downarrow 0\) and \(L_t(w)\to 0\) as \(w\to+\infty\). Next, the endogenous supply
\[
\delta^{\ast}(t)=\Big(\tfrac{w}{r}\Big)^{1/(r-1)}
\]
is continuous and strictly increasing in \(w\) for \(r>1\); moreover \(\delta^{\ast}(w)\to 0\) as \(w\downarrow 0\) and \(\delta^{\ast}(w)\to +\infty\) as \(w\to +\infty\). Define \(G_t(w)=\delta^{\ast}(t)-l_0(t)-L_t(w)\). Since \(\delta^{\ast}(w)\) and \(L_t(w)\) are continuous, \(G_t\) is continuous on \(\mathbb{R}_+\). Differentiating, yields
$$
G_t'(w)=\frac{d\delta^{\ast}}{dw}(w)-L_t'(w).
$$
Because \(\dfrac{d\delta^{\ast}}{dw}(w)>0\) and \(L_t'(w)<0\), we have \(G_t'(w)>0\) for all \(w>0\). Hence, \(G_t\) is strictly increasing on \(\mathbb{R}_+\). Finally, the end-point behavior follows from the limits of \(\delta^{\ast}\) and \(L_t\). In fact, as \(w\downarrow 0\),
\(\delta^{\ast}(w)\to 0\) while \(L_t(w)\to +\infty\), so \(G_t(w)\to -\infty\). Instead, as \(w\to +\infty\),
\(\delta^{\ast}(w)\to +\infty\) while \(L_t(w)\to 0\). Therefore,
$$
\displaystyle \lim_{w\to 0} G_t(w)=-\infty \quad \mbox{ and } \quad \lim_{w\to +\infty} G_t(w)=+\infty.
$$
By continuity and strict monotonicity, there exists a unique \(w^*(t)>0\) such that \(G_t(w^*(t))=0\), i.e.
$$
\delta^{\ast}(t)-l_0(t)=L_t\big(w^*(t)\big).
$$
\end{proof}

Beyond the consistency between Lemma \ref{lemma:L} and Lemma \ref{lemma:L_endog}, for the existence and uniqueness of the equilibrium wage under exogenous (Lemma \ref{lemma:L}) and endogenous (Lemma \ref{lemma:L_endog}) labor supply, we explore hereafter the generalizability of Proposition \ref{prop:demography_long_mem} when labor supply is endogenously obtained. 

\begin{proposition}[Endogenous labor supply and path--dependent expansion]\label{prop:demography_endog} Define the unique equilibrium wage $w^{\ast}(t)$ by the market--clearing condition under endogenous supply, as established in Lemma \ref{lemma:L_endog}. Then there exists a monotone sequence of government labor demands $\{l_0(t)\}_t$ and the associated equilibrium sequence $\{w^{\ast}(t)\}_t$ with $w^{\ast}(t)\nearrow\infty$ such that, letting
\[
\tilde z_i^{(1)}(t)=\sum_{\ell=1}^{t} \ln\!\Big(\xi_{i,\ell}\,[w^{\ast}(\ell)]^{\kappa_{i,\ell}}\Big),
\qquad
\tilde z_i^{(2)}(t)=\sum_{\ell=1}^{t} \Big[\ln\!\big(\xi_{i,\ell}\,[w^{\ast}(\ell)]^{\kappa_{i,\ell}}\big)\Big]^2,
\]
for any fixed $\xi_{i,\ell}>0$, the following limits hold:
\[
\lim_{t\to\infty} \big|\tilde z_i^{(1)}(t)\big|=+\infty,
\qquad
\lim_{t\to\infty} \tilde z_i^{(2)}(t)=+\infty,
\]
and the ratio in case \textup{(5)} of Proposition~\ref{prop:long_memory_learning} diverges in magnitude:
\[
\lim_{t\to\infty}\frac{\zeta_i^{(0)}+\gamma_i \zeta_i^{(*)}\,\tilde z_i^{(2)}(t)}{\gamma_i \sigma_i \,\tilde z_i^{(1)}(t)}
=\pm\infty.
\]
In particular, the economy exhibits the same path--dependent long-memory behavior as in the exogenous supply case.\footnote{If one requires the limit to be $+\infty$ (rather than $\pm \infty$), it suffices to implement a monotone \emph{decreasing} wage path $w^{\ast}(t)\searrow 0$ via an appropriate sequence $\{l_0(t)\}_t$; the proof is identical with signs reversed.}
\end{proposition}

\begin{proof}
Define \(G_t(w)=\delta^{\ast}(w)-l_0(t)-L_t(w)\). By Lemma~\ref{lemma:L_endog} the equation
\(G_t(w)=0\) admits a unique solution \(w^\ast(t)>0\) for any admissible $l_0(t)$. To show that an equilibrium wage path can be made strictly increasing and unbounded, proceed constructively. At each period $t$, consider the function
$$
L^{(0)}_t(w)\;:=\;\delta^{\ast}(w)-L_t(w).
$$
From Lemma \ref{lemma:L_endog}, we have
$$
\lim_{w\to 0^+}L^{(0)}_t(w)=-\infty \quad \mbox{ and }  \quad \lim_{w\to\infty}L^{(0)}_t(w)=+\infty.
$$
Hence, since $L_t(w)$ is continuous and strictly decreasing on $\mathbb{R}_+$, there exists a finite threshold \(W_t>0\) such that \(L^{(0)}_t(w)>0\) for all \(w\ge W_t\). Note that, a-priori, $W_t$ may depend on $t$ because the coefficients $\chi_{i,t}$ and exponents $\kappa_{i,t}$ entering $L_t(w)=\sum_i \chi_{i,t}w^{\kappa_{i,t}}$ can vary across dates; thus the tail behaviour of $L_t(w)$ need not be uniform in $t$. However, since $\hat{\zeta}_i(t) = \max\left\{ \underline{\zeta}, \min\left\{ \overline{\zeta}, \zeta_i(t) \right\} \right\}$, then there exist constants $\bar\chi_i>0$ and $\bar\kappa<0$ such that for every $t$ and every $i\in\{1,\dots,n\}$ we have $\chi_{i,t}\le \bar\chi_i$ and $\kappa_{i,t}\le \bar\kappa<0$. This implies that there exists a single finite $W>0$ such that
$$
L^{(0)}_t(w)=\delta^\ast(w)-L_t(w)>0\quad\text{for all }w\ge W \text{ and all } t.
$$
To see why this is the case, note that, under the stated bounds,
$$
L_t(w)=\sum_{i=1}^n \chi_{i,t}w^{\kappa_{i,t}}
\le \sum_{i=1}^n \bar\chi_i w^{\bar\kappa}=: \bar L(w).
$$
Because $\bar\kappa<0$ we have $\bar L(w)\to 0$ as $w\to\infty$. Since $\delta^\ast(w)=(w/r)^{1/(r-1)}\to\infty$ as $w\to\infty$, there exists $W>0$ large enough that $\delta^\ast(w)-\bar L(w)>0$ for all $w\ge W$. For such $w$ and any $t$,
$$
L^{(0)}_t(w)=\delta^\ast(w)-L_t(w)\ge \delta^\ast(w)-\bar L(w)>0,
$$
which proves the existence of a uniform bound $W$. Now fix any strictly increasing unbounded sequence \(\{\bar w(t)\}_{t\ge1}\) with \(\bar w(t)\ge W\). For each \(t\) define the gouvernamental demand for labor as
$$
\bar{l}_0(t)\;:=\;\delta^{\ast}\big(\bar{w}(t)\big)-L_t\big(\bar{w}(t)\big).
$$
By construction \(\bar{l}_0(t)\ge 0\), and since \(L_t(\cdot)\ge 0\) we also have \(\bar{l}_0(t)\le \delta^{\ast}(\bar w(t))\), so $\bar{l}_0(t)$ is feasible. Moreover,
\[
G_t\big(\bar w(t)\big)=\delta^{\ast}\big(\bar w(t)\big)- \bar{l}_0(t) -L_t\big(\bar{w}(t)\big)=0.
\]
The uniqueness of the root of \(G_t(\cdot)=0\) implies \(w^\ast(t)=\bar w(t)\) for every \(t\). Since \(\{\bar w(t)\}\) was chosen strictly increasing and unbounded, the equilibrium wage path satisfies
\(w^\ast(t)\nearrow\infty\). Using $l_i(t)=\chi_{i,t}\,[w^{\ast}(t)]^{\kappa_{i,t}}$ with $\kappa_{i,t}< -1$, because $w^{\ast}(t)\to\infty$, we have
\[
\ln\!\Big(\xi_{i,t}\,[w^{\ast}(t)]^{\kappa_{i,t}}\Big)
=\ln(\xi_{i,t})+\kappa_{i,t}\,\ln w^{\ast}(t) \to -\infty.
\]
Hence each summand in $\tilde{z}_i^{(2)}(t)$ grows without bound and $\tilde z_i^{(2)}(t)\to +\infty$, while the partial sums $\tilde z_i^{(1)}(t)$ tend to $-\infty$ so that $|\tilde z_i^{(1)}(t)|\to\infty$. For the ratio, apply l'Hopital to the continuous-time interpolation (or discrete step argument as in the baseline proof): for large $w$,
\[
\frac{d}{d\,w}\Big[\ln\!\big(\xi_{i,t} w^{\kappa_{i,t}}\big)\Big]
=\frac{\kappa_{i,t}}{w}\;<0,
\qquad
\frac{d}{d\,w}\Big[\ln\!\big(\xi_{i,t} w^{\kappa_{i,t}}\big)\Big]^2
=\frac{2\kappa_{i,t}}{w}\,\ln\!\big(\xi_{i,t} w^{\kappa_{i,t}}\big).
\]
Therefore,
\[
\frac{d\,\tilde z_i^{(2)}(t)}{d\,\tilde z_i^{(1)}(t)}
=\frac{\frac{2\kappa_{i,t}}{w}\,\ln\!\big(\xi_{i,t} w^{\kappa_{i,t}}\big)}
{\frac{\kappa_{i,t}}{w}}
=2\,\ln\!\big(\xi_{i,t} w^{\kappa_{i,t}}\big)\;\longrightarrow\;-\infty
\quad\text{as } w=w^{\ast}(t)\to\infty.
\]
It follows that
\[
\lim_{t\to\infty}\frac{\zeta_i^{(0)}+\gamma_i \zeta_i^{(*)}\,\tilde z_i^{(2)}(t)}{\gamma_i \sigma_i \,\tilde z_i^{(1)}(t)}
=\pm\infty,
\]
i.e., the ratio diverges in magnitude. 
\end{proof}

Thus, under endogenous labor supply, one can implement a monotone (policy-induced) wage path that produces the same long-memory divergence as in the fixed-supply case presented in Proposition \ref{prop:demography_long_mem}.
}


\newpage

\RRew{
\section{The Impact of Returns-to-Scale Uncertainty on Equilibrium Wage}\label{Section:appendix_E}

To highlight the importance of firms' uncertainty regarding their returns to scale and the role of learning mechanisms, this appendix provides a simple numerical illustration of how such uncertainty affects the equilibrium wage. We focus on the case \(n=1\), with \(\zeta_{i}^{(\ast)} = \alpha = 0.5\), \(m_i = 0\), and \(\sigma_1 = 1\). Figure~\ref{fig:Example1} plots the dynamics of the relative wage-price ratio \(w^*(t)/p_1^*(t)\) over the time horizon \(\{1, \ldots, 100\}\), for three different values of \(\zeta_i(t)\). 

\begin{figure}[H]
    \centering
    \begin{subfigure}[b]{0.33\textwidth}     
        \includegraphics[scale=0.6]{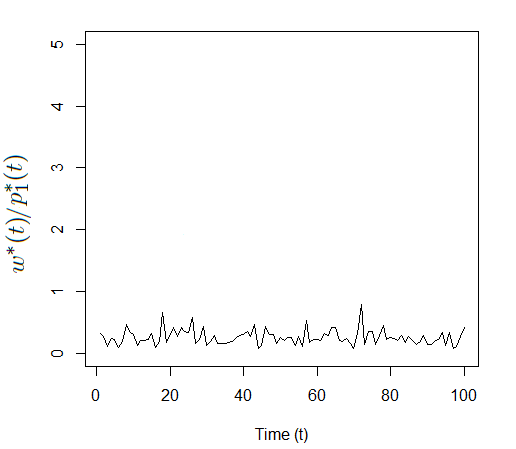}
        \caption{\footnotesize Setting \(v_1(t) = 0.1\).}
    \end{subfigure}~~
    \begin{subfigure}[b]{0.33\textwidth}
        \includegraphics[scale=0.6]{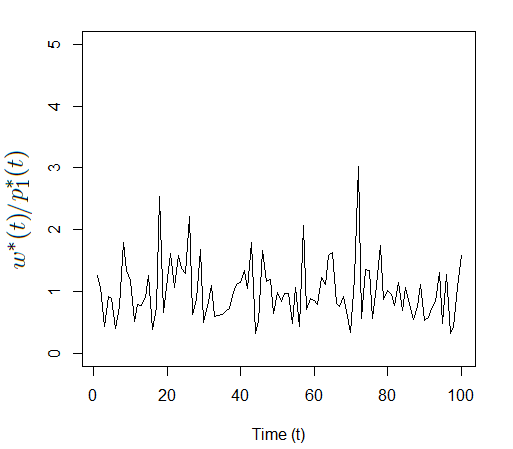}
        \caption{\footnotesize Setting \(v_1(t) = 0.5\).}
    \end{subfigure}~~
    \begin{subfigure}[b]{0.33\textwidth}
        \includegraphics[scale=0.6]{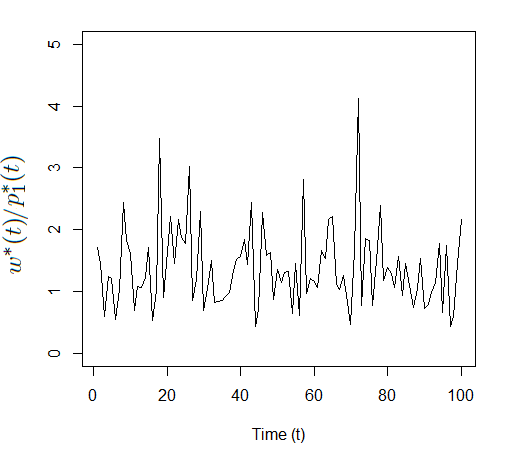}
        \caption{\footnotesize Setting \(v_1(t) = 0.9\).}
    \end{subfigure} 
    \caption{\footnotesize Dynamics of the equilibrium wage-price ratio, \(w^*(t)/p_1^*(t)\).  \label{fig:Example1}}
\end{figure}

The figure shows that firms' beliefs about their returns to scale exert a substantial influence on the dynamics of the equilibrium wage-price ratio. This observation motivates the analysis of learning paths, which investigates how firms update their beliefs and how this process shapes equilibrium outcomes.

}

\end{document}